\documentclass[aps,pra,superscriptaddress,11pt,twoside]{revtex4-2}

\usepackage{amsmath,amssymb,amsthm,amsfonts,amscd,mathrsfs,latexsym,enumerate}

\usepackage{enumitem}
\usepackage{mathtools}
\usepackage{hyperref}
\usepackage{color}
\usepackage{bm}
\usepackage{standalone}
\usepackage{comment}
\usepackage{listings} % Pacchetti per stampare il codice 
\usepackage{orcidlink}
\usepackage{tikz-cd}
\usepackage{upgreek}
\usepackage{bbm}

\newtheorem{theorem}{Theorem}
\numberwithin{theorem}{section}
\newtheorem{corollary}[theorem]{Corollary}
\newtheorem{lemma}[theorem]{Lemma}
\newtheorem{proposition}[theorem]{Proposition}

\theoremstyle{definition}
\newtheorem{definition}[theorem]{Definition}
\newtheorem{example}[theorem]{Example}
\newtheorem{remark}[theorem]{Remark}
\newtheorem{fact}[theorem]{Fact}

\def\squareforqed{\hbox{\rlap{$\sqcap$}$\sqcup$}}
\def\qed{\ifmmode\squareforqed\else{\unskip\nobreak\hfil
\penalty50\hskip1em\null\nobreak\hfil\squareforqed
\parfillskip=0pt\finalhyphendemerits=0\endgraf}\fi}
\def\endenv{\ifmmode\;\else{\unskip\nobreak\hfil
\penalty50\hskip1em\null\nobreak\hfil\;
\parfillskip=0pt\finalhyphendemerits=0\endgraf}\fi}

\newcommand{\beq}{\begin{equation}}
\newcommand{\eeq}{\end{equation}}
\newcommand{\SO}{\mathrm{SO}}
\newcommand{\so}{\mathfrak{so}}
\newcommand\SU{\mathrm{SU}}
\newcommand\SL{\mathrm{SL}}
\newcommand\GL{\mathrm{GL}}
\newcommand\U{\mathrm{U}}
\newcommand{\midd}{\textup{ s.t. }}
\newcommand{\diag}{\operatorname{diag}}
\DeclareMathOperator\ord{ord}

\newcommand{\cH}{\mathcal{H}}
\newcommand{\ZZ}{\mathbb{Z}}
\newcommand{\NN}{\mathbb{N}}
\newcommand{\QQ}{\mathbb{Q}}
\newcommand{\CC}{\mathbb{C}}
\newcommand{\RR}{\mathbb{R}}
\renewcommand\Re{\operatorname{Re}}
\newcommand{\mM}{\mathsf{M}}
\newcommand{\mI}{\mathrm{I}}

\newcommand{\ox}{\otimes}
\newcommand{\ket}[1]{\vert #1\rangle}
\newcommand{\bra}[1]{\langle #1\vert}

\newcommand{\braket}[2]{\langle#1|#2\rangle}
\newcommand{\proj}[1]{\ket{#1}\!\bra{#1}}

\DeclareMathOperator\Tr{Tr}
\DeclareMathOperator\Hom{Hom}
\DeclareMathOperator{\Span}{Span}
\DeclareMathOperator\cProj{Pr}
\DeclareMathAlphabet{\mathdutchcal}{U}{dutchcal}{m}{n}
\newcommand{\dude}{\mathdutchcal{d}}
\newcommand{\norm}[1]{\left\lVert#1\right\rVert}
\newcommand{\abs}[1]{\left\lvert#1\right\rvert}

\allowdisplaybreaks 

\usepackage{etoolbox}
\makeatletter
\patchcmd{\frontmatter@RRAP@format}{(}{}{}{}
\patchcmd{\frontmatter@RRAP@format}{)}{}{}{}
\renewcommand\Dated@name{}
\makeatother

\begin{document}

\title{Composing \texorpdfstring{$\mathbf{p}$}{Lg}-adic qubits: from representations of \texorpdfstring{$\mathbf{SO(3)_p}$}{Lg}\protect\\ to entanglement and universal quantum logic gates}

\author{Ilaria Svampa\,\orcidlink{0000-0002-1389-0319}}
\email{ilaria.svampa@uni-koeln.de}
\affiliation{Department Mathematik/Informatik--Abteilung Informatik, Universit\"at zu K\"oln, Albertus-Magnus-Platz, 50923 K\"oln, Germany}

\author{Sonia L'Innocente\,\orcidlink{0000-0002-9224-7451}}
\email{sonia.linnocente@unicam.it}
\affiliation{School of Science and Technology, Universit\`a di Camerino, Via Madonna delle Carceri 9, I-62032 Camerino, Italy}

\author{Stefano Mancini\,\orcidlink{0000-0002-3797-3987}}
\email{stefano.mancini@unicam.it}
\affiliation{School of Science and Technology, Universit\`a di Camerino, Via Madonna delle Carceri 9, I-62032 Camerino, Italy}
\affiliation{Istituto Nazionale di Fisica Nucleare, Sezione di Perugia,\\ via A.~Pascoli, I-06123 Perugia, Italy}              

\author{Andreas Winter\,\orcidlink{0000-0001-6344-4870}}
\email{andreas.winter@uni-koeln.de}
\affiliation{Department Mathematik/Informatik--Abteilung Informatik, Universit\"at zu K\"oln, Albertus-Magnus-Platz, 50923 K\"oln, Germany}
\affiliation{ICREA {\&} Grup d'Informaci\'o Qu\`antica, Departament de F\'isica, Universitat Aut\`onoma de Barcelona, 08193 Bellaterra (BCN), Spain}
\affiliation{Institute for Advanced Study, Technische Universit\"at M\"unchen,\\  Lichtenbergstra{\ss}e 2a, D-85748 Garching, Germany\vspace{2mm}}

%\date{(20 January 2026)}

%%%%%%%%%%%%%%%%%%%%%%%%%%%%%%%%%

\begin{abstract}
In the context of $p$-adic quantum mechanics, we investigate composite systems of $p$-adic qubits and $p$-adically controlled quantum logic gates. We build on the notion of a single $p$-adic qubit as a two-dimensional irreducible representation of the compact $p$-adic special orthogonal group $\SO(3)_p$. We show that the classification of these representations reduces to the finite case, as they all factorise through some finite quotient $\SO(3)_p\bmod p^k$. Then, we tackle the problem of $p$-adic qubit composition and entanglement,
% for systems of two $p$-adic qubits
fundamental for a $p$-adic formulation of quantum information processing. We classify the representations of $\SO(3)_p\bmod p$,
%: interestingly, there are several $p$-adic qubit representations for $p>3$, and only $\SO(3)_3\bmod 3$ has $4$-dimensional irreducibles. In this work
and analyse tensor products of two $p$-adic qubit representations lifted from $\SO(3)_p\bmod p$. We solve the Clebsch-Gordan problem for such systems, revealing that the coupled bases decompose into singlet and doublet states. We further study entanglement arising from those stable subsystems. %We further demonstrate that these irreducible subsystems are spanned by maximally entangled Bell states. Moreover, we propose a circuit model where logic gates are driven by the actions of $\SO(3)_p$. %For $p=3$ (only case where $\SO(3)_p\bmod p$ has $4$-dimensional irreducible representations) 
For $p=3$, we construct a set of gates from $4$-dimensional irreducible representations of $\SO(3)_p\bmod p$ %(only existing for $p=3$) 
that we prove to be universal for quantum computation.
%(encoded on real numbers)
\end{abstract}

\maketitle

\textbf{Mathematics Subject Classification: }{11S31%Class field theory; p-adic formal groups
; 20C20%Modular representations and characters
; 20C25%Projective representations and multipliers
; 81P40%Quantum coherence, entanglement, quantum correlations
; 81P65%Quantum gates
}

%\tableofcontents

%%%%%%%%%%%%%%%%%%%%%%%%%%%%%%%%%%%%%%%%%%%%%%%%%%%%%%%%%%%%%%%%%%%%%%%%

\section{Introduction}
\label{sec:intro}
Applying $p$-adic numbers to physics was an idea put forward in 1968 by two pure mathematicians, F. van der Blij and A. Monna~\cite{BlijMonna}. Since the 1980s, $p$-adic numbers have been increasingly used in quantum physics. $p$-Adic strings and gravity were among the first models of $p$-adic quantum physics (see, for example,~\cite{FreundWitten,VolSeminal,ArafevaDragovich}). Physicists' interest in $p$-adic numbers stems from attempts to construct new models of space-time capable of describing phenomena at the (extremely small) Planck scale, where the Archimedean geometry is no longer valid. 

The pioneering studies on $p$-adic string theory stimulated further research in $p$-adic quantum mechanics and field theory (see the books~\cite{VVZ,Khrennikov}). These investigations, in turn, fostered the development of $p$-adic mathematics in many directions: theory of distributions~\cite{Albeverio}, differential and pseudo-differential equations~\cite{Krennpall,VVZpaper}, probability theory~\cite{Khrennikbastaa}, and spectral theory of operators in a $p$-adic analogue of Hilbert space~\cite{Albeveriobis,Albeveriotris,Khrennikbis}.

The representation of $p$-adic numbers as sequences of digits makes it possible to use this number system for information coding. Consequently, $p$-adic models can serve as a framework for describing various information processes, particularly those based on $p$-adic dynamical systems~\cite{AlbeverioMore,AlbeverioPal}. There are several areas of information processing where $p$-adic dynamics have proven effective, like computer science (straight-line programs), numerical analysis and simulations (pseudorandom numbers), uniform distribution of sequences, cryptography (stream ciphers, $T$-functions), combinatorics (Latin squares), automata theory and formal languages, and genetics (a comprehensive survey of these applications can be found in Ref.~\cite{Anashin}). Further studies in computer science and cryptography were driven by the observation that fundamental computer instructions (and thus programs composed of them) can be interpreted as continuous transformations with respect to the $2$-adic metric~\cite{Anashinbis,Anashintris}. Additionally, $p$-adic numbers provide a natural framework for describing a broad class of neural networks with hierarchical structures~\cite{AlbeverioKhTir,ZunigaGal}.

Given the nowadays deep intertwining of quantum theory with computation and information theory, it is particularly intriguing to explore whether fundamental entities of these disciplines --- such as the qubit --- can be formalized within the $p$-adic framework.

It is important to remark that $p$-adic numbers can appear in two distinct conceptual roles. Consider a scalar field $\phi \colon X \to Y$, where $X$ represents space-time. In some models $X$ is taken to be $p$-adic, while $Y$ remains real or complex. In other approaches, both $X$ and $Y$ are assumed to be non-Archimedean. This distinction is well-known in the context of the Langlands Program, which itself divides into two separate arenas depending on whether one's preferred $\zeta$ (or $L$-) functions take values in $\CC$ or in $\QQ_p$~\cite{Bump}.

Adhering to the first line of reasoning, in Ref.~\cite{our2nd}, the notion of a qubit was introduced, opening new perspectives for the intersection of $p$-adic mathematics, quantum theory, and information science. Actually, the $p$-adic qubit arises from the two-dimensional unitary irreducible representations of the $p$-adic special orthogonal group $\SO(3)_p$ in a three-dimensional $p$-adic space, of which constructed examples for all primes $p$.

Going beyond the description of a single $p$-adic qubit, it is natural to address the problem of combining multiple such qubits, beginning with the case of two. This, in turn, requires establishing the mathematical foundations for the $p$-adic composition of angular momentum and spin (see also~\cite{Crespo} for an approach based on symplectic geometry).
The progression of this framework involves four key stages: \emph{representations}, \emph{Clebsch-Gordan problem}, \emph{entanglement}, and \emph{logic gates}. 

The structure of the paper reflects this progression. Section~\ref{sec:introrappels} reviews the characterization of the group $\SO(3)_p$ as an inverse limit of finite groups $\SO(3)_p\bmod p^k,\,k\in\NN$, along with the necessary background on projective unitary representations. In Section~\ref{sec:irrepsSO3pDp}, we begin presenting our original results. We show that every finite-dimensional projective unitary representation of $\SO(3)_p$ factorise through $\SO(3)_p\bmod p^k$, for some $k\in\NN$ (Subsec.~\ref{fact:repnspk}). Moreover, we prove that $\SO(3)_p\bmod p$ is (isomorphic to) the semidirect product of $\ZZ/p\ZZ\times \ZZ/p\ZZ$ and the dihedral group $\mathrm{D}_{p+1}$ (Subsec.~\ref{subsec:strirrpsGp}). This allows us to classify the unitary irreducible representations of $\SO(3)_p\bmod p$ by the method of little groups by Wigner and Mackey~\cite{Mackey}. In particular, we identify $p$-adic qubit representations for every prime $p$~\cite{our2nd}. Then, this work primarily investigates systems of two $p$-adic qubits, via the tensor products of the found representations. Interestingly, for primes $p > 3$, there exist multiple non-equivalent $p$-adic qubits, resulting in a broader landscape of multi-qubit systems in the $p$-adic framework. Section~\ref{CBdecocoedf} addresses the Clebsch-Gordan decomposition of two $p$-adic qubits into a direct sum of irreducible subrepresentations. This boils down to the Clebsch-Gordan problem of the dihedral group $\mathrm{D}_{p+1}$. Since the latter has only one- and two-dimensional irreducible representations, the coupled bases are given by singlet and/or doublet states, in contrast to a singlet and a triplet in standard quantum mechanics. We also find the Clebsch-Gordan coefficients providing the change of basis which realises a Clebsch-Gordan decomposition, explicitly for the smallest primes, $p=2,3,5,7$ (Subsec.~\ref{subsec:CBcpeff2357}). For the general treatment, for every odd prime $p$, we examine the structural nature of the coupled bases realising the various Clebsch-Gordan decompositions. This is done in Section~\ref{sec:entangldectp}, where we study maximally entangled vector states spanning the stable subsystems. We will underline that, except for the singlets, the projectors onto doublets (and triplets) are separable quantum states. Lastly, in Section~\ref{3pUnigates}, we propose a circuit model of quantum computation where logic gates are controlled by the actions of $\SO(3)_p$. Our programme is to construct $p$-adically controlled quantum logic gates on the single-, two- and $n$-qubit levels, using elements from the same $2^n$-dimensional unitary representations of $\SO(3)_p$. For real orthogonal representations%(as a subclass of complex unitary ones)
, as those lifted from $\SO(3)_p\mod p$, an extra qubit is added to encode unitary circuits through orthogonal gates according to the recipe by Bernstein and Vazirani~\cite{BernVazi}. In this setting, we find a set of gates from $4$-dimensional representations of $\SO(3)_3$ which is universal for quantum computation. The choice $p=3$ is not our restriction; rather, it is motivated by the absence of $4$-dimensional irreducible representations of $\SO(3)_p\bmod p$ for every other prime. This set of gates is obtained with the aid of the softwares GAP (Groups, Algorithms, and Programming) and Wolfram Mathematica, while its universality is established by using Lie-algebraic tools. The existence of such a universal set of $p$-adically controlled gates is of fundamental importance, as it guarantees the computational completeness of the proposed model and opens new avenues for the study of quantum circuits.

%%%%%%%%%%%%%%%%%%%%%%%%%%%%%%%%%%%%%%%%%%%%%%%%%%%%%%%%%%%%%%%%%%%%%%%%
\section{Rapells}
\label{sec:introrappels}
%In this introductory section, we present the objects of our investigation, i.e. the group $\SO(3)_p$ of rotations on $\QQ_p^3$, some of its representations and $p$-adic qubits, for every prime $p$. We take this opportunity to introduce the fundamental concepts from literature.
In this introductory section, we outline the objects of our investigation, together with some of their representations and the notion of $p$-adic qubits for each prime $p$. We also take this opportunity to introduce the fundamental concepts drawn from the existing literature.

\subsection{The profinite group \texorpdfstring{$\SO(3)_p$}{Lg}}
Let $p$ be a prime number greater than or equal to $2$. Let $\ZZ_p$ and $\QQ_p$ denote respectively the ring of $p$-adic integers and the field of $p$-adic numbers~\cite{folland2016course,Gouvea,Serre, Cassels}. Topologically, they are the metric completions of $\ZZ$ and $\QQ$ respectively, with respect to the $p$-adic absolute value $\abs{\,\cdot\,}_p$. %Moreover, $\ZZ_p$ is compact, $\QQ_p$ is only locally compact, and they both are totally disconnected. 
Let $\mathbb{U}_p$ denote the multiplicative group of $p$-adic units.
%Every $x$ in $\QQ_p^\times$ is written uniquely as $x=\sum_{k\geq-k_0}c_kp^k$ with $k_0\in\ZZ$ and $c_k\in\{0,\dots,p-1\}$ for $k\geq-k_0$.  If $\mathbb{U}_p$ denotes the multiplicative group of $p$-adic units, every $x\in\QQ_p^\times$ writes uniquely as $x=p^{-k_0}u$ for some $k_0\in\ZZ$ and $u\in\mathbb{U}_p$. The $p$-adic absolute value on $\QQ_p$ is given by \beq \abs{x}_p\coloneqq\begin{cases} p^{k_0}, & \textup{ if } x=p^{-k_0}u=\sum_{k\geq-k_0}c_kp^k\textup{ with } c_{-k_0}\neq0,\\ 0, & \textup{ if } x=0. \end{cases} \eeq For $n\in\NN$, the $p$-adic norm on $\QQ_p^n$ is $\norm{\mathbf{x}}_p\coloneqq \max_{i=1}^n\abs{x_i}_p$ for $\textbf{x}=(x_i)_{i=1}^n\in \QQ_p^n$. The following structure of the group $\QQ_p^\times/(\QQ_p^\times)^2$ of non-zero $p$-adic numbers modulo squares is well-known~\cite{Serre, Cassels,Gouvea}. If $p>2$, then $\QQ_p^\times/(\QQ_p^\times)^2 \simeq \langle u, p\rangle = \{1,u,p,up\}\simeq (\ZZ/2\ZZ)^2$ for a non-square $p$-adic unit $u\in\mathbb{U}_p$; when $p=2$, then $\QQ_2^\times/(\QQ_2^\times)^2\simeq \langle -1,2,5\rangle =\{\pm1,\pm2,\pm5,\pm10\}\simeq (\ZZ/2\ZZ)^3$. 

\begin{proposition}[\cite{our1st}, Sec. 2]\label{quadform}
For every prime $p$, up to equivalence, there is a unique non-isotropic quadratic form over $\QQ_p^3$:
\beq \label{eq:quadrform3}
Q_+(\mathbf{x})\coloneqq \begin{cases}
x_1^2-vx_2^2+px_3^2,&p\textup{ odd},\\
x_1^2+x_2^2+x_3^2,&p=2,
\end{cases}
\eeq
where $v\coloneqq
\begin{cases}
 -1, &\textup{if } p\equiv3 \mod4,\\
 -u, &\textup{if } p\equiv1 \mod4,
\end{cases}$\ \,  for a non-square $u\in\mathbb{U}_p$. Then, for every prime $p$, up to isomorphism, there is a unique compact special orthogonal group on $\QQ_p^3$:
\beq 
\SO(3)_p\coloneqq \left\{L\in \mM(3,\QQ_p)\midd L^\top A_+L=A_+\textup{ and } \det(L)=1\right\},\label{SO3p}
\eeq 
where $A_+\coloneqq \begin{cases}
\diag(1,-v,p),&p\textup{ odd},\\
\mI_3,&p=2,
\end{cases}$\ \, is the matrix representation of $Q_+$ with respect to the canonical basis.
\end{proposition}
Compact $p$-adic special orthogonal groups exist only up to dimension $4$, in contrast to the analogous groups over $\RR$. Also, there is more than one non-isotropic quadratic form (or scalar product) on $\QQ_p^2$.

The group $\SO(3)_p$ is a topological group, once supplied with the $p$-adic norm $\norm{L}_p=\norm{(\ell_{ij})_{ij}}_p\coloneqq \max_{i,j=1}^3\abs{\ell_{ij}}_p$, and hence with \emph{$p$-adic ultrametric topology}. It is homeomorphic to a subspace of $\mM(3,\QQ_p)\simeq\QQ_p^9$, from which it inherits the properties of being a \emph{totally disconnected}, locally compact, \emph{second countable} and \emph{Hausdorff} group. The group $\SO(3)_p$ is indeed \emph{compact}: it is closed %homeomorphic to a closed subset of $\mM(3,\QQ_p)$, as it is a group of matrices whose entries are solutions of a system of continuous (polynomial) equations (given by the special orthogonal conditions)
and bounded, as its matrices have entries in $\ZZ_p=\{x\in\QQ_p\midd \abs{x}_p\leq 1\}$~\cite[Thm.~5]{our1st}:
\beq \label{eq:inclsSO4SO3SO2SLint}
\SO(3)_p\subset \SL(3,\ZZ_p).
\eeq 
%\begin{theorem}[\cite{PhDtesi}, Theorem 2.19, and]\label{thm:compactness} For every prime $p>2$, $\dude \in \{-v,p,up\}$, and for $p=2$, $\dude\in\{1,\pm2,5,\pm10\}$, \begin{align}  & \SO(2)_{p,\dude}\subset \SL(2,\ZZ_p)\quad \, \textup{while}\quad  \SO(2)_{2,-5}\subset\SL(2,2^{-1}\ZZ_2); \label{eq:inclsSO4SO3SO2SLintA}\\ &\SO(3)_p\subset \SL(3,\ZZ_p);\label{eq:inclsSO4SO3SO2SLint}\\ & \SO(4)_p\subset\SL(4,\ZZ_p) \textup{ for }p>2; \quad  \SO(4)_2\subset \SL(4,2^{-1}\ZZ_2).\label{eq:inclsSO4SO3SO2SLintB} \end{align} \end{theorem}
The elements of $\SO(3)_p$ act as rotations on $\QQ_p^3$~\cite[Thm.~6]{our1st}.

\medskip
We now briefly describe topological groups, in particular totally disconnected and compact groups, to further characterise $\SO(3)_p$.

\begin{remark}
\label{rem:clopen}
A topological space $X$ is \emph{connected} if and only if the only \emph{clopen subsets} of $X$ are $\emptyset$ and $X$. For $n\in\NN$, the space $\RR^n$ is connected, while $\QQ_p^n$ and $\ZZ_p^n$ are totally disconnected and admit proper (compact) clopen subsets, as well as the $p$-adic orthogonal groups. 

An \emph{open subgroup} $H$ of a topological group $G$ is \emph{closed}. %This is easily shown, by using that cosets of an open subgroup are open (since translations are homeomorphisms in a topological group), and $G\setminus H$ is a union of cosets of $H$. 
As a consequence, a connected group has no proper open subgroups%perchè sarebbe clopen
; e.g. the additive group $\RR$%and no open compact subgroups.. Vedi Corollary p.23 sec3.4 di Robert
. On the other hand, the additive group $\QQ_p$ has a infinitely many proper open compact subgroups: $p^m\ZZ_p=B_{\leq p^{-m}}(0)$ for $m\in\ZZ$. %Vedi Prop in sec 3.5 di Robert
Actually, they are all and the only proper closed subgroups of $\QQ_p$~\cite{RobertWOW}. A base for the topology of $p^m\ZZ_p$ is $\{B_{<p^k}(x_0)\midd x_0\in p^m\ZZ,\, k\in J\}$, since $p^m\ZZ$ is a countable dense subset of $p^m\ZZ_p$, and where $J$ is a subset of $\ZZ_{\leq -m+1}$ without a minimum. 

If $G$ is a compact topological group, a subgroup $H$ is open if and only if $H$ is closed of finite index in $G$~\cite[Lemma~2.1.2]{profinite}. For example, for the compact $\ZZ_p$, the ideals $p^m\ZZ_p$, $m\in\NN$ are clopen of finite index; in fact, they are normal subgroups and $\ZZ_p/p^m\ZZ_p\simeq \ZZ/p^m\ZZ$ is of order $p^m$. 
\end{remark}

Let $(I,\leq)$ be a \emph{(right-)directed partially ordered set}. This is a non-empty set $I$ supplied with a partial order (i.e. a reflexive, transitive and antisymmetric binary relation) $\leq$, such that any finite subset of $I$ has upper bounds in $I$. We first recall the definition of inverse family and inverse limit of topological groups~\cite{BourTop} (see also~\cite{profinite} for a more categorical approach).
\begin{definition}\label{def:invlimsetgroup}
Let $\{G_i\}_{i\in I}$ be a family of topological groups, and $\{f_{ij}\colon X_j\rightarrow G_i\}_{i\leq j,\ i,j\in I}$ a family of continuous group homomorphisms such that
\begin{enumerate}
    \item $f_{ii}$ is the identity map on $G_i$, for every $i \in I$;
    \item $f_{ik}=f_{ij}\circ f_{jk}$, for every $i\leq j\leq k$,\ \ $i,j,k\in I$.
\end{enumerate}
We call $\big\{\{G_i\}_{i\in I}, \{f_{ij}\colon X_j\rightarrow G_i\}_{i\leq j,\ i,j\in I}\big\}\equiv \{G_i,f_{ij}\}_I$ an \emph{inverse family of topological groups}. 

Let now $\prod\limits_{i\in I}G_i$ be the Cartesian product of the family of groups $\{G_i\}_{i\in I}$, and $\cProj_i\colon\prod\limits_{i\in I}G_i\rightarrow G_i,\ g=(g_i)_{i\in I}\mapsto g_i$ the projection on the $i$-th component. 
The \emph{inverse} (or \emph{projective}) \emph{limit} of the inverse family of topological groups $\{G_i,f_{ij}\}_I$ is the subgroup 
\begin{align}
    &\varprojlim\{G_i,f_{ij}\}_I\coloneqq\left\{x\in \prod_{i\in I}G_i\midd \cProj_i(x)=f_{ij}\circ \cProj_j(x),\textup{ for every }i\leq j,\ i,j\in I\right\}\notag\\
    &= \left\{(g_i)_{i\in I}\in \prod_{i\in I}G_i\midd g_i=f_{ij}(x_j),\textup{ for every }i\leq j,\ i,j\in I\right\}\leq \prod_{i\in I}G_i,\label{invlimset}
\end{align}
endowed with the coarsest topology for which all $\cProj_i$ are continuous ($i\in I$), coinciding with the topology induced by the product topology of $\prod\limits_{i\in I}G_i$.

The restriction $f_i$ of $\cProj_i$ to $\varprojlim\{G_i,f_{ij}\}_I$ is called the {\it canonical projection} of $\varprojlim\{G_i,f_{ij}\}_I$ into $G_i$, and $f_i=f_{ij}\circ f_j$ for $i\leq j$.
\end{definition}

%The inverse limit of an inverse family of topological groups always exists (this is not true in the broader setting of an inverse family in an arbitrary category). In any category, the definition of inverse limit is given by means of a universal property, so that if an inverse limit exists, it is necessarily unique: if $X$ and $X'$ are two inverse limits of the same inverse family, with canonical projection maps $\{f_i\}_i$ and $\{f'_i\}_i$ respectively, then there exists a \emph{unique} isomorphism $f\colon X\rightarrow X'$ such that $f'_i\circ f = f_i$ for every $i\in I$. 

\begin{example}
\label{ex:p-adciinvlim}
One has~\cite[p.~11]{Serre}, 
\beq\label{eq:invlimZpppp}
\ZZ_p\simeq \varprojlim\left\{\ZZ_p/p^k\ZZ_p,\,\phi_{kl}\right\}_{\NN}\simeq \varprojlim \left\{\ZZ/p^k\ZZ,\Phi_{kl}\right\}_\NN
\eeq
where $\phi_{kl}\colon \ZZ_p/p^l\ZZ_p\rightarrow \ZZ_p/p^k\ZZ_p$, $\phi_{kl}(x + p^l\ZZ_p)\coloneqq x+p^k\ZZ_p$ and $\Phi_{kl}\colon \ZZ/p^l\ZZ\rightarrow \ZZ/p^k\ZZ$, $\Phi_{kl}\left(x_l\right) \coloneqq x_l\mod p^k$, for every $l\geq k>0$. We denote by $\Phi_k$ the canonical projection $\Phi_k\colon \ZZ_p\rightarrow \ZZ/p^k\ZZ$, $x\mapsto x_k\mod p^k$. 

Similarly~\cite[Exercise~E1.16]{Hofmann},
\beq 
\QQ_p\simeq \varprojlim\left\{\QQ_p/p^k\QQ_p,\,\phi_{kl}'\right\}_{\NN}\simeq \varprojlim\left\{\frac{1}{p^\infty}\ZZ/p^k\ZZ,\Phi_{kl}'\right\}_\NN,
\eeq 
where $\frac{1}{p^\infty}\ZZ\coloneqq \bigcup_{n\geq1}\frac{1}{p^n}\ZZ\subseteq\QQ$, $\phi_{kl}'\colon \QQ_p/p^l\QQ_p\rightarrow \QQ_p/p^k\QQ_p$, $\phi'_{kl}(x + p^l\QQ_p)\coloneqq x+p^k\QQ_p$ and $\Phi'_{kl}\colon \frac{1}{p^\infty}\ZZ/p^l\ZZ\rightarrow \frac{1}{p^\infty}\ZZ/p^k\ZZ$, $\Phi_{kl}'\left(x_l\right) \coloneqq x_l\mod p^k$, for every $l\geq k>0$. 

Note that here, since the index set is the totally ordered set $\NN$, it is sufficient to consider $l=k+1$, $k\in\NN$.
\end{example}

\begin{definition}\label{def.2}
A topological group $G$ is said to be \emph{profinite} if it is (isomorphic to) the inverse limit of an inverse family of finite groups, each given the discrete topology.
\end{definition}

\begin{theorem}[{\cite[Thm.~1.1.12]{profinite}}]
\label{prop:profcompdisc}
A topological group $G$ is profinite if and only if it is compact (Hausdorff) and totally disconnected.
\end{theorem}

The additive group $\ZZ_p$ is profinite, but $\QQ_p$ is not since it is not compact.

\begin{lemma}[{\cite[Lemma~2.1.1]{profinite}}]
\label{lemma:baseeeeker}
Let $G=\varprojlim\{G_i\}_I$ be a profinite group, with canonical projection $f_i\colon G\rightarrow G_i$, $i\in I$. Then \vspace{-0.25cm}
\beq 
\vspace{-0.2cm} \Big\{\ker(f_i)=G\cap\big[{e}_{G_i}\times \big(\prod_{j\neq i}G_j\big)\big]\Big\}_{i\in I}
\eeq 
is a fundamental system of open neighbourhoods of $e$ in $G$.
\end{lemma}
%Each $\ker(f_i)=f_i^{-1}(e_{G_i})$ is an open normal subgroup, since the singleton $\{e_{G_i}\}$ is open in the discrete topology of $G_i$. By Remark~\ref{rem:clopen}, $\ker(f_i)$ is a clopen normal subgroup and $G/\ker(f_i)$ has finite order.
Each $\ker(f_i)$ is a clopen normal subgroup and $G/\ker(f_i)$ has finite order (cf.~Remark~\ref{rem:clopen}).
\begin{comment}
\begin{theorem}[\cite{profinite}, Theorem 2.1.3]\label{prop:profequiv}
The following conditions on a topological group $G$ are equivalent.
\begin{enumerate}
\item $G$ is a profinite group;
\item $G$ is compact, Hausdorff, totally disconnected, and for each open normal subgroup $N$ of $G$, $G/N$ is finite;
\item $G$ is compact, and $e\in G$ admits a fundamental system $\mathcal{N}_e$ of neighbourhoods $N$ of $e$ such that $\bigcap_{N\in \mathcal{N}_e}N=\{e\}$ and each $N$ is an open normal subgroup of $G$ with $G/N$ finite;%Montgomery-Zipping p.56;
\item $e\in G$ admits a fundamental system $\mathcal{N}_e$ of neighbourhoods $N$ of $e$ such that each $N$ is an open normal subgroup of $G$ with $G/N$ finite, and
\beq 
G\simeq \varprojlim_{N\in \mathcal{N}_e}G/N.
\eeq 
%Corollary 1 p.118 Cassel
\end{enumerate}
\end{theorem}
In point 3., if $G=\varprojlim\{G_i,f_i\}_{i\in I}$, then $\{\ker(f_i)\}_{i\in I}$ provides an example of such a fundamental system of neighbourhoods of $e$ in $G$. E.g., $\left\{p^m\mathbb{Z}_p\mid m\in\mathbb{N}\cup\{0\}\right\}$ is a fundamental system of neighbourhoods of $1$ in $\ZZ_p$.
\end{comment}
\begin{theorem}[{\cite[Thm.~2.1.3]{profinite}}]
\label{prop:profequiv}
A topological group $G$ is profinite if and only if $e\in G$ admits a fundamental system $\mathcal{N}_e$ of neighbourhoods $N$ of $e$ such that each $N$ is an open normal subgroup of $G$ with $G/N$ finite and \vspace{-0.3cm}
\beq 
G\simeq \varprojlim_{N\in \mathcal{N}_e}G/N.
\eeq 
%Corollary 1 p.118 Cassel
\end{theorem}

According to Theorem~\ref{prop:profcompdisc}, the group $\SO(3)_p$ is profinite, i.e., it is an inverse limit of an inverse family of finite discrete groups. 
Recall Eq.~\eqref{eq:inclsSO4SO3SO2SLint}, and that the elements of $\ZZ_p$ can be reduced $\bmod \ p^k$, $k\in\NN$, via the canonical projection $\Phi_k$ as in Example~\ref{ex:p-adciinvlim}. In $\mM(3,\ZZ_p)$ the matrix product is defined through sums and products of entries, for which $\Phi_k$ are homomorphisms. Therefore, the map
\begin{align}\label{eq:grouphoprpi}
\pi_k\colon &\SO(3)_p\subset \mM(3,\ZZ_p) \rightarrow \pi_k\left(\SO(3)_p\right)\subset \mM(3,\ZZ/p^k\ZZ),\nonumber\\
&\pi_k(L)\coloneqq \begin{pmatrix} \Phi_k(\ell_{ij}) \end{pmatrix}_{ij}=\begin{pmatrix} \ell_{ij}\mod p^k \end{pmatrix}_{ij},
\end{align} 
for $L=(\ell_{ij})_{ij}\in\SO(3)_p$, is a group homomorphism. Note that $\pi_k\left(\SO(3)_p\right)$ is a finite group since $\abs{\mM(3,\ZZ/p^k\ZZ)}=p^{9k}$. Moreover, $\ker(\pi_k)=\left(\mI_3+p^k\mM(3,\ZZ_p)\right)\cap \SO(3)_p$ is a normal subgroup of $\SO(3)_p$ --- the so-called \emph{$k$-th congruence subgroup} --- which provides a fundamental system of neighbourhoods of $e$ in $\SO(3)_p$. Lastly, $\SO(3)_p/\ker(\pi_k)=\pi_k\left(\SO(3)_p\right)$ (refer to Theorem~\ref{prop:profequiv}). 
\begin{theorem}[{\cite[Thm.~3.1]{Haar2} and~\cite[Sec.~III]{our2nd}}]
\label{theor-invlimSO3p}
For every prime $p$, 
\beq 
\SO(3)_p\simeq\varprojlim \{G_{p^k},\pi_{kl}\}_\NN,
\eeq 
where, for every $k\leq l$, $k,l\in \NN$, $G_{p^k} \coloneqq \pi_k\left(\SO(3)_p\right)=\SO(3)_p\mod p^k$ is a finite group, $\pi_k\colon \SO(3)_p\rightarrow G_{p^k}$ is the surjective homomorphism in Eq.~\eqref{eq:grouphoprpi}, and 
\beq
\pi_{kl}:G_{p^l}\rightarrow G_{p^k},\quad (\ell_{ij}\mod p^l)_{ij}\mapsto (\ell_{ij}\mod p^k)_{ij}. \label{eq:mapinverslimit} 
\eeq
\end{theorem}
We recall from~\cite[Thm.~IV.5]{our2nd} that, for every odd prime $p$,
\beq \label{eq:paramGp}
G_p=\left\{L(a,b,c,d,s)\coloneqq\begin{pmatrix}
a & svb & 0 \\
b & s a & 0 \\
c & d & s
\end{pmatrix}\midd  a^2-vb^2\equiv 1,s\equiv \pm1,c,d\in \ZZ/p\ZZ\right\},
\eeq
where the solutions $(a,b)\in (\ZZ/p\ZZ)^2$ to $a^2-vb^2\equiv 1\mod p$ form a cyclic group of order $p+1$ (see~\cite[App.~A]{our2nd}). Then $\abs{G_p}=2p^2(p+1)$.

\subsection{Representation theory}
\label{subsec:prologo}
Here we introduce projective unitary representations, and motivate their fundamental role in quantum mechanics. Let $G$ be a topological group, and $\cH$ some non-zero Hilbert space over $\CC$. 

\begin{definition}[Cf.~\cite{folland2016course}]
A \emph{(linear)} \emph{unitary representation} of $G$ on $\cH$ is a homomorphism $\mathcal{U}\colon G\rightarrow \U(\cH)$ %cioè homomorphismo \mathcal{U}(gh)=\mathcal{U}(g)\mathcal{U}(h), e unitarietà \mathcal{U}(g)^{-1} = \mathcal{U}(g)^\dagger <-> <\mathcal{U}(g)v,\mathcal{U}(g)w> = <v,w>
that is continuous with respect to the topology of $G$ and the strong operator topology on $\U(\cH)$% (that is, $G\ni g\mapsto \mathcal{U}(g)\mathbf{v}\in\cH$ is continuous for every $\mathbf{v}\in\cH$)
.
\end{definition}
\begin{comment}
The strong operator topology is chosen so to make $\U(\cH)$ a second countable Polish group; other topologies %(such as the topology induced by the operator norm) 
would be too restrictive and a regular representation would not be continuous in general.
\end{comment}
The \emph{dimension} of a representation $\mathcal{U}\colon G\rightarrow \U(\cH)$ is $d(\mathcal{U})\coloneqq \dim(\cH)$. For $n\in\NN$, an $n$-dimensional unitary representations of $G$ is on a Hilbert space $\cH\simeq \CC^n$.
\begin{comment}
The unitary group $\U(n)\equiv \U(\CC^n)$ is a Lie group, with %manifold structure according to the 
standard subspace topology as $\U(n)\subset \mM(n,\CC)\simeq \CC^{n^2}\simeq\RR^{2n^2}$, induced by the Euclidean norm. This coincides with the strong operator topology in the finite-dimensional framework, where all norms are equivalent. The scalar product on $\CC^n$ is $\langle\mathbf{x},\mathbf{y}\rangle=\sum_{i=1}^n\overline{x_i}y_i$, for all $\mathbf{x}=(x_i)_{i=1}^n,\mathbf{y}=(y_i)_{i=1}^n\in\CC^N$.
\end{comment}
The unitary Lie group $\U(n)\equiv \U(\CC^n)$ is considered with standard Euclidean topology, which coincides with strong operator topology in the finite-dimensional framework, where all norms are equivalent. \emph{Subrepresentations} are defined in terms of stable subspaces with respect to the representation action of the group. A representation is said to be \emph{irreducible} if it has no non-trivial subrepresentations.  Hereafter, a linear unitary irreducible representation of a compact group will be referred to as \emph{irrep}, for short.

If $G$ is a locally compact abelian group, each of its irreps is one-dimensional, and can therefore be identified with its character. The irreps of $G$ close a group, $\widehat{G}\simeq\Hom(G,\U(1))$, called the \emph{dual group} of $G$, investigated within the context of Pontryagin duality~\cite{reiter2000,folland2016course}. 
\begin{comment}
According to Proposition~$4.3.3$ in~\cite{reiter2000}, the dual group of an inverse-limit group is isomorphic to the direct limit of the dual groups of the inverse family: 
\beq \label{eq:invdirlimhomo}
\widehat{\varprojlim G_k}\simeq \varinjlim\widehat{G_k}\simeq\bigcup_k\widehat{G_k}.
\eeq 
%\begin{proposition}[\cite{reiter2000}, Proposition~$4.3.3$]\label{prop:Nik} Let $\{G_i\}_\NN$ be an inverse family of (locally compact) abelian groups. Then, the dual group of the inverse limit group, $\widehat{\varprojlim G_i}$, contains a sequence of open subgroups (isomorphic to) $\widehat{G_i}$ such that \beq \label{eq:invlimreps} \widehat{\varprojlim G_i}\simeq \varinjlim\widehat{G_i}\simeq\bigcup_i\widehat{G_i}, \eeq where $\varinjlim\widehat{G_i}$ denotes the direct limit of the family $\widehat{G_i}$. \end{proposition} We do not describe direct limits here (which are categorically dual to inverse limits), as it is enough to regard $\varinjlim\widehat{G_i}$ as the union of the $\widehat{G_i}$s, each one embedded in the subsequent.
For example, this applies to the abelian groups of $p$-adic planar rotations $\SO(2)_{p,\dude}$.
\end{comment}

If $G$ is not abelian, the study of its representations becomes more challenging, and it is helpful to require that $G$ is compact, for which fundamental results of representation theory of finite groups extend% (e.g. Schur's Lemma for irreducible unitary representations)
. Every complex finite-dimensional representation of a compact group $G$ is equivalent to a unitary representation of $G$.
\begin{theorem}[{\cite[Thm.~5.2]{folland2016course} and~\cite[Thms.~3~{\&}~4,~pp.168-169]{BarutRaczka}}] 
\label{barurazcdiresumun}
If $G$ is a compact group, then every irreducible representation of $G$ is finite-dimensional, and every unitary representation of $G$ is completely reducible into a direct sum of %(finite-dimensional unitary) 
irreducible subrepresentations.
\end{theorem}
%The second part of this statement can be regarded as a generalisation of \emph{Maschke's Theorem}. %every repn of $G$ finito e' completely reducible

Let now $\widehat{G}$ denote the set of (unitary) equivalence classes of irreps of $G$. In contrast to the abelian case, here $\widehat{G}$ is not a group in general. The \emph{decomposition} of a unitary representation of $G$ into direct sum of irreps is generally not unique; however, the decomposition into subspaces corresponding to distinct equivalence classes (so-called isotypic subspaces) in $\widehat{G}$ is uniquely determined. Also, the \emph{multiplicity} of an irreducible subrepresentation (up to equivalence) of a unitary representation is the same for all the possible decompositions of the latter.

A milestone for the theory of representations of a compact group $G$ is the \emph{Peter-Weyl Theorem} (cf.~\cite[Thm.~5.12]{folland2016course}, \cite[Cor.~VII.10.2]{simon}, \cite[Thm.~1,~p.~172~{\&}~ Thm.~2,~p.~174]{BarutRaczka}), and relies on the concept of Haar measure $\mu$ on $G$. The Peter-Weyl Theorem establishes that the collection of all matrix coefficient functions taken from the set of irreps of $G$ (up to equivalences) forms a complete orthonormal basis for the space $\mathrm{L}^2(G,\mu)$ of complex-valued square-integrable functions on $G$. %This contains generalisations of classical results for finite groups, such us Schur's orthogonality relations and the completeness identity \beq \label{eq:complsumsquarirr} \sum_{\rho\in\widehat{G}}d(\rho)^2=\ord(G). \eeq 
Moreover, the Peter-Weyl Theorem states that the regular representation of $G$ on $\mathrm{L}^2(G,\mu)$ is equivalent to the direct sum of all the irreps in $\widehat{G}$, each occurring with multiplicity equal to its dimension.%(extension of the similar result about the regular representation of a finite group on its group algebra)

%\begin{remark}\label{rem:Peter-Weylmac} The Peter-Weyl theorem establishes a method to classify (up to equivalence) all the irreps of a compact group $G$: once the Haar measure $\mu$ on $G$ is known, one constructs the left or right regular representation $\rho_{\textup{reg}}$ of $G$ on $\mathrm{L}^2(G,\mu)$, and every irrep of $G$ occurs (and can be studied) as a subrepresentation of $\rho_{\textup{reg}}$. This method can be applied to our group $\SO(3)_p$, on which we already constructed the Haar measure~\cite{Haar1, Haar2}. However, this is a non-easy task, and we will introduce a more suitable approach to classify the irreps of $\SO(3)_p$, depending on its particular structure of inverse-limit group. \end{remark}
Besides the decomposition of the regular representation, there is another well-suited approach to classify the irreps of $\SO(3)_p$, depending on its particular structure of inverse-limit group. 

\medskip
Our focus on unitary representations is motivated by quantum-theoretical considerations: the evolution of a vector state in a Hilbert space $\cH$ 
is represented by a unitary operator on $\cH$.
%if the state of an isolated system is represented by a vector in a Hilbert space $\cH$, its state change is postulated to be described by a unitary operator on $\cH$, which preserve norms and probability amplitudes $\langle Uv,Uv\rangle=\langle v,v\rangle$ e $\langle Uv,Uw\rangle=\langle v,w\rangle$. 
Actually, in quantum theory, the field of interest is broadened to \emph{projective} unitary representations~\cite{Bargmann,varadarajan,Masahito}. Indeed, a \emph{pure state} of a quantum-mechanical system is described by a \emph{unit vector ray} $\underline{\mathbf{v}}\coloneqq\{e^{i\theta}\mathbf{v},\,\theta\in\RR\}$ for a vector $\mathbf{v}\in\cH$ of norm $1$. Wigner's theorem~\cite{Wigner} establishes that any symmetry transformation preserving transition probabilities is implemented by a unitary or antiunitary operator $\mathdutchcal{U}$ unique only up to a phase factor. Consequently, a symmetry of a quantum system under the group $G$ is represented by operator rays $\underline{\mathdutchcal{U}(g)}\coloneqq\{e^{i\theta}\mathdutchcal{U}(g),\ \theta\in\RR\}$, satisfying the group law up to a scalar multiplier: $\underline{\mathcal{U}(g)}\,\underline{\mathcal{U}(h)}=\underline{\mathcal{U}(gh)}$, for every $g,h\in G$. This motivates how in quantum theory one necessitates projective unitary representations.

\begin{definition}
A \emph{projective unitary representation} of a topological group $G$ is a map $\underline{\mathdutchcal{U}}\colon G\rightarrow \mathrm{PU}(\cH)\coloneqq \U(\cH)/\U(1)$, $g\mapsto \underline{\mathdutchcal{U}}(g)\coloneqq \underline{\mathdutchcal{U}(g)}$, which assigns a unitary operator ray $\underline{\mathdutchcal{U}(g)}$ to every $g\in G$, and is
\begin{enumerate}
\item a homomorphism, i.e., 
\beq \label{eq:projrepsn}
\underline{\mathdutchcal{U}(g)}\,\underline{\mathdutchcal{U}(h)}=\underline{\mathdutchcal{U}(gh)},\qquad g,h\in G,
\eeq 
\item continuous with respect to the topology of $G$ and the quotient topology on $\mathrm{PU}(\cH)$.
\end{enumerate}
\end{definition}

It follows that $\underline{\mathdutchcal{U}(e)}=\underline{id}$ and $\underline{\mathdutchcal{U}(g^{-1})}=\underline{\mathdutchcal{U}(g)}^{-1}$. The quotient topology on $\mathrm{PU(\cH)}$ has a physical meaning: it coincides with the topology which makes continuous any mapping of a correspondence $\Phi$ as above to the associated transition probabilities from states $\underline{\mathbf{v}}$ to $\Phi \underline{\mathbf{v}}$. This topology makes $\mathrm{PU}(\cH)$ a second countable Polish group.

Choosing representatives $\mathdutchcal{U}(g)\in\U(\cH)$ for $\underline{\mathdutchcal{U}(g)}$, for every $g\in G$, Eq.~\eqref{eq:projrepsn} yields
\beq \label{eq:weakerreo}
\mathdutchcal{U}(g)\mathdutchcal{U}(h)=\omega(g,h)\mathdutchcal{U}(gh),
\eeq 
where $\omega(g,h)\in\U(1)$. The function $\omega\colon G\times G\rightarrow\U(1)$ is an instance of what is called \emph{Schur multiplier} or \emph{2-cocycle}, and it depends on the choice of the representatives. %Thus, the unitary operators $\mathdutchcal{U}_g$ define a representation of $G$ only up to a phase factor, which depends on the selection of the representatives.
Due to this weaker condition, the set of projective unitary representations of $G$ contains and is generally larger than the set of unitary representations of $G$: unitary representations are projective ones for which there is a set of representatives such that $\omega\equiv1$. %In many quantum-mechanical problems, one is able to select representatives for which $\omega(g,h)=0$ for all $g,h\in G$, but this factor cannot always be eliminated like this (e.g. Galilean group). Investigate whether this can be achieved or not.

The definitions of subrepresentation, equivalence and irreducibility are naturally extended to projective unitary representations. The Haar measure on a group remains a valuable tool for the analysis of orthogonality relations between its projective irreps.

\begin{remark}\label{proj:rem:saràperfat}
Let $\underline{\mathdutchcal{U}}\colon G\rightarrow \mathrm{PU}(n)$ be an $n$-dimensional projective unitary representation of $G$ (on $\cH\simeq \CC^n$). It uniquely gives a unitary representation $\rho$ on the vector space of $n\times n$ complex matrices, by conjugation:
\beq \label{eq:chanlrepn}
\rho\colon G\rightarrow \U(\mM(n,\CC)),\quad \rho(g)M\coloneqq \mathdutchcal{U}(g)M\mathdutchcal{U}(g)^{\dagger},
\eeq 
where $\mathdutchcal{U}(g)^\dagger$ denotes the adjoint operator of $\mathdutchcal{U}(g)$. This does not depend on the choice of the representatives $\mathdutchcal{U}(g)\in\U(n)$ of $\underline{\mathdutchcal{U}(g)}\in\mathrm{PU}(n)$, as the phases factors of $\mathdutchcal{U}(g)$ and $\mathdutchcal{U}(g)^\dagger$ cancel out. Every $\rho(g)$ is in fact unitary, according to the \emph{Hilbert-Schmidt} product $\langle A,B\rangle\coloneqq \Tr(A^\dagger B)$ in $\mM(n,\CC)$: we have $\langle \rho(g)M,\rho(g)L\rangle%=\langle\mathdutchcal{U}(g)M\mathdutchcal{U}(g)^{\dagger},\,\mathdutchcal{U}(g)L\mathdutchcal{U}(g)^{\dagger}\rangle
=\Tr\left(\mathdutchcal{U}(g)M^\dagger\mathdutchcal{U}(g)^{\dagger}\mathdutchcal{U}(g)L\mathdutchcal{U}(g)^{\dagger}\right) %= \Tr\left(\mathdutchcal{U}(g)M^\dagger L^\dagger\mathdutchcal{U}(g)^{\dagger}\right)= \Tr\left(M^\dagger L^\dagger\mathdutchcal{U}(g)^{\dagger} \mathdutchcal{U}(g)\right)
= \Tr\left(M^\dagger L\right)=\langle M,L\rangle$ by cyclicity of the trace, for every $L,M\in\mM(n,\CC)$ and every $g\in G$.

Conversely, $\rho$ allows to uniquely reconstruct $\underline{\mathdutchcal{U}}$, by any choice of unitary operators $\mathdutchcal{U}(g)$ implementing $\rho$; the $\mathdutchcal{U}(g)$s are representatives of the $\underline{\mathdutchcal{U}(g)}$, which automatically satisfy Eq.~\eqref{eq:weakerreo}. %basta che poi faccio le classi di U(g) al variare della fase e ho la repn proiettiva
\end{remark}

\medskip

We conclude with a brief survey on the representations of angular momentum and spin in standard quantum mechanics. The $p$-adic translation of angular momenta and spins is what we aim to develop in our wider program~\cite{our2nd} in line with Volovich's philosophy~\cite{VolSeminal}.
% \begin{remark}\label{rem:proremSO3} We give a brief overview of representations of angular momentum and spin in standard quantum mechanics.

The study of real and complex representations of a standard Lie group often proceeds via the associated Lie algebra. %repn di G -> repn di Lie (Lie ne ha di più in generale) %l'associazione è per derivazione o per exponential map
In particular, the Lie algebra isomorphism $\so(3) \simeq \mathfrak{su}(2)$ implies an identical local structure, yet the global properties differ: $\SU(2)$ is the universal covering group of $\SO(3)$, with $\SO(3)\simeq \SU(2)/\{\pm\mI\}$. While, $\so(3)$ and $\mathfrak{su}(2)$ share the same representation theory, the representation theory of $\SU(2)$ encompasses that of $\SO(3)$.

Because $\SU(2)$ is simply connected, there is a one-to-one correspondence between its irreps and those of $\mathfrak{su}(2)$%questo vale in generale per homomorphismi da questi gruppi/algebre, Hall 5.7
. The Lie algebra $\mathfrak{su}(2)$ (and then the group $\SU(2)$) has exactly one irrep in every dimension, up to equivalences. These irreps are conventionally indexed by $l \in \{0,\frac{1}{2},1,\frac{3}{2},\dots\}$ and have dimension $2l+1$.

On the other hand, $\SO(3)$ is connected but not simply connected. The $l$-th irrep of $\so(3)\simeq\mathfrak{su}(2)$ is associated with an irrep of $\SO(3)$ if and only if $l$ is an integer. Instead, half-integer values of $l$ correspond to %an irrep of the universal covering $\SU(2)$ but
projective irreps of $\SO(3)$. In quantum mechanics, those representations are associated with integer spin (bosons) or orbital angular momentum when $l$ is integer, and with half-integer spin (fermions) when $l$ is half-integer. In particular, a spin one-half particle, or abstractly a qubit, is realised by $l=\frac{1}{2}$, for which the action of $\SU(2)$ on the state space $\CC^2$ corresponds to a projective action of $\SO(3)$ on the Bloch sphere~\cite{Hall,Georgi}. 
%\end{remark}

Accordingly, we give the following definition of \emph{$p$-adic qubit}, as proposed in~\cite{our2nd}.
\begin{definition}
For every prime $p$, we call \emph{$p$-adic qubit} any pair $(\cH,\underline{\mathdutchcal{U}})$, where $\cH\simeq\CC^2$ is a two-dimensional complex Hilbert space, and $\underline{\mathdutchcal{U}}$ is a projective unitary irreducible representation of $\SO(3)_p$ on $\cH$. 
\end{definition}

\section{About representations of \texorpdfstring{$\SO(3)_p$}{Lg}}
\label{sec:irrepsSO3pDp}
In this section, we begin presenting our contributions on the representation theory of $\SO(3)_p$.

\subsection{Factorisation of representations modulo \texorpdfstring{$p^k$}{Lg}}
\label{fact:repnspk}
A fundamental tool to understand the structural aspects of $\SO(3)_p$ is quotienting the group modulo $p^k$, $k\in\NN$, which is possible since the matrix entries are $p$-adic integers, for every prime $p$. As seen in Theorem~\ref{theor-invlimSO3p}, these quotients are finite groups $G_{p^k}$ and, thus, they are much easier to study, while reproducing the profinite $\SO(3)_p$ as an inverse limit. The groups $G_{p^k}$ are pivotal also because their representations give rise to representations of the entire $\SO(3)_p$: if $\underline{\mathdutchcal{U}_{p,k}}$ is a projective irrep of $G_{p^k}$ on $\CC^n$ for some $k,n\in\NN$, then
\beq \label{eq:factorisationreps}
\underline{\mathdutchcal{U}_p}\coloneqq \underline{\mathdutchcal{U}_{p,k}}\circ \pi_k
\eeq 
is a projective irrep of $\SO(3)_p$ on $\CC^n$. In fact, $\pi_k$ from Eq.~\eqref{eq:grouphoprpi} is a continuous homomorphism, and so is $\underline{\mathdutchcal{U}_p}$; since $\pi_k$ is surjective, for every subspace $\mathcal{M}\subseteq\CC^n$, $\underline{\mathdutchcal{U}_p}\left(\SO(3)_p\right)\mathcal{M}=\mathcal{M}$ is equivalent to $\underline{\mathdutchcal{U}_{p,k}}\left(G_{p^k}\right)\mathcal{M}=\mathcal{M}$%, which in turn is equivalent to $\mathcal{M}=\{\mathbf{0}\},\CC^n$, because $\rho_{p,k}$ is irreducible, hence $\rho_p$ is irreducible too.
, thus $\underline{\mathdutchcal{U}_p}$ is irreducible if and only if $\underline{\mathdutchcal{U}_{p,k}}$ is irreducible (i.e. $\mathcal{M}=\{\mathbf{0}\},\CC^n$).

A crucial question is whether the (projective) irreps of $\SO(3)_p$ are exhausted by those that factorise through the finite quotients $G_{p^k}$. This holds true for a profinite abelian group (e.g. the $p$-adic groups of planar rotations $\SO(2)_{p,\dude}$), by~\cite[Prop.~4.3.3]{reiter2000}. In this case, the irreps are all one-dimensional and close the dual group, and the dual group of an inverse-limit group is isomorphic to the direct limit of the dual groups of the inverse family. %poi si estende questo alle projective come facciamo sotto anche noi

We therefore ask if a similar factorisation property holds for the non-abelian profinite group $\SO(3)_p$. Specifically, do all projective irreps of $\SO(3)_p$ factorise through $G_{p^k}$ for some $k\in\mathbb{N}$, or do other such representations of $\SO(3)_p$ exist that are not ``lifted'' from the $G_{p^k}$s? We are going to provide a positive answer, even for finite-dimensional projective unitary representations.

\begin{proposition}\label{prop:gennotoJess}
Let $G$ be a profinite group and $\rho\colon G\rightarrow \GL(n,\CC)$ an $n$-dimensional representation of $G$. Then, $\rho$ is continuous if and only if $\ker(\rho)$ is open in $G$.
\end{proposition}
\begin{proof}
For the direct part, it is enough to prove that $\ker(\rho)$ contains an open neighbourhood, say $U_e$, of the identity $e$ of $G$. In fact, $gU_e\subseteq \ker(\rho)$ for every $g\in\ker(\rho)$, as $\ker(\rho)$ is a (normal) subgroup, so $\bigcup_{g\in\ker(\rho)}gU_e\subseteq \ker(\rho)$; also $e\in U_e$ implies $g\in gU_e$ and $\bigcup_{g\in\ker(\rho)}gU_e$ covers $\ker(\rho)$. Thus, $\ker(\rho)=\bigcup_{g\in\ker(\rho)}gU_e$, and each $gU_e$ is open because translations in a topological group are homeomorphisms and $U_e$ is open. Then, $\ker(\rho)$ is open as a union of open sets (then clopen, cf.\ Remark~\ref{rem:clopen}). 

We invoke the fact that $\GL(n,\CC)$, as a Lie group, has \emph{no small subgroups} (cf.\ Proposition~$2.17$ in~\cite{Liegen}): there exists an open neighbourhood $U_{\mI_n}$ of $\mI_n\in \GL(n,\CC)$ which contains no non-trivial subgroups of $\GL(n,\CC)$, because powers of some non-identity element $g\in \GL(n,\CC)$ close to $\mI_n$ ``escape'' from $\mI_n$. Let $V\coloneqq\rho^{-1}(U_{\mI_n})$. $V$ is open in $G$, since $\rho$ is continuous and $U_{\mI_n}$ is open. $e\in V$ because $\mI_n\in U_{\mI_n}$, $\rho$ is a homomorphisms, and $\rho^{-1}(\mI_n)=\rho^{-1}(\rho(e))\ni e$. This shows that $V$ is an open neighbourhood of $e$ in $G$. Also, $\ker(\rho)\subseteq V$ since $\ker(\rho)=\rho^{-1}(\mI_n)$%allargo un po' intorno all'identità
. By Theorem~\ref{prop:profequiv}, point~3. or~4., there is a fundamental system of neighbourhoods of $e$ which are open normal subgroups of $G$%, e.g. $\ker(f_i)$
. Then, there is an open normal subgroup $N$ of $G$ such that $N\subseteq V$. Since $\rho$ is a homomorphism, $\rho(N)$ is a subgroup of $\GL(n,\CC)$, and $\rho(N)\subseteq \rho(V)\subseteq U_{\mI_n}$. By the no-small-subgroup argument $\rho(N)=\{\mI_n\}$, thus $N\subseteq\ker(\rho)$. We have found $N$ as an open neighbourhood of $e$ in $G$ contained in $\ker(\rho)$, hence $\ker(\rho)=\bigcup_{g\in\ker(\rho)}gN$ is open.

For the converse part, we equivalently prove that if $U\subseteq \GL(n,\CC)$ is open, then $\rho^{-1}(U)$ is open in $G$. If $g\in \rho^{-1}(U)$, then $g\ker(\rho)\subseteq\rho^{-1}(U)$, and $g\ker(\rho)$ is open. Again, $\rho^{-1}(U)=\bigcup_{g\in\rho^{-1}(U)}g\ker(\rho)$ is open as the union of open subsets.
\end{proof}

\begin{corollary}\label{cor:fattrepn}
Let $G\simeq\varprojlim\{G_k\}_{k\in I}$ be a profinite group, with canonical projections $f_k$. If $\rho\colon G\rightarrow\GL(n,\CC)$ is a continuous representation, then there exists $k\in I$ such that $\rho$ factorises through $G_k$ as $\rho=\rho_k\circ f_k$, for a continuous representation $\rho_k\colon G_k\rightarrow\GL(n,\CC)$.
\end{corollary}
\begin{proof}
By Proposition~\ref{prop:gennotoJess}, $\ker(\rho)$ is open in $G$. In particular, by %point~3. of Theorem~\ref{prop:profequiv}, 
Lemma~\ref{lemma:baseeeeker}, %def di base, e def fundamental neighborhood in base, un aperto contiene un elemento di base e a sua volta di fund syste
$\ker(f_k)\subseteq \ker(\rho)$ for some $k\in I$. Then, $\rho$ factorises through $G_k$ by the fundamental theorem of homomorphism. 
\end{proof}

\begin{corollary}\label{cor:seratoizaloafatrza}
Let $G\simeq\varprojlim\{G_k\}_{k\in I}$ be a profinite group, with canonical projections $f_k$. Every finite-dimensional projective unitary representation $\underline{\mathdutchcal{U}}\colon G\rightarrow \mathrm{PU}(n)$ of $G$ factorises through $G_k$ as $\underline{\mathdutchcal{U}}=\underline{\mathdutchcal{U}_k}\circ f_k$, for some $k\in I$ and projective unitary representation $\underline{\mathdutchcal{U}_k}\colon G_k\rightarrow \mathrm{PU}(n)$.
\end{corollary}
\begin{proof}
According to Remark~\ref{proj:rem:saràperfat}, $\underline{\mathdutchcal{U}}$ gives a unitary representation on $\U(\mM(n,\CC))$ as in Eq.~\eqref{eq:chanlrepn}. Observe that 
%\beq \label{eq:eqkerproj} \ker(\rho) = \ker(\underline{\mathdutchcal{U}}), \eeq because 
\begin{align}
    \ker(\rho)%&=\{g\in G\midd \rho(g)=id\}\nonumber\\
    & = \{g\in G\midd \mathdutchcal{U}(g)M\mathdutchcal{U}(g)^{\dagger} =M\textup{ for all } M\in \mM(n,\CC)\}\nonumber\\
    & = \{g\in G\midd \mathdutchcal{U}(g)=e^{i\theta}\mI_n\textup{ for some }\theta\in\RR\}\nonumber\\
    %& = U^{-1}\big(\U(1)\mathrm{I}_{d\times d}\big)\\
    & = \{g\in G\midd \underline{\mathdutchcal{U}(g)}=\underline{\mI_n}\in\mathrm{PU}(n)\}\nonumber\\
    & = \ker(\underline{\mathdutchcal{U}}),
\end{align}
where, in the third equality, we used the fact that the only unitary matrices $\mathdutchcal{U}(g)\in\U(n)$ which commute with every $M\in\mM(n,\CC)$ are phase-multiples of the identity matrix. By Corollary~\ref{cor:fattrepn}, $\ker(f_k)\subseteq \ker(\rho)=\ker(\underline{\mathdutchcal{U}})$ for some $k\in I$, hence $\underline{\mathdutchcal{U}}$ factorises through $G_k$ too, by the fundamental theorem of homomorphism.
\end{proof}

All these results apply to the profinite group $\SO(3)_p\simeq \varprojlim\{G_{p^k}\}_\NN$, with canonical projections $\pi_k$ of entry-wise reduction modulo $p^k$ (cf.~Theorem~\ref{theor-invlimSO3p}), leading to the following result.
\begin{theorem}\label{cor:fattrepNOP}
For every prime $p$, and every $n$-dimensional ($n\in\NN$) projective unitary representation (resp. projective irrep) $\underline{\mathdutchcal{U}_p}$ of $\SO(3)_p$, there exists $k\in\NN$ and an $n$-dimensional projective unitary representation (resp. projective irrep) $\underline{\mathdutchcal{U}_{p,k}}$ of $G_{p^k}$ such that $\underline{\mathdutchcal{U}_p}$ factorises through $G_{p^k}$ as $\underline{\mathdutchcal{U}_p}=\underline{\mathdutchcal{U}_{p,k}}\circ\pi_k$, i.e., such that the following diagram commutes.
\beq \notag 
\begin{tikzcd}
\mathrm{SO}(3)_p \arrow[d, "\pi_k"'] \arrow[r, "\underline{\mathdutchcal{U}_p}"] & \mathrm{PU}(n) \\
G_{p^k} \arrow[ru, "{\underline{\mathdutchcal{U}_{p,k}}}"']                      &   
\end{tikzcd}
\eeq 
\end{theorem}
\begin{proof}
This theorem in terms of projective unitary representations is the specialisation of Corollary~\ref{cor:seratoizaloafatrza} to the group $\SO(3)_p$. %Now, consider a projective irrep $\underline{\mathdutchcal{U}_p}\colon \SO(3)_p\rightarrow \mathrm{PU}(n)$. There exists some projective unitary representation $\underline{\mathdutchcal{U}_{p,k}}\colon G_{p^k}\rightarrow \mathrm{PU}(n)$ such that $\underline{\mathdutchcal{U}_p}=\underline{\mathdutchcal{U}_{p,k}}\circ\pi_k$. 
Moreover, $\underline{\mathdutchcal{U}_p}$ is an irrep if and only if $\underline{\mathdutchcal{U}_{p,k}}$ is an irrep, because $\pi_k$ is surjective (see below Eq.~\eqref{eq:factorisationreps}).
\end{proof}
Recall that $\pi_k=\pi_{kl}\circ\pi_l$ for every $k\leq l$. If $k\in\NN$ is the minimum natural number such that $\underline{\mathdutchcal{U}_p}$ factorises through $G_{p^k}$, as $\underline{\mathdutchcal{U}_p}=\underline{\mathdutchcal{U}_{p,k}}\circ\pi_k$, then $\underline{\mathdutchcal{U}_p}=(\underline{\mathdutchcal{U}_{p,k}}\circ\pi_{kl})\circ\pi_l=\underline{\mathdutchcal{U}_{p,l}}\circ\pi_l$, where $\underline{\mathdutchcal{U}_{p,l}}\coloneqq \underline{\mathdutchcal{U}_{p,k}}\circ\pi_{kl}$ is a representation of $G_{p^l}$, i.e., $\underline{\mathdutchcal{U}_p}$ factorises through $G_{p^l}$ for every $l\geq k$.

\begin{remark}\label{rem:invlimrepmac}
By Theorem~\ref{cor:fattrepNOP}, the representations lifted by the irreps %(hence finite-dimensional, by Theorem~\ref{barurazcdiresumun}) 
of $G_{p^k}=\SO(3)_p\mod p^k$, $k\in\NN$, are all and the only irreps of $\SO(3)_p$, for every prime $p$. This offers a machinery to exhaustively investigate and classify irreps of the infinite group $\SO(3)_p$, through those of the finite groups $G_{p^k}$. The latter are easier to study, via standard tools of representation theory of finite groups~\cite{Etingof,Serre:reps,FultonHarris}. This has already served to provide first irreps of $\SO(3)_p$ and $p$-adic qubits, for every prime $p$~\cite{our2nd}. 

The future research programme is to achieve a systematic classification of the irreps of $\SO(3)_p$. This can be done by the Kirillov Orbit Method and Howe's theory~\cite{kirillov,Howe,altromaKirillov}, to analyse the representations of the pro-$p$ main congruence subgroup $\ker(\pi_1)=\left(\mI_3+p\mathsf{M}(3,\ZZ_p\right)\cap \SO(3)_p$ %the normal pro-$p$ subgroup of finite index
 and its associated $\ZZ_p$-Lie algebra%(by Lazard correspondence)
 . %repns di \ket(\pi_1) directly from the geometry of co-adjoint orbits in the dual of the Lie algebra.
One extend these results to the full inverse-limite group $\SO(3)_p$%by Clifford Theory
, by handling the action of the finite quotient $G_p$ on the representations of $\ker(\pi_1)$.
% filtration analysis
This programme lies beyond the scope of the present article, in which we focus on the irreps of $\SO(3)_p$ lifted from $G_p$, and combine them to build systems of $p$-adic qubits and gates acting on them.
\end{remark}

From this point on, we focus on irreps of $\SO(3)_p$ --- as they form a subset of the projective irreps --- factorising through $G_p$.

%%%%%%%%%%%%%%%%%%%%%%%%%%%%%%%%%%%%%%%%%%%%%%%%%%%%%%%%%%%%%%%%%%%%%%%%%%%%%
\subsection{Structure and irreps of \texorpdfstring{$G_p$}{Lg}}\label{subsec:strirrpsGp}
%In this subsection, we assume that $p$ is an odd prime (if $p=2$, we already know~\cite{our2nd} that $G_2\simeq\mathrm{D}_3$). 

If $G$ is a group, $H$ a subgroup of $G$, and $N$ a normal subgroup of $G$, we say that $G$ is the \emph{(inner) semidirect product of $N$ and $H$}, denoted by $G=N\rtimes H$, iff $G=NH$ with $N\cap H=\{e\}$. Also, $G=N\rtimes H$ is uniquely defined up to isomorphism by the left action of $H$ by automorphisms on $N$, i.e. by the homomorphism $\psi\colon H\rightarrow \mathrm{Aut}(N)$, $\psi(h)(m)\coloneqq hmh^{-1}$%se fai h^{-1}mh, l'azione definitoria non sarà più banalmente a sinistra, quindi non darà più correttamente un omomorfismo a meno che non lo pensi da destra (è un antiomomorfismo \phi(ab)=phi(b)phi(a))
. Let us also present the \emph{dihedral group} of order $2n$ by 
\beq \label{eq:presDieh}
\mathrm{D}_n=\langle r,x,\midd r^m=x^2=e,xrx=r^{-1}\rangle.
\eeq 

\begin{theorem}\label{theor:structGp}
For every odd prime $p$,
\beq 
G_p\simeq N_{p^2}\rtimes H_{2(p+1)}\simeq C_p^2\rtimes \mathrm{D}_{p+1},
\eeq 
where $N_{p^2}\coloneqq\{L(1,0,c,d,1),\textup{ for } c,d,\in\ZZ/p\ZZ\}\simeq C_p^2$ is the unique Sylow $p$-subgroup of $G_p$, and $H_{2(p+1)}\coloneqq\{L(a,b,0,0,s)\midd a^2-vb^2\equiv1,s\equiv\pm1\}\simeq \mathrm{D}_{p+1}$. The action defining this semidirect product is
\begin{align}
\psi_p\colon &H_{2(p+1)}\rightarrow\mathrm{Aut}(N_{p^2}),\\
&\psi_p(L(a_0,b_0,0,0,1))\colon L(1,0,c,d,1)\mapsto L(1,0,a_0c-b_0d,a_0d-vb_0c,1),\notag\\
&\psi_p(L(1,0,0,0,-1))\colon L(1,0,c,d,1)\mapsto L(1,0,-c,d,1), \notag    
\end{align}
where $(a_0,b_0)$ is a generator of the group of solutions $(a,b)$ to $a^2-vb^2\equiv1\mod p$. This action corresponds to the following representation of the dihedral group%looking at the action of these matrices on the column vector (c,d)\in\mathbb{F}_p^2
: %non è la rappresentazione del diedrale usuale che avevamo nel secondo paper ortogonale rispetto alla forma A=(1&0\\0&-v) come R^TAR=A, ma questa è con matrice di r inversa trasposta e matrice di x opposta e in totale è ortogonale rispetto alla forma (-v&0\\0&1) equivalente ad A, oppure rispetto alla stessa A=(1&0\\0&-v) se scrivo RAR^T=A... le due sono legate così: se R ed X sono le classiche, e R', X' sono queste nuove strane, allora cambio di base C:=(0&1\\-1&0) tale che C^{-1}RC=R' e C^{-1}XC=X'.. si ha anche C^T=C^{-1}, e C è il cambio di base che mappa l'una nell'altra forma quadratica
\beq 
\widetilde{\psi_p}\colon\mathrm{D}_{p+1}\rightarrow \mathrm{Aut}(C_p^2)\simeq\mathrm{GL}(2,\mathbb{F}_p),\quad \widetilde{\psi_p}(r)=\begin{pmatrix}a_0 & -b_0\\ -vb_0 & a_0 \end{pmatrix},\ \widetilde{\psi_p}(x)=\begin{pmatrix}-1&0\\0&1 \end{pmatrix}.
\eeq 
\end{theorem}
\begin{proof}
$N_{p^2}$ is a $p$-subgroup of $G_p$. Since $\lvert G_p\rvert=2(p+1)p^2$%e tutti i sottogruppi devono avere ordine che divide |Gp| per Lagrange, e un Sylow $p$-subgroup è un sottogruppo di Gp di ordine potenza di p massimale, e (p+1,p)=1 per cui non trovo in p+1 altri fattori di per avere $p$-sottogruppi più grandi contenenti quelli da p^2 
, a subgroup of order $p^2$ of $G_p$ is a Sylow $p$-subgroups of $G_p$. By the Sylow's theorems, the number $n_p$ of Sylow $p$-subgroups of a finite group $G$ is $n_p=\lvert N_G(P):G\rvert$ where $P$ is any Sylow $p$-subgroup of $G$ and $N_G(P)$ denotes the normalizer of $P$ in $G$. Then, a Sylow $p$-subgroup $P$ is unique if and only if $P$ is normal (i.e. $N_G(P)=G$). It is a straightforward calculation to check that $L(a,b,c,d,s)L(1,0,g,h,1)L(a,b,c,d,s)^{-1}%=\begin{pmatrix}a&svb&0\\b&sa&0\\ c&d&s\end{pmatrix}\begin{pmatrix} 1&0&0\\0&1&0\\g&h&1 \end{pmatrix}\begin{pmatrix}a&-vb&0\\ -sb&sa&0\\ bd-sac&svbc-ad&s\end{pmatrix}=\begin{pmatrix} 1&0&0\\0&1&0\\g'&h'&1 \end{pmatrix}
\in N_{p^2}$ for all $a,b,c,d,g,h,s\in \ZZ/p\ZZ$, $a^2-vb^2\equiv1\mod p$, $s=\pm1$, i.e. $N_{p^2}$ is a normal subgroup of $G_p$.
%è normale \beq MNM^{-1}= \eeq 
Then, $L(1,0,g,h,1)L(1,0,g',h',1)=L(1,0,g+g',h+h',1)$ %è sottogruppo
enlightens the isomorphism 
\beq \label{eq:normalparGpcoZpZp}
N_{p^2}\simeq (\ZZ/p\ZZ)^2,\qquad L(1,0,c,d,1)\leftrightarrow (c,d).
\eeq %è quindi normale abeliano

We have $\left\lvert G_p/N_{p^2}\right\rvert = 2(p+1)$.
Since $L(a,b,c,d,s)L(1,0,g,h,1)%=\begin{pmatrix}a&svb&0\\ b&sa&0\\ c&d&s\end{pmatrix}\begin{pmatrix}1&0&0\\ 0&1&0\\ g&h&1\end{pmatrix}=\begin{pmatrix}a&svb&0\\ b&sa&0\\ c+sg&d+sh&s\end{pmatrix}
=L(a,b,c+sg,d+sh,s)$, a left coset $L(a,b,c,d,s)N_{p^2}$ is identified by $(a,b,s)$ and has representative $L(a,b,0,0,s)$. The matrices of the form $L(a,b,0,0,s)$ close a group, say $H_{2(p+1)}$%because $L(a,b,0,0,s)L(a',b',0,0,s')=L(aa'+svbb', a'b+sab',0,0,ss')$
, isomorphic to $G_p/N_{p^2}$. we recall from~\cite[Subsec.~IV.B]{our2nd} that 
\beq \label{eq:homochepoiisop}
\mathdutchcal{F}_p\colon G_p\rightarrow %\mathdutchcal{F}_p(G_p)\subset 
\GL(2,\ZZ/p\ZZ),\quad \begin{pmatrix}
a & svb & 0 \\
b & s a & 0 \\
c & d & s
\end{pmatrix} \mapsto \begin{pmatrix}a & svb \\ b & sa\end{pmatrix}
\eeq 
is a homomorphism, and
\beq \label{eq:giaiso}
\upvarphi_p\colon \mathdutchcal{F}_p(G_p)\rightarrow\mathrm{D}_{p+1}, \quad \begin{pmatrix}a_0 & vb_0 \\ b_0 & a_0\end{pmatrix}\mapsto r,\ \ \begin{pmatrix} 1 & 0\\ 0 & -1 \end{pmatrix}\mapsto x
\eeq
is an isomorphism, where $(a_0,b_0)$ is a generator of the group of solutions $(a,b)$ to $a^2-vb^2\equiv1\mod p$. The maps~\eqref{eq:homochepoiisop},~\eqref{eq:giaiso} provide the isomorphisms $H_{2(p+1)}\simeq \mathdutchcal{F}_p(H_{2(p+1)})=\mathdutchcal{F}_p(G_p)\simeq\mathrm{D}_{p+1}$. %We have shown that G_p/N_{p^2}\simeq H_{2(p+1)}\simeq \mathrm{D}_{p+1}.

Lastly, $G_p\simeq N_{p^2}H_{2(p+1)}$ where $N_{p^2}\cap H_{2(p+1)}=\{\mI_3\}$. The action defining the semidirect product is $\psi_p\colon H_{2(p+1)}\rightarrow \mathrm{Aut}(N_{p^2})$, $\psi_p(h)(n)=HNH^{-1}$. It is enough to express it for the generators $L(a_0,b_0,0,0,1)\leftrightarrow r$ and $L(1,0,0,0,-1)\leftrightarrow x$.
\end{proof}

\begin{comment}
By using GAP, we find that
\beq 
G_3\simeq (\mathrm{S}_3\times \mathrm{S}_3)\rtimes C_2.
\eeq 
Indeed, a possibility is $\mathrm{S}_3\times \mathrm{S}_3\simeq N_3 \unlhd G_3$, $N_3=\big\{L(\pm1,0,c,d,s),\textup{ for }s\in\{\pm1\},\,c,d\in\ZZ/3\ZZ\big\}$ as in Eq.~\eqref{eq:G3normal36N}, where $\mathrm{S}_3\times\{e\}\simeq \langle (g_1,e)\coloneqq L(1,0,0,-1,1)%=\begin{pmatrix} 1& 0& 0 \\ 0& 1& 0 \\ 0& -1& 1 \end{pmatrix}
,\, (g_2,e)\coloneqq L(-1,0,0,-1,-1)% =\begin{pmatrix} -1& 0& 0 \\ 0& 1& 0\\ 0& -1& -1 \end{pmatrix}
\rangle\leq N_3$, $\{e\}\times \mathrm{S}_3\simeq\langle (e,h_1)\coloneqq L(1,0,-1,0,1) %=\begin{pmatrix} 1& 0& 0 \\ 0& 1& 0 \\ -1& 0& 1 \end{pmatrix}
,\, (e,h_2)\coloneqq L(1,0,1,0,-1) %= \begin{pmatrix} 1& 0& 0 \\ 0& -1& 0 \\ 1& 0& -1 \end{pmatrix}
\rangle\leq N_3$, and $C_2\simeq\{\mI_2,t\coloneqq M\coloneqq L(0,-1,1-1,-1)%=\begin{pmatrix} 0&-1&0 \\ -1&0&0 \\ 1&-1&-1 \end{pmatrix}
\}\leq G_3$. Here, the action of $C_2$ on $\mathrm{S}_3\times \mathrm{S}_3$, defining the semidirect product is given by the swap of the direct components:
\beq 
t(g_i,e)t^{-1}=(e,h_i)\quad \textup{for }i=1,2.
\eeq 
while, for every prime from $5$ to $29$ included,
\beq 
G_p\simeq (C_p\times C_p)\rtimes \mathrm{D}_{p+1}.
\eeq
MA IN REALTà DA SOPRA è EVIDENTE CHE SONO LA STESSA COSA PER TUTTI GLI ODD PRIMES
\end{comment}

Being a semidirect product by a normal abelian subgroup $N_{p^2}$, the irreps of $G_p$ can be constructed from those of certain subgroups of $H_{2(p+1)}$, by the so called \emph{method of little groups} of Wigner and Mackey~\cite{Serre:reps,Mackey}, which we briefly recall.

Let $G$ be a finite group such that $G=A\rtimes H$ where $A$ is a normal abelian subgroup of $G$. The irreps of $A$ are one-dimensional and close the dual group $\widehat{A}=\mathrm{Hom}(A,\U(1))\simeq A$. %Serre: The group $G$ acts on $\widehat{A}$ by $(g.\chi)(a)=\chi(g^{-1}ag)$ for $g\in G, a \in A,\chi\in\widehat{A}$... questo lo fa per avere una azione qui a sinistra, ma complicando invece le cose sulla def di prodotto semidiretto che ora non sarebbe più data da un omomorfismo ma da un antiomomorfismo cioè dall'azione h^{-1}ah da destra (inoltre fa l'azione su tutto G anzichè solo H, ma questo è strettamente correlato=....    ma faccio alla Mackey hah^{-1} (gruppo definito a sinistra ok, azione qui sui caratteri a destra):
The group $H$ acts on $\widehat{A}$ from the right by $(\chi.h)(a)\coloneqq\chi(hah^{-1})$ for $h\in H, a \in A,\chi\in\widehat{A}$.
Let $\mathcal{O}_\chi\coloneqq\{\chi.h\textup{ for }h\in H\}$ be the orbit of $\chi\in \widehat{A}$ by $H$, and let $(\chi_i)_{i\in \widehat{A}/H}$ be a system of representatives for the orbits by $H$ on $\widehat{A}$. For every $i$, let $H_{\chi_i}\coloneqq\{h\in H\midd \chi_i.h=\chi_i\}$ be the stabiliser subgroup of $\chi_i$ of $H$, and $G_{\chi_i}\coloneqq A \rtimes H_{\chi_i}$ the corresponding subgroup of $G$. %**Serre: Extend the function $\chi_i$ from $A$ to $G_{\chi_i}$ by setting $\widetilde{\chi}_i(ah)=\chi_i(a)$ for $a\in A,h\in H_i$. Then, \widetilde{\chi}_i$ is an irreducible character of $G_{\chi_i}$ of dimension $1$, since $H_{\chi_i}$ stabilises $\chi_i$. 
Under these hypotheses, the following result holds true.

\begin{proposition}[{\cite[Prop.~25]{Serre:reps} and~\cite[Thms.~A~and~B,~pp.~42-43]{Mackey}}]
If $\rho$ is an irrep of $H_{\chi_i}$, %**Serre: then $\widetilde{\rho}\coloneqq\rho\circ\pi_i$ is an irrep of $G_{\chi_i}$, where $\pi_i\colon G_{\chi_i}\rightarrow G_{\chi_i}/A=H_{\chi_i}$ is the canonical projection. Also $\widetilde{\chi}_i\otimes \widetilde{\rho}$ is an irrep of $G_{\chi_i}$, 
then $\chi_i\rho$ defined by $\chi_i\rho(ah)=\chi_i(a)\rho(h)$ is an irrep of $G_{\chi_i}$%quell'estensione di carattere e di rappresentazione era solo per portare il tutto su G_{\chi_i}, in modo da vedere lì la rappresentazione \chi\rho di Mackey come prodotto tensore... ma non me frega e me va bene alla Mackey
, which induces an irrep $\theta_{i,\rho}$ of $G$. If $\theta_{i,\rho}$ and $\theta_{i',\rho'}$ are equivalent, then $i=i'$ and $\rho,\rho'$ are equivalent; every irreducible representation of $G$ is isomorphic to one of the $\theta_{i,\rho}$.
\end{proposition}
%The set of irreps of $G$ is in one-to-one correspondence with the (disjoint) union over $i\in \widehat{A}/H$ of the set of irreps of $H_i$ up to equivalences; every irrep of $G$ is uniquely located by a pair $(i,\rho)$.

The trivial character $\chi_{\textup{triv}}$ has orbit $\{\chi_{\textup{triv}}\}$ and stabiliser group $H_{\chi_{\textup{triv}}}=H$: this provides those irreps of $G$ induced by the irreps of $H=G/A$. We have $\dim(\theta_{i,\rho})=|\mathcal{O}_{\chi_i}|\dim(\rho)$ %, because $\dim\left(\textup{Ind}_{G_{\chi_i}}^G(V)\right)=[G_{\chi_i}:G]\dim(V)$ and $[G_{\chi_i}:G] = |G|/|G_{\chi_i}| = \frac{|A||H_{\chi_i}|}{|A||H|} = |H_{\chi_i}|/|H| = |\mathcal{O}_{\chi_i}|$
because, as a representation induced by $\chi_i\rho$ of $G_{\chi_i}$ to $G$, the dimension of $\theta_{i,\rho}$ is equal to the dimension of $\chi_i\rho$ times the index of $G_{\chi_i}$ in $G$.

By applying this method to $G_p$, we prove the following.
\begin{theorem}\label{theo:irrepsGp}
For every odd prime $p$, up to equivalences, $G_p$ has the following irreps, up to equivalence:
\begin{itemize}
    \item four one-dimensional irreps, and $\frac{p-1}{2}$ two-dimensional irreps, induced by $\mathrm{D}_{p+1}$;
    \item other $2(p-1)$ $(p+1)$-dimensional irreps, induced by the two irreps of each of the $p-1$ stabiliser subgroups $C_2$ of the orbits of the non-trivial characters of $C_p^2$.
\end{itemize}
\end{theorem}
\begin{proof}
Recall that $G_p\simeq N_{p^2}\rtimes H_{2(p+1)}$ from Theorem~\ref{theor:structGp}, where $N_{p^2}\simeq (\ZZ/p\ZZ)^2$ and $H_{2(p+1)}\simeq \mathrm{D}_{p+1}$. The group $\widehat{\ZZ/p\ZZ}\simeq \ZZ/p\ZZ$ is composed of the characters $\chi_k$ defined by $\chi_k(1)\coloneqq e^{\frac{2\pi i}{p}k}$ for $k\in\ZZ/p\ZZ$. Then, the characters of $(\ZZ/p\ZZ)^2$ are $\chi_{k_1,k_2}$ defined by $\chi_{k_1,k_2}(c,d)\coloneqq\chi_{k_1}(c)\chi_{k_2}(d)=e^{\frac{2\pi i}{p}(k_1c+k_2d)}$ for $c,d\in\ZZ/p\ZZ$, for $k_1,k_2\in\ZZ/p\ZZ$. This provides the dual group $\widehat{N_{p^2}}$ according to the isomorphism~\eqref{eq:normalparGpcoZpZp}. %We can visualise $N_{p^2}\simeq\mathbb{F}_p^2$ by column vectors as $L(1,0,c,d,1)\leftrightarrow \begin{pmatrix} c\\ d\end{pmatrix}$, and $\widehat{N_{p^2}}\simeq\mathbb{F}_p^2$ by row vectors as $\chi_{k_1,k_2}\leftrightarrow \begin{pmatrix} k_1& k_2\end{pmatrix}$, and we have $\chi_{k_1,k_2}(c,d)=\exp\left[\frac{2\pi i}{p}\begin{pmatrix} k_1& k_2\end{pmatrix}\begin{pmatrix} c\\ d\end{pmatrix}\right]$.

The action of $H_{2(p+1)}$ on the characters in $\widehat{N_{p^2}}$ is given by $\left(\chi_{k_1,k_2}.H\right)(N)=\chi_{k_1,k_2}\left(HNH^{-1}\right)=\chi_{k_1,k_2}\big(\psi_p(H)(N)\big)$. If $N=L(1,0,c,d,1)$ in $N_{p^2}$ corresponds to $(c,d)$ in $(\ZZ/p\ZZ)^2\simeq \mathbb{F}_p^2$, and $H$ in $H_{2(p+1)}$ corresponds to $h$ in $\mathrm{D}_{p+1}$, then $\chi_{k_1,k_2}\big(\psi_p(H)(N)\big)=\exp\left[\frac{2\pi i}{p}\begin{pmatrix} k_1& k_2\end{pmatrix}\begin{pmatrix}\widetilde{\psi_p}(h)\end{pmatrix}\begin{pmatrix} c\\ d\end{pmatrix}\right]=\exp\left[\frac{2\pi i}{p}\left(\begin{pmatrix}\widetilde{\psi_p}(h)\end{pmatrix}^\top\begin{pmatrix} k_1\\ k_2\end{pmatrix}\right)^\top\begin{pmatrix} c\\ d\end{pmatrix}\right]$, which means that $\chi_{k_1,k_2}.H = \chi_{\tiny{\left(\begin{pmatrix}\widetilde{\psi_p}(h)\end{pmatrix}^\top\begin{pmatrix} k_1\\ k_2\end{pmatrix}\right)^\top}}$. For every $h\in \mathrm{D}_{p+1}$, %we have $\begin{pmatrix}\widetilde{\psi_p}(h)\end{pmatrix}^\top\in\SO(2)_{p,-v}\bmod p$, i.e. 
$\begin{pmatrix}\widetilde{\psi_p}(h)\end{pmatrix}^\top$ is orthogonal with respect to the quadratic form $Q_{-v}$ of matrix representation $A_{-v}=\diag(1,-v)\bmod p$:
\beq 
\begin{pmatrix}a_0 & -b_0\\ -vb_0 & a_0 \end{pmatrix}\begin{pmatrix}1&0\\0&-v \end{pmatrix} \begin{pmatrix}a_0 & -b_0\\ -vb_0 & a_0 \end{pmatrix}^\top =
\begin{pmatrix}-1 & 0\\ 0 & 1 \end{pmatrix}\begin{pmatrix}1&0\\0&-v \end{pmatrix} \begin{pmatrix}-1 & 0\\ 0 & 1 \end{pmatrix}^\top =
\begin{pmatrix}1&0\\0&-v \end{pmatrix}.
\eeq 
Therefore, the value of the quadratic form $Q_{-v}$, redefined on $(\ZZ/p\ZZ)^2$, %i.e. $Q_{-v}((k_1,k_2))=k_1^2-vk_2^2$, 
is the invariant defining the orbits by $H_{2(p+1)}$ on $\widehat{N_{p^2}}$.

Since $v$ is a non-square in the Galois field $\mathbb{F}_p\simeq\ZZ/p\ZZ$, the ring $\mathbb{F}_p[x]/(x^2-v)$ is a field, isomorphic to the quadratic field extension $\mathbb{F}_p[i]$ with $i$ denoting a root of $x^2=v$. This is in turn isomorphic to $\mathbb{F}_p+i\mathbb{F}_p\simeq\mathbb{F}_{p^2}$. The Galois group $\mathrm{Aut}(\mathbb{F}_p[i]/\mathbb{F}_p) = \{\textup{automorphisms }\psi \colon \mathbb{F}_p[i]\rightarrow \mathbb{F}_p[i]\midd \psi(x) = x \textup{ for all }x \in \mathbb{F}_p\}$ contains only two elements: the identity and the non-trivial automorphism $i\mapsto-i$. The latter maps  $z=k_1+ik_2\in \mathbb{F}_p[i]$ to $\overline{z}\coloneqq k_1-ik_2$, and is equivalently expressed as $z\mapsto z^p$~\cite{galois}. The quadratic form $Q_{-v}$ on $(\ZZ/p\ZZ)^2$ corresponds to the (squared) norm $\mathcal{N}\colon \mathbb{F}_p[i]\rightarrow\mathbb{F}_p$, $\mathcal{N}(z)=z\overline{z}=z^{p+1}=k_1^2-vk_2^2$. Clearly $\mathcal{N}$ is a group homomorphism. Also, $\mathbb{F}_p[i]^\times$ with the multiplicative operation is a cyclic group (since the $\mathbb{F}_p[i]$ is a finite field) of order $p^2-1$, and $\ker(\left.\mathcal{N}\right\rvert_{\mathbb{F}_p[i]^\times})=\{z\in \mathbb{F}_p[i]^\times\midd z^{p+1}=1\}$ has cardinality $p+1$ (cf.~\cite[App.~A]{our2nd}). By the first isomorphism theorem, %Immagine=Dominio/Kernel
$\abs{\mathcal{N}\left(\mathbb{F}_p[i]^\times\right)}=\frac{\abs{\mathbb{F}_p[i]^\times}}{\abs{\ker(\left.\mathcal{N}\right\rvert_{\mathbb{F}_p[i]^\times})}}=\frac{p^2-1}{p+1}=p-1$, i.e., $\left.\mathcal{N}\right\rvert_{\mathbb{F}_p[i]^\times}$ is a surjective homomorphism. Therefore $Q_{-v}((\ZZ/p\ZZ)^2)=\ZZ/p\ZZ$, meaning that for every $t\in \ZZ/p\ZZ$ the equation $k_1^2-vk_2^2=t$ admits solution $(k_1,k_2)\in (\ZZ/p\ZZ)^2$. If $t=0$, then $k_1^2-vk_2^2=0$ has unique solution $(k_1,k_2)=(0,0)$, since $v$ is not a square and $Q_{-v}$ is non-isotropic. If $t\neq0$, then the set of solutions corresponds to $\{z\in \mathbb{F}_p[i]\midd \mathcal{N}(z)=t\}=\{kz_0\midd k\in\ker(\mathcal{N})\}$ where $z_0\in \mathbb{F}_p[i]^\times$, $\mathcal{N}(z_0)=t$%siccome \mathcal{N} è un homomorphismo suriettivo, tutte le fibre (le pre-immagini) sono "copie traslate" del Kernel.
, and has always cardinality $p+1$. Indeed, the action of the rotation part $C_{p+1}$ of $\widetilde{\psi_p}\left(\mathrm{D}_{p+1}\right)^\top$ is transitive on the vectors in $\mathbb{F}_p^2$ of fixed norm, and each set of solution has structure of cyclic group $C_{p+1}$. Finally, the orbits by $H_{2(p+1)}$ on $\widehat{N_{p^2}}$ are given as follows: for every $t\in \ZZ/p\ZZ$, pick $(k_1,k_2)\in (\ZZ/p\ZZ)^2$ such that $Q_{-v}(k_1,k_2)=t$, and
\beq 
\mathcal{O}_{\chi_{k_1,k_2}}=\{\chi_{k_1',k_2'}\midd (k_1',k_2')\in(\ZZ/p\ZZ)^2 \textup{ with }Q_{-v}(k_1',k_2')=Q_{-v}(k_1,k_2)\}.
\eeq 
Each of these orbits has cardinality $p+1$, except from $\abs{\mathcal{O}_{\chi_{0,0}}}=\abs{\chi_{0,0}}=1$. There are a total of $p$ orbits.

Every stabiliser group $H_{\chi_{k_1,k_2}}$ is of order $\abs{H_{\chi_{k_1,k_2}}}=\frac{\abs{H_{2(p+1)}}}{\abs{\mathcal{O}_{\chi_{k_1,k_2}}}} = 2$, i.e. $H_{\chi_{k_1,k_2}}\simeq C_2$, except from $H_{\chi_{0,0}}=H_{2(p+1)}\simeq \mathrm{D}_{p+1}$. Each of the proper $p-1$ subgroups $H_{\chi_{k_1,k_2}}\simeq C_2$ induces two irreps of dimension $\abs{\mathcal{O}_{\chi_{k_1,k_2}}}\cdot 1=p+1$. We check that these are all and the only inequivalent irreps of $G_p$ by the fact that the sum of the squares of the dimensions of all the irreps of $G_p$ must equal $\abs{G_p}$: $1^2\cdot4+2^2\cdot \frac{p-1}{2}+(p+1)^22(p-1)=2p^2(p+1)$.
\end{proof}

%*%*%*%*%*%*  Qui, si potrebbero specificatamente trovare i due gruppi stabilizzatori isomorfi a $C_2$ che inducono le quattro irreps $4$-dimensionali di $G_3$. Si potrebbero cosi indurre esplicitamente queste $4$-dimensional irreps di $G_3$, piuttosto che farle calcolare da GAP w.r.t basi scomode in Subsec.~\ref{subsec:irreps4G3}   %*%*%*

\begin{remark}\label{rem:solo4Dp3}
According to Theorem~\ref{theo:irrepsGp}, $G_3$ has four $4$-dimensional irreps, while for every prime $p>3$ the group $G_p$ has no $4$-dimensional irreps ($\SO(3)_p$ could have $4$-dimensional irreps factorising through $G_{p^k}$ for some $k>1$).
\end{remark}

For every odd prime $p$, we explicitly write the irreps of $G_p$ induced by $\mathrm{D}_{p+1}$. Also, we complete this section with structure and irreps of $G_2$. These irreps will constitute the main objects of our forthcoming results. 

For every odd prime $p$, consider $n\coloneqq p+1$, which is an even number greater than or equal to $4$. The dihedral group $\mathrm{D}_n$ has the following irreps, up to equivalence~\cite{Serre:reps}: four one-dimensional irreps,
\begin{itemize}
\item the trivial representation $\sigma_{\text{1D}}^\text{triv}$, mapping all elements to $1$, 
\item the representation $\sigma_{\text{1D}}^{(1)}$ which maps all the elements in $\langle r \rangle$ to $1$ and all the others to $-1$,
\item the representation $\sigma_{\text{1D}}^{(2)}$ which maps all the elements in $\langle r^2, x \rangle$ to $1$ and all the others to $-1$,
\item the representation $\sigma_{\text{1D}}^{(3)}$ which maps all the elements in $\langle r^2, rx \rangle$ to $1$ and all the others to $-1$,
\end{itemize}
and $\frac{n-2}{2}$ two-dimensional irreps: for $j=1,\dots,\frac{n-2}{2}$,
\beq\label{eq:repnD4caz}
\sigma_{n-1}^{(j)}\colon\quad r\mapsto
 \begin{pmatrix}
        \cos\frac{2\pi j}{n} & -\sin\frac{2\pi j}{n}\\
        \sin\frac{2\pi j}{n} & \cos\frac{2\pi j}{n}
       \end{pmatrix},
 \quad
 x\mapsto
 \begin{pmatrix}
        1 & 0\\
        0 & -1
       \end{pmatrix}.
\eeq
The one-dimensional irreps of $G_p$ induced by $\mathrm{D}_{p+1}$ are as follows, up to equivalence: 
\beq \label{1dirrepsDp+1Gp}
\mathdutchcal{triv}\coloneqq \sigma_{\text{1D}}^\text{triv} \circ \upvarphi_p \circ \mathdutchcal{F}_p, \quad \mathdutchcal{s}\coloneqq \sigma_{\text{1D}}^{(1)} \circ \upvarphi_p \circ \mathdutchcal{F}_p,\quad \mathdutchcal{t}\coloneqq \sigma_{\text{1D}}^{(2)} \circ \upvarphi_p \circ \mathdutchcal{F}_p, \quad \mathdutchcal{st}\coloneqq \sigma_{\text{1D}}^{(3)} \circ \upvarphi_p \circ \mathdutchcal{F}_p.
\eeq

For $p=2$, there exists an isomorphism $\upphi_2$ from $G_2$ to $\mathrm{D}_3$~\cite[Sec.~VII]{our2nd}:
\beq 
G_2\simeq \mathrm{D}_3.
\eeq 
$\mathrm{D}_3\simeq \mathrm{S}_3$ has the following irreps, up to equivalence:
\begin{itemize}
    \item the trivial representation $\operatorname{TRIV}$, sending all elements to $1$,
    \item the sign representation $\operatorname{SIGN}$, which maps %all rotations
    $\langle r\rangle \simeq \langle (123)\rangle  \mapsto1$ and the other elements %the reflections $x,rx,r^2x\mapsto-1$
    to $-1$,
    \item the two-dimensional irrep
\beq \label{eq:D3sigma2irrep}
\sigma_2\colon\quad r\mapsto\begin{pmatrix}
        -\frac{1}{2}&-\frac{\sqrt{3}}{2}\\\frac{\sqrt{3}}{2}&-\frac{1}{2}\end{pmatrix}%che corrisponde a \sigma_5^{(2)}, cioè a una rotazione di 120°, ossia ruoto antiorario per 2 vertici del triangolo, o ruoto in verso orario 1 vertice del triangolo
        ,\quad  x\mapsto
 \begin{pmatrix}
        1 & 0\\
        0 & -1
       \end{pmatrix}.
\eeq 
\end{itemize}

\begin{theorem}[\cite{our2nd}]\label{thrortheorpqub}
With the above notation, for every prime $p$, there exists at least a $2$-dimensional irrep of $G_p$,
\beq \label{eq:GpdaDp+1}
\rho_{2,1}\coloneqq \sigma_2\circ \upphi_2\ \ \textup{for}\ \ p=2,\qquad \rho_{p,1}^{(j)}\coloneqq \sigma_p^{(j)}\circ\upvarphi_p\circ \mathdutchcal{F}_p\ \ \textup{for}\ \ p\textup{ odd},
\eeq 
and there is at least a $p$-adic qubit,
\begin{align}
(\CC^2,\rho_2) &\ \ \textup{where}\ \  \rho_2\coloneqq \rho_{2,1}\circ\pi_1\colon \SO(3)_2\rightarrow\mathrm{O}(2),\quad \textup{for}\ \ p=2,\label{eq:2qubitirr}\\
(\CC^2,\rho_p^{(j)}) &\ \ \textup{where}\ \ \rho_p^{(j)}\coloneqq \rho_{p,1}^{(j)}\circ\pi_1\colon \SO(3)_p\rightarrow\mathrm{O}(2),\quad \textup{for}\ \ p\textup{ odd}, \label{eq:pqubitirr}
\end{align}
for $j=1,\dots,\frac{p-1}{2}$.
\end{theorem}
Note that for $p=2,3$ there exists a unique $p$-adic qubit arising from the irreps of $G_p$, while this is no longer the case for larger primes.

%%%%%%%%%%%%%%%%%%%%%%%%%%%%%%%%%%%%%%%%%%%%%%%%%%%%%%%%%%%%%%%%%%%%%%%%
\section{Clebsch-Gordan problem for \texorpdfstring{$p$}{Lg}-adic qubits}\label{CBdecocoedf}
We consider systems of more $p$-adic qubits, modelled by the tensor product of the corresponding two-dimensional irreps of $\SO(3)_p$. The \emph{Clebsch-Gordan problem} consists of identifying the %unitary 
transformation that decomposes the \emph{tensor-product representation} into a direct sum of irreps. The \emph{Clebsch-Gordan coefficients} mediate the change of basis between the uncoupled system (i.e., tensor product of individual spins) and the coupled system (with basis formed by the eigenstates of the total spin operator).

More generally, for $i$ ranging in a finite set, let ${\mathdutchcal{U}_{p}^{(i)}}\colon\SO(3)_p\rightarrow \U(\cH_i)$ be an irrep of $\SO(3)_p$. By Theorem~\ref{cor:fattrepNOP}, for every $i$, there exists $k_i\in\NN$ and an irrep ${\mathdutchcal{U}_{p,k_i}^{(i)}}\colon G_{p^{k_i}}\rightarrow \U(\cH_i)$ such that ${\mathdutchcal{U}_{p}^{(i)}} = {\mathdutchcal{U}_{p,k_i}^{(i)}}\circ\pi_{k_i}$. Take $k\coloneqq\max_i\{k_i\}$, so that ${\mathdutchcal{U}_{p,k}^{(i)}}\coloneqq {\mathdutchcal{U}_{p,k_i}^{(i)}}\circ\pi_{k_i,k}$ is an irrep of $G_{p^k}$ and
\beq 
\bigotimes_i{\mathdutchcal{U}_{p}^{(i)}}=\bigotimes_i \left({\mathdutchcal{U}_{p,k_i}^{(i)}}\circ\pi_{k_i}\right) = \bigotimes_i\left({\mathdutchcal{U}_{p,k}^{(i)}}\circ\pi_k\right) = \bigotimes_i{\mathdutchcal{U}_{p,k}^{(i)}}\circ\pi_k.
\eeq 
The \emph{Clebsch-Gordan decomposition} of this tensor-product representation is a direct sum of irreps of $G_{p^k}$. %Some of these representation could factorise through $G_{p^m}$ for $m<k$, however 
%From this decomposition, no representation arises for which $m>k$ is the minimum value such that it factorises through $G_{p^m}$. %Such a Clebsch-Gordan decomposition could be useful to find unknown representations for which $k$ is the minimum to factorise on $G_{p^k}$, once known only one such representation in tensor product with representations factorising on a ``lower level''.
%però io posso mette insieme il prod tensore di una che fattorizza a mod p, e una che fattorizza a mod p^2, e magari ne trovo di nuove che fattorizzano a mod p^2

\medskip 

In this section, we address the Clebsch-Gordan problem for the $p$-adic qubits in Eqs.~\eqref{eq:2qubitirr},~\eqref{eq:pqubitirr}, for every prime $p$. At given $p$, we consider a system of two $p$-adic qubits. This amounts to the Clebsch-Gordan problem for the tensor product of two two-dimensional irreps of the dihedral group $\mathrm{D}_n$, for $n\coloneqq p+1$. The latter is a finite group, hence we shall resort to character theory.

\medskip 

Given representations $\rho$ and $\sigma$ of a finite group $G$, with associated characters $\chi_\rho$ and $\chi_\sigma$ respectively, the character of the tensor-product representation $\rho\ox\sigma$ is given by the product of the single characters: $\chi_{\rho\ox\sigma}(g)\coloneqq \chi_\rho(g)\chi_\sigma(g)$ for $g\in G$. One defines the following inner product between characters $\chi_i$ and $\chi_j$ of representations of $G$: 
\beq \label{eq:prodcaratt}
(\chi_i,\chi_j)_G\coloneqq \frac{1}{\ord(G)}\sum_{k=1}^{\#\textup{CC}}\abs{C_k} \chi_i(C_k)\chi_j(C_k)^\ast,
\eeq 
where $C_k$ is a conjugacy class of $G$, $\#\textup{CC}$ denotes the number of conjugacy classes of $G$, and $\phantom{.}^\ast$ denotes the complex conjugation. If $\chi_i,\chi_j$ are irreducible, then $(\chi_i,\chi_j)_G=\delta_{ij}$ (\emph{first orthogonality relation}% between irreducible characters of $G$, i.e., between two rows of the character table of $G$
). In a direct-sum decomposition of the kind $\cH\ox \cH'\simeq \bigoplus_i\cH_i^{\oplus m_i}$, for every $i$, the multiplicity $m_i$ of an irreducible representation $\cH_i$ within $\cH\ox \cH'$ is given by
\beq \label{eq:prodcarattMUL}
m_i=\left(\chi_{\cH\ox \cH'},\,\chi_{\cH_i}\right)_G.
\eeq 

\medskip 

We develop two different treatments of the Clebsch-Gordan problem for $p=2$ and odd primes. We start with $p=2$. Recalling the irreps of $\mathrm{D}_3$ from Eq.~\eqref{eq:D3sigma2irrep} and above, we find the following character table.
\beq \notag
\begin{array}{ |c|c|c|c| } 
 \hline
\mathrm{D}_3%\simeq\mathrm{S}_3
&\ \{e\}\ & \{r,r^2\}%=(123),(132) tricicli
& \{x,xr,xr^2\}%=(12),(13),(23) scambi
\\  \hline
\chi_{\text{TRIV}} &1 & 1&1 \\  \hline
\chi_{\text{SIGN}} & 1 & 1&-1 \\  \hline
\chi_{\sigma_2} & 2 & -1&0 \\ 
 \hline
\end{array}
\eeq

\begin{proposition}\label{prop:clebsch2deco}
The Clebsch-Gordan decomposition of the second tensor power of the two-dimensional irrep of $\mathrm{D}_3$ is as follows, up to equivalences:
\beq 
\sigma_2^{\otimes2}\simeq \textup{TRIV}\oplus\textup{SIGN}\oplus\sigma_2.
\eeq 
\end{proposition}
\begin{proof}
The character $\chi_{\sigma_2^{\otimes2}}$ of $\sigma_2^{\otimes2}$ is given 
\beq
\chi_{\sigma_2}(e)^2=4,\quad \chi_{\sigma_2}(r)^2=1,\quad \chi_{\sigma_2}(x)^2=0.
\eeq 
We compute the multiplicity of each irrep of $\mathrm{D}_3$ in $\sigma_2^{\otimes2}$ through Eqs.~\eqref{eq:prodcaratt},~\eqref{eq:prodcarattMUL}: 
\begin{align}
&\left(\chi_{\sigma_2^{\otimes2}},\,\chi_{\text{TRIV}}\right)_{\mathrm{D}_3}=\left(\chi_{\sigma_2^{\otimes2}},\,\chi_{\text{SIGN}}\right)_{\mathrm{D}_3}=\frac{1}{6}(1\cdot4\cdot1^\ast+2\cdot1\cdot1^\ast+0)=1;\nonumber\\
&\left(\chi_{\sigma_2^{\otimes2}},\,\chi_{\sigma_2}\right)_{\mathrm{D}_3}=\frac{1}{6}(1\cdot4\cdot2^\ast+2\cdot1\cdot(-1)^\ast+0)=1.
\end{align}
In terms of the dimensions of the stable subspaces, we have $2\otimes 2\simeq1\oplus1\oplus 2$.
\end{proof}
The Clebsch-Gordan decomposition of a system of two $2$-adic qubits as in Theorem~\ref{thrortheorpqub} is
\beq \label{eq:clebscgorddeco2}
\rho_2^{\otimes2}\simeq (\textup{TRIV}\circ\upphi_2\circ\pi_1)\oplus (\textup{SIGN}\circ\upphi_2\circ\pi_1)\oplus (\sigma_2\circ\upphi_2\circ\pi_1).
\eeq

\medskip

Now we assume that $p$ is an odd prime. Here, we consider $\mathrm{D}_n$ with even $n\geq 4$, as in particular we will take $n=p+1$ for the Clebsch-Gordan problem of $p$-adic qubits. Then, $\mathrm{D}_n$ has $\frac{n+6}{2}$ conjugacy classes:
\begin{itemize}
    \item $\{e\}$ and $\{r^{n/2}\}$ of one element each;
    \item $\frac{n-2}{2}$ conjugacy classes of two elements each, i.e., $\{r,r^{n-1}\}$, $\{r^2,r^{n-2}\}$, \dots, $\{r^{\frac{n-2}{2}},r^{\frac{n+2}{2}}\}$;
    \item $\{x,xr^2,\dots, xr^{n-2}\}$ and $\{xr,xr^3,\dots, xr^{n-1}\}$ of $n/2$ elements each.
\end{itemize}
Recall also the irreps of $\mathrm{D}_n$ from Eq.~\eqref{eq:repnD4caz} and above. 
Then, the character table of $\mathrm{D}_n$ can be summarised as follows~\cite{Serre:reps}, for every even $n\geq4$, $k=1,\dots,n$, $j=1,\dots,\frac{n-2}{2}$.
\beq\notag
\begin{array}{ |c|c|c| } 
 \hline
\mathrm{D}_n &r^k &\ xr^k \\  \hline
\chi^\text{triv}_{\text{1D}}&1&1\\  \hline
\chi^{(1)}_{\text{1D}}&1&-1 \\  \hline
\chi^{(2)}_{\text{1D}}&(-1)^k&(-1)^k\\  \hline
\chi^{(3)}_{\text{1D}}&(-1)^k&\quad (-1)^{k+1}\quad \\  \hline
\chi_{n-1}^{(j)}&\quad 2\cos\left(\frac{2\pi jk}{n}\right)\quad &0 \\  \hline
\end{array}
\eeq
The character $\chi_{j,l}$ of a four-dimensional tensor-product representation $\sigma_{n-1}^{(j)}\ox \sigma_{n-1}^{(l)}$ is given by
\beq \label{eq:car2x2GNE}
\chi_{j,l}(r^k)=4\cos\left(\frac{2\pi jk}{n}\right)\cos\left(\frac{2\pi lk}{n}\right),\qquad \chi_{i,j}(xr^k)=0,
\eeq 
for $k=1,\dots,n$. Dimensionally, the Clebsch-Gordan decomposition of some $\sigma_{n-1}^{(j)}\ox \sigma_{n-1}^{(l)}$ is one of the following kinds: 
\begin{align*}
&2\ox 2\simeq 1\oplus1\oplus1\oplus1;\\
&2\ox 2\simeq 1\oplus1\oplus2;\\
&2\ox 2\simeq 2\oplus2.
\end{align*}
The multiplicity of each irrep of $\mathrm{D}_n$ in any $\sigma_{n-1}^{(j)}\ox \sigma_{n-1}^{(l)}$ is given by the following inner products between characters: for every $j,l,m=1,\dots,\frac{n-2}{2}$,
\begin{align}
\left(\chi_{j,l}, \chi^{\text{triv}}_{\text{1D}}\right)_{\mathrm{D}_n}&=\left(\chi_{j,l}, \chi^{(1)}_{\text{1D}}\right)_{\mathrm{D}_n}\nonumber\\
& =\frac{2}{n}\left(1+(-1)^{j+l}+2\sum_{k=1}^{(n-2)/2}\cos\left(\frac{2\pi jk}{n}\right)\cos\left(\frac{2\pi lk}{n}\right)\right);\label{eq:permultGENCB1}\\
\left(\chi_{j,l}, \chi^{(2)}_{\text{1D}}\right)_{\mathrm{D}_n}&=\left(\chi_{j,l}, \chi^{(3)}_{\text{1D}}\right)_{\mathrm{D}_n}\nonumber\\
& =\frac{2}{n}\left(1+(-1)^{j+l+n/2}+2\sum_{k=1}^{(n-2)/2}(-1)^k\cos\left(\frac{2\pi jk}{n}\right)\cos\left(\frac{2\pi lk}{n}\right)\right);\label{eq:permultGENCB2}\\
\left(\chi_{j,l}, \chi_{n-1}^{(m)}\right)_{\mathrm{D}_n}&=
\frac{4}{n}\left(1+(-1)^{j+l+m}+2\sum_{k=1}^{(n-2)/2}\cos\left(\frac{2\pi jk}{n}\right)\cos\left(\frac{2\pi lk}{n}\right) \cos\left(\frac{2\pi mk}{n}\right)\right).\label{eq:permultGENCB3}
\end{align}
\begin{remark}
In the Clebsch-Gordan decomposition of $\sigma_{n-1}^{(j)}\ox \sigma_{n-1}^{(l)}$, the one-dimensional irreps $\sigma_{\textup{triv}}^{\textup{1D}}$ and $\sigma_1^{\textup{1D}}$ always come in pair, as well as $\sigma_2^{\textup{1D}}$ with $\sigma_3^{\textup{1D}}$. 
\end{remark}

It is important to detect whether two tensor-product representations $\sigma_{n-1}^{(j)}\ox \sigma_{n-1}^{(l)}$ and $\sigma_{n-1}^{(m)}\ox \sigma_{n-1}^{(q)}$ are equivalent, i.e., they share the same character. According to Eq.~\eqref{eq:car2x2GNE}, this happens if and only if
\beq \label{eq:pqubequiviff}
\cos\left(\frac{2\pi jk}{n}\right)\cos\left(\frac{2\pi lk}{n}\right) = 
\cos\left(\frac{2\pi mk}{n}\right)\cos\left(\frac{2\pi qk}{n}\right),
\eeq 
for every $k=1,\dots, n$. %Setting $n=p+1$, this provides equivalent systems of two $p$-adic qubits. We shall decompose those tensor-product representations into a direct sum of irreps of $\mathrm{D}_{p+1}$, to find the Clebsch-Gordan decomposition and coefficients of a system of two $p$-adic qubits.

%\begin{remark}\label{rem:simmetria} Eq.~\eqref{eq:pqubequiviff} holds if $(j,l)=(m,q),(q,m)$, for every $j,l,m,q=1,\dots\frac{n-2}{2}$. Also, Eqs.~\eqref{eq:permultGENCB1},~\eqref{eq:permultGENCB2},~\eqref{eq:permultGENCB3} are symmetric in $j$ and $l$. \end{remark}

\begin{proposition}\label{prop:clebschpdeco}
Let $n\geq4$ be an even number. The Clebsch-Gordan decomposition of the tensor product of two two-dimensional irreps of $\mathrm{D}_n$ is as follows, up to equivalences. If $n/2$ is even, 
\beq \label{eq:p+1/42unica1}
{\sigma_{n-1}^{(n/4)}}^{\ox2}\simeq 
\sigma^{\textup{triv}}_{\textup{1D}}\oplus \sigma^{(1)}_{\textup{1D}} \oplus
\sigma^{(2)}_{\textup{1D}}\oplus \sigma^{(3)}_{\textup{1D}}.
\eeq 
For every $n\geq 6$, $j\in [1,n/4)\cap\NN$%ossia: When $n/2$ is odd, for every $j=1,\dots,\frac{n-2}{4}$, and when $n>4$, $n/2$ is even, for every $j=1,\dots,n/4-1$
, 
\beq \label{eq:quadratinoncentrali}
{\sigma_{n-1}^{(j)}}^{\ox2}\simeq {\sigma_{n-1}^{(n/2-j)}}^{\ox2}\simeq \sigma^{\textup{triv}}_{\textup{1D}}\oplus \sigma^{(1)}_{\textup{1D}}\oplus\sigma_{n-1}^{(2j)},
\eeq 
and %ossia j< l, non n/4 entrambi
\beq \label{prodsimmdalcentro}
\sigma_{n-1}^{(j)}\ox \sigma_{n-1}^{(n/2-j)}\simeq \sigma_{n-1}^{(n/2-j)} \ox \sigma_{n-1}^{(j)}\simeq \sigma_{\textup{1D}}^{(2)}\oplus \sigma_{\textup{1D}}^{(3)}\oplus \sigma_{n-1}^{(n/2-2j)}.
\eeq
For every $n\geq 8$, $j\in [1,n/4-1]\cap\NN$ and $j'=j'(j)=j+1,\dots,n/2-1-j$ at each $j$,
\begin{align}\label{ultimirinfusa}
&\sigma_{n-1}^{(j)}\ox \sigma_{n-1}^{(n/2-j')}\simeq \sigma_{n-1}^{(n/2-j')}\ox \sigma_{n-1}^{(j)}\nonumber\\
&
\simeq \sigma_{n-1}^{(j')}\ox \sigma_{n-1}^{(n/2-j)} \simeq \sigma_{n-1}^{(n/2-j)}\ox \sigma_{n-1}^{(j')} \nonumber\\
&\simeq \sigma_{n-1}^{(n/2-j-j')}\oplus \sigma_{n-1}^{(n/2+j-j')}.
\end{align}
\end{proposition}
\begin{proof}
Let $j,l=1,\dots,\frac{n-2}{2}$. By Eq.~\eqref{eq:pqubequiviff}, $\sigma_{n-1}^{(j)}\ox \sigma_{n-1}^{(l)}\simeq \sigma_{n-1}^{(l)}\ox \sigma_{n-1}^{(j)}$. Then, the number of unordered pairs $\{j,l\}$ locating a tensor product of two two-dimensional irreps of $\mathrm{D}_n$ up to commutation is the $2$-combination of $\frac{n-2}{2}$ elements with repetition, i.e., $\frac{n(n-2)}{8}$.

If $n/2$ is even, then
\beq \notag 
\left(\chi_{n/4,n/4},\ \chi^{\text{triv}}_{\text{1D}}\right)_{\mathrm{D}_n} =\left(\chi_{n/4,n/4},\ \chi^{(1)}_{\text{1D}}\right)_{\mathrm{D}_n} = \frac{4}{n}\left(1+\sum_{k=1}^{n/2-1}\cos\left(\frac{\pi k}{2}\right)^2\right) = \frac{4}{n}\left(1+\frac{n}{4}-1\right) = 1,
\eeq
and similarly $\left(\chi_{n/4,n/4}, \ \chi^{(2)}_{\text{1D}}\right)_{\mathrm{D}_n} =\left(\chi_{n/4,n/4}, \ \chi^{(3)}_{\text{1D}}\right)_{\mathrm{D}_n} =1$, since only the $n/4-1$ even indices $k$ in the sums in Eqs.~\eqref{eq:permultGENCB1},~\eqref{eq:permultGENCB2} provide non-zero contributions. Therefore, by dimensional analysis, Eq.~\eqref{eq:permultGENCB3} gives $\left(\chi_{n/4,n/4},\ \chi_{n-1}^{(m)}\right)_{\mathrm{D}_n}=0$ for every $m$, and Eq.~\eqref{eq:p+1/42unica1} is proved.

For every $j=1,\dots,\frac{n-2}{2}$, different from $n/4$ if $n/2$ is even, consider $j=l$ and $m=q=n/2-j$ in Eq.~\eqref{eq:pqubequiviff}: its right-hand side becomes $\cos\left(\pi k-\frac{2\pi jk}{n}\right)^2 = \cos\left(\frac{2\pi jk}{n}\right)^2$, which is equal to the left-hand side. This means that ${\sigma_{n-1}^{(j)}}^{\ox2}\simeq {\sigma_{n-1}^{(n/2-j)}}^{\ox2}$. Then, Eq.~\eqref{eq:permultGENCB1} yields
\begin{align*}
\left(\chi_{j,j},\ \chi^{\text{triv}}_{\text{1D}}\right)_{\mathrm{D}_n} & =\frac{4}{n}\left(1+\sum_{k=1}^{n/2-1}\cos\left(\frac{2\pi jk}{n}\right)^2\right) \\
& = 
\frac{4}{n}\left[1+\frac{1}{2}\sum_{k=1}^{n/2-1}\left(1+\cos\left(\frac{4\pi jk}{n}\right)\right)\right]\\
& = \frac{4}{n}\left(1+\frac{1}{2}\left(\frac{n}{2}-1\right)+\frac{1}{2}(-1)\right)= 1,
\end{align*}
where, by defining $\omega_j\coloneqq \exp\left(i\frac{4\pi j}{n}\right)$, we compute 
$\sum_{k=1}^{n/2-1}\cos\left(\frac{4\pi jk}{n}\right)=\Re\sum_{k=1}^{n/2-1}\omega_j = \Re\frac{1-\omega_j^{n/2}}{1-\omega_j}-1=-1$. 
\begin{comment}
Similarly, Eq.~\eqref{eq:permultGENCB2} for even $n/2$ gives
\begin{align*}
\left(\chi_{n/4,n/4},\ \chi^{(2)}_{\text{1D}}\right)_{\mathrm{D}_n} & =\frac{4}{n}\left(1+\sum_{k=1}^{n/2-1}(-1)^k\cos\left(\frac{2\pi jk}{n}\right)^2\right) \\
& = 
\frac{4}{n}\left[1+\frac{1}{2}\sum_{k=1}^{n/2-1}(-1)^k\left(1+\cos\left(\frac{4\pi jk}{n}\right)\right)\right]\\
& = \frac{4}{n}\left(1-\frac{1}{2}-\frac{1}{2}\right)= 0,
\end{align*}
as well as for odd $n/2$, 
\begin{align*}
\left(\chi_{n/4,n/4},\ \chi^{(2)}_{\text{1D}}\right)_{\mathrm{D}_n} & =\frac{4}{n}\sum_{k=1}^{n/2-1}(-1)^k\cos\left(\frac{2\pi jk}{n}\right)^2 \\
& = \frac{4}{n}\left(0+0\right)= 0,
\end{align*}
\end{comment}
Similarly, Eq.~\eqref{eq:permultGENCB3} with $j=l$ and $m=2j$ gives
\begin{align*}
\left(\chi_{j,j}, \chi_{n-1}^{(2j)}\right)_{\mathrm{D}_n} & = \frac{8}{n}\left(1+\sum_{k=1}^{n/2-1}\cos\left(\frac{2\pi jk}{n}\right)^2\cos\left(\frac{4\pi jk}{n}\right)\right)\\
& = \frac{8}{n}\left(1+\frac{1}{2}\sum_{k=1}^{n/2-1}\cos\left(\frac{4\pi jk}{n}\right)+\frac{1}{4}\sum_{k=1}^{n/2-1}1+\frac{1}{4}\sum_{k=1}^{n/2-1}\cos\left(\frac{8\pi jk}{n}\right)\right)\\
& =  \frac{8}{n}\left(1-\frac{1}{2}+\frac{n-2}{8}-\frac{1}{4}\right)=1.
\end{align*}
This shows Eq.~\eqref{eq:quadratinoncentrali}, by considering $j<n/4$ so that $2j$ locates a two-dimensial irrep of $\mathrm{D}_n$%2j\leq (n-2)/2 sse j<n/4
, and $n>4$ accordingly. This exhausts all the possible tensor products of a two-dimensional irrep of $\mathrm{D}_n$ with itself, which are $\frac{n-2}{2}$ in number.

Putting $l=n/2-j$ in Eq.~\eqref{eq:permultGENCB2}, we obtain $\left(\chi_{j,n/2-j}, \chi^{(2)}_{\text{1D}}\right)_{\mathrm{D}_n}=\left(\chi_{j,n/2-j}, \chi^{(3)}_{\text{1D}}\right)_{\mathrm{D}_n}=1$. Then, Eq.~\eqref{eq:permultGENCB3} leads to $\left(\chi_{j,n/2-j}, \chi_{n-1}^{(n/2-2j)}\right)_{\mathrm{D}_n}=1$, with $j\neq n/4$ if $n/2$ is even. This shows Eq.~\eqref{prodsimmdalcentro}, again bounding $j<n/4$ so that $n/2-2j$ locates a two-dimensional irrep of $\mathrm{D}_n$%n/2-2j\geq1 iff j<n/4, whenever $n>4$
. This exhausts all the possible unordered pairs $\{j,l\}$ locating $\sigma_{n-1}^{(j)}\ox \sigma_{n-1}^{(l)}$ such that $j+l=n/2$, $j< l$%ossia j\neq n/4 quando n/2 pari... sono tutte le coppie simmetriche rispetto al centro n/4 degli indici
. They are $\frac{n-2}{4}$ when $n/2$ is odd, and $n/4-1$ when %$n>4$,
$n/2$ is even.

The left-out unordered pairs $\{j,l\}$ are in number $\frac{n(n-2)}{8}-\frac{n-2}{2}-\frac{n-2}{4} = \frac{(n-2)(n-6)}{8}$ when $n/2$ is odd, and $\frac{n(n-2)}{8}-\frac{n-2}{2}-\frac{n}{4}+1 = \frac{(n-4)^2}{8}$ when $n/2$ is even. They are $\{j,n/2-j'\}$ and $\{j',n/2-j\}$ for every $n\geq8$, $j\in[1,n/4-1]\cap\NN$ %ossia j=1,...,n/4-1 per n/2 pari, e j=1,...,(n-6)/4 per n/2 dispari
and $j'=j'(j)=j+1,\dots,n/2-1-j$ at each $j$. Here, the left-hand side of Eq.~\eqref{eq:pqubequiviff} with $l=n/2-j'$ becomes $(-1)^k\cos\left(\frac{2\pi jk}{n}\right)$, equal to the right-hand side with $m=j'$, $q=n/2-j$; therefore $\sigma_{n-1}^{(j)}\ox \sigma_{n-1}^{(n/2-j')}\simeq \sigma_{n-1}^{(n/2-j')}\ox \sigma_{n-1}^{(j)}
\simeq 
\sigma_{n-1}^{(j')}\ox \sigma_{n-1}^{(n/2-j)}\simeq \sigma_{n-1}^{(n/2-j)}\ox \sigma_{n-1}^{(j')}$. Eq.~\eqref{eq:permultGENCB3} with $l=n/2-j'$ and $m=n/2-j-j'$ gives
{\small \begin{align*}
\left(\chi_{j,n/2-j'}, \chi_{n-1}^{(n/2-j-j')}\right)_{\mathrm{D}_n}%&= \frac{4}{n}\left(1+(-1)^{n-2j'}+2\sum_{k=1}^{(n-2)/2}\cos\left(\frac{2\pi jk}{n}\right)\cos\left(\pi k-\frac{2\pi j'k}{n}\right) \cos\left(\pi k-\frac{2\pi (j+j')k}{n}\right)\right)\\
& = \frac{4}{n}\left(1+(-1)^{n-2j'}+2\sum_{k=1}^{(n-2)/2}\cos\left(\frac{2\pi jk}{n}\right)\cos\left(\frac{2\pi j'k}{n}\right) \cos\left(\frac{2\pi (j+j')k}{n}\right)\right)\\
& = %prodotto di coseni di \alpha_i è somma di coseni di permutazioni di segni della somma degli \alpha_i
\frac{8}{n}\left[1+\frac{1}{4}\sum_{k=1}^{(n-2)/2}\left(1 + \cos\left(\frac{4\pi jk}{n}\right) + \cos\left(\frac{4\pi j'k}{n}\right) + \cos\left(\frac{4\pi (j+j')k}{n}\right)\right)\right]\\
& = \frac{8}{n}\left[1+\frac{1}{4}\left(\frac{n-2}{2}-3\right)\right]=1.
\end{align*}}Similarly, one gets $\left(\chi_{j,n/2-j'}, \chi_{n-1}^{(n/2+j-j')}\right)_{\mathrm{D}_n}=1$. This proves Eq.~\eqref{ultimirinfusa}.

We treated all the possible tensor products of two two-dimensional irreps of $\mathrm{D}_n$. From the direct sum of irreps in Eqs.~\eqref{eq:p+1/42unica1},~\eqref{eq:quadratinoncentrali},~\eqref{prodsimmdalcentro},~\eqref{ultimirinfusa}, there are no other equivalences among tensor products of two two-dimensional irreps of $\mathrm{D}_n$: in Eq.~\eqref{eq:p+1/42unica1} there is the only $1\oplus1\oplus1\oplus1$ decomposition, in Eqs.~\eqref{eq:quadratinoncentrali},~\eqref{prodsimmdalcentro} there are the only $1\oplus 1\oplus2$ decompositions with different one-dimensional irreps, and any of the tensor-product representations in Eq.~\eqref{ultimirinfusa} (decomposed as $2\oplus 2$) located by $j,j'$ is equivalent to one located by $k,k'$ if and only if $n/2\pm j-j'=n/2\pm k-k'$ or $n/2\pm j-j'=n/2\mp k-k'$, respectively if and only if $j=k,\,j'=k'$ (same representation) or $j=-k,\,j'=k'$ (impossible).
\end{proof}

The Clebsch-Gordan decompositions of two $p$-adic qubits as in Theorem~\ref{thrortheorpqub} are written similarly to Eq.~\eqref{eq:clebscgorddeco2} for all the possibilities in Eqs.~\eqref{eq:p+1/42unica1},~\eqref{eq:quadratinoncentrali},~\eqref{prodsimmdalcentro},~\eqref{ultimirinfusa} for $n=p+1$.

\begin{remark}
Since the irreps of $\SO(3)_p$ we consider are complex (not $p$-adic), states of qubits and $p$-adic qubits are vector rays in $\CC^2$. Possible bases spanning the state space $\CC^2\otimes \CC^2\simeq \CC^4$ of two qubits are the canonical basis $\big(\ket{00},\,\ket{01},\,\ket{10},\,\ket{11}\big)$ or the basis of the four (\emph{maximally entangled}) \emph{Bell states}, 
\beq \label{eq:Bellbasis}
\big(\ket{\phi^+},\ket{\phi^-},\ket{\psi^+},\ket{\psi^-}\big),
\eeq 
where
\beq \label{eq:Bellpsiphi}
\ket{\phi^\pm}\coloneqq\frac{1}{\sqrt{2}}(\ket{00}\pm\ket{11}),\qquad \ket{\psi^\pm}\coloneqq
\frac{1}{\sqrt{2}}(\ket{01}\pm\ket{10}).
\eeq
In standard quantum mechanics~\cite{sakurai}, The Clebsch-Gordan decomposition of the tensor product of two qubits is given in terms of stable subspaces as
\beq \label{eq:singltripl+}
2\otimes 2\simeq 1\oplus 3.
\eeq 
This means that the second tensor power of the unique two-dimensional irrep of $\SU(2)$ decomposes into the direct sum of the unique one-dimensional irrep with the unique three-dimensional irrep of $\SU(2)$, up to equivalences. This is realised by the change of basis from the canonical one to
\beq \label{eq:singltripl}
\big(\ket{\psi^-},\ \ket{00},\ket{\psi^+}, \ket{11}\big).
\eeq 
The vector $\ket{\psi^-}$ spans the one-dimensional subrepresentation (associated to a spin-$0$ particle) and is called \emph{singlet} state, while $\big(\ket{00},\ket{\psi^+}, \ket{11}\big)$ spans the three-dimensional subrepresentation (associated to a spin-$1$ particle) and is called \emph{triplet}. The matrix of the Clebsch-Gordan coefficients realising this basis change for the composite system of two qubits is
\beq\label{ClebschGordan}
\operatorname{T}=\begin{pmatrix}
  1&0&0&0\\0&\frac{1}{\sqrt{2}}&0&\frac{1}{\sqrt{2}}\\0&\frac{1}{\sqrt{2}}&0&-\frac{1}{\sqrt{2}}\\0&0&1&0
\end{pmatrix}.
\eeq 
Lastly, the Clebsch-Gordan decomposition of the tensor powers of the qubit representation gives rise to all the other irreps of $\SU(2)$~\cite{Weizsaecker}.
\end{remark}

\subsection{Clebsch-Gordan coefficients}\label{subsec:CBcpeff2357}
We explicitly calculate the Clebsch-Gordan coefficients of the Clebsch-Gordan decompositions of systems of two $p$-adic qubits corresponding to Propositions~\ref{prop:clebsch2deco} and~\ref{prop:clebschpdeco} for the smallest 
emblematic %p=2 sempre a parte, p=3 =3(4), p=5=1(4), p=7 il primo in cui capita 2x2=2+2, così copro tutte le possibilità
primes $p=2,3,5,7$ (cf.~\cite[Ch.~7]{PhDtesi} for complete calculations). For the general treatment, for every odd prime $p$, we will look at the structural nature of the coupled bases realising the various direct-sum decompositions (see next Section~\ref{sec:entangldectp}).
%the change of basis required to decompose the tensor product corresponds

\medskip

$\boldsymbol{p=2}$\\
According to Proposition~\ref{prop:clebsch2deco}, $2\otimes 2\simeq 1\oplus 1\oplus 2$, and we find a basis for $\CC^4$ such that $\sigma_2^{\otimes2}$ stabilises two lines and a two-dimensional subspace (which are pairwise orthogonal), acting on them isomorphically to $\textup{TRIV},\,\textup{SIGN},\,\sigma_2$ respectively. This means to find a basis $\big(\mathbf{v},\mathbf{w},\mathbf{y},\mathbf{z}\big)$ such that
\beq\notag  \begin{array}{ll}
    \sigma_2^{\otimes2}(r)(\mathbf{v})=\mathbf{v},&\sigma_2^{\otimes2}(x)(\mathbf{v})=\mathbf{v},\\
    \sigma_2^{\otimes2}(r)(\mathbf{w})=\mathbf{w},&\sigma_2^{\otimes2}(x)(\mathbf{w})=-\mathbf{w},\\
    \sigma_2^{\otimes2}(r)(\mathbf{y})=-\frac{1}{2}\mathbf{y}+\frac{\sqrt{3}}{2}\mathbf{z},\quad\quad &\sigma_2^{\otimes2}(x)(\mathbf{y})=\mathbf{y},\\
    \sigma_2^{\otimes2}(r)(\mathbf{z})=-\frac{\sqrt{3}}{2}\mathbf{y}-\frac{1}{2}\mathbf{z},& \sigma_2^{\otimes2}(x)(\mathbf{z})=-\mathbf{z}.
\end{array}\eeq 
%esattamente come i quadrati in p=5
With respect to the canonical basis, these conditions give
\beq 
\mathbf{v}=\begin{pmatrix}
  v_1\\0\\0\\v_1
\end{pmatrix},\quad \mathbf{w}=\begin{pmatrix}
  0\\w_2\\-w_2\\0
\end{pmatrix},\quad \mathbf{y}=\begin{pmatrix}
  y_1\\0\\0\\-y_1
\end{pmatrix},\quad \mathbf{z} = \begin{pmatrix}
  0\\-y_1\\-y_1\\0
\end{pmatrix},
\eeq 
for $v_1,w_2,y_1\in \CC$. Normalising these vectors to $1$, and ignoring global phases, the new basis $\big(\mathbf{v},\mathbf{w},\mathbf{y},\mathbf{z}\big)$ with respect to the canonical one of $\CC^2\otimes\CC^2$ is
\beq \label{Bellbasis2}
\big(\ket{\phi^+},\,\ket{\psi^-},\ket{\phi^-},%-
\ket{\psi^+}\big),
\eeq 
that is, some ordering of the Bell basis~\eqref{eq:Bellbasis}. The first two are \emph{singlet} states, while the last two form a \emph{doublet}. The matrix of the associated change of basis $\big(\ket{00},\,\ket{01},\,\ket{10},\,\ket{11}\big)\mapsto \big(\ket{\phi^+},\,\ket{\psi^-},\ket{\phi^-},%-
\ket{\psi^+}\big)$ is
\beq \label{2CBcoeff}
\operatorname{T}_2\coloneqq\begin{pmatrix}
  \frac{1}{\sqrt{2}}&0&\frac{1}{\sqrt{2}}&0\\0&\frac{1}{\sqrt{2}}&0&%-
  \frac{1}{\sqrt{2}}\\0&-\frac{1}{\sqrt{2}}&0&%-
  \frac{1}{\sqrt{2}}\\\frac{1}{\sqrt{2}}&0&-\frac{1}{\sqrt{2}}&0
\end{pmatrix}.
\eeq 
\begin{proposition}
The matrix elements of $\operatorname{T}_2$ in Eq.~\eqref{2CBcoeff} are the Clebsch-Gordan coefficients for two $2$-adic qubits $\big((\CC^2)^{\otimes2},(\rho_2)^{\otimes2}\big)$ with respect to the canonical basis.
\end{proposition}

\medskip

$\boldsymbol{p=3}$\\
There is a unique $3$-adic qubit arising from $G_3$, given by the irrep $\rho_3\coloneqq \rho_3^{(1)}$ as in Theorem~\ref{thrortheorpqub}. This writes as $\rho_3=\rho_{3,1}\circ\pi_1=\sigma_3\circ\upvarphi_3\circ \mathdutchcal{F}_3\circ\pi_1$, where $\rho_{3,1}\coloneqq\rho_{3,1}^{(1)}$ and $\sigma_3\coloneqq \sigma_3^{(1)}$ are respectively the unique two-dimensional irreps of $G_3$ and $\mathrm{D}_4$. As a consequence of Proposition~\ref{prop:clebschpdeco}, the Clebsch-Gordan decomposition of a system of two $3$-adic qubits as in Theorem~\ref{thrortheorpqub} is
\beq 
\rho_3^{\otimes2}\simeq (\mathdutchcal{triv}\circ\pi_1)\oplus (\mathdutchcal{s}\circ\pi_1)\oplus (\mathdutchcal{t}\circ\pi_1)\oplus(\mathdutchcal{st}\circ\pi_1),
\eeq 
where the one-dimensional irreps of $G_3$ are as in Eq.~\eqref{1dirrepsDp+1Gp} (cf.\ also~\cite[Prop.~V.3]{our2nd}). This decomposition is of the kind $2\otimes 2\simeq 1\oplus1\oplus1\oplus1$. For the Clebsch-Gordan coefficients, we set up a similar calculation to that done for $p=2$.
\begin{comment}
we need to find a basis for the complex four-dimensional vector space, such that $\sigma_3^{\otimes2}$ stabilises the lines through each basis vector, acting on them isomorphically to $\sigma_\text{triv}^{\text{1D}},\sigma_1^{\text{1D}},\sigma_2^{\text{1D}},\sigma_3^{\text{1D}}$ respectively. This means to find a basis $\big(\mathbf{v},\mathbf{w},\mathbf{y},\mathbf{z}\big)$ such that
\beq\notag  \begin{array}{ll}
    \sigma_3^{\otimes2}(r)(\mathbf{v})=\mathbf{v},&\sigma_3^{\otimes2}(x)(\mathbf{v})=\mathbf{v},\\
    \sigma_3^{\otimes2}(r)(\mathbf{w})=\mathbf{w},&\sigma_3^{\otimes2}(x)(\mathbf{w})=-\mathbf{w},\\
    \sigma_3^{\otimes2}(r)(\mathbf{y})=-\mathbf{y},\quad\quad &\sigma_3^{\otimes2}(x)(\mathbf{y})=\mathbf{y},\\
    \sigma_3^{\otimes2}(r)(\mathbf{z})=-\mathbf{z},& \sigma_3^{\otimes2}(x)(\mathbf{z})=-\mathbf{z}.
\end{array}\eeq 
With respect to the canonical basis, these conditions give
\beq 
\mathbf{v}=\begin{pmatrix}
  v_1\\0\\0\\v_1
\end{pmatrix},\quad \mathbf{w}=\begin{pmatrix}
  0\\w_2\\-w_2\\0
\end{pmatrix},\quad \mathbf{y}=\begin{pmatrix}
  y_1\\0\\0\\-y_1
\end{pmatrix},\quad \mathbf{z}=\begin{pmatrix}
  0\\z_2\\z_2\\0
\end{pmatrix},
\eeq 
for $v_1,w_2,y_1,z_2\in \CC$. Normalising these vectors to $1$, and choosing $v_1=w_2=y_1=z_2=\frac{1}{\sqrt{2}}$, the new basis $\big(\mathbf{v},\mathbf{w},\mathbf{y},\mathbf{z}\big)$ 
\end{comment}
The new basis with respect to the canonical one of $\CC^2\otimes\CC^2$ is
\beq 
\big(\ket{\phi^+},\,\ket{\psi^-},\ket{\phi^-},\ket{\psi^+}\big),
\eeq 
as in Eq.~\eqref{Bellbasis2} for $p=2$%in realtà là l'ultimo vettore deve avere segno meno
. However, here all four states are \emph{singlet} states. 
\begin{comment}
The matrix of the associated change of basis $\big(\ket{00},\,\ket{01},\,\ket{10},\,\ket{11}\big)\mapsto \big(\ket{\phi^+},\,\ket{\psi^-},\ket{\phi^-},\ket{\psi^+}\big)$ is
\beq \label{3CBcoeff}
\operatorname{T}_3=\begin{pmatrix}
  \frac{1}{\sqrt{2}}&0&\frac{1}{\sqrt{2}}&0\\0&\frac{1}{\sqrt{2}}&0&\frac{1}{\sqrt{2}}\\0&-\frac{1}{\sqrt{2}}&0&\frac{1}{\sqrt{2}}\\\frac{1}{\sqrt{2}}&0&-\frac{1}{\sqrt{2}}&0
\end{pmatrix}.
\eeq 
%equal to $\operatorname{T}_2$. %in realtà l'ultima colonna di T2 e T3 sarebbero con segno scambiato
\end{comment}
\begin{proposition}
The matrix elements of $\operatorname{T}_3\coloneqq\operatorname{T}_2$ in Eq.~\eqref{2CBcoeff} are the Clebsch-Gordan coefficients for two $3$-adic qubits $\big((\CC^2)^{\otimes2},(\rho_3)^{\otimes2}\big)$ with respect to the canonical basis.
\end{proposition}

\medskip 

\boldsymbol{$p=5$}\\
According to Theorem~\ref{thrortheorpqub}, there are two inequivalent $5$-adic qubits, arising from the irreps $\rho_5^{(1)}$ and $\rho_5^{(2)}$ of $G_5$. As a consequence of Proposition~\ref{prop:clebschpdeco}, the Clebsch-Gordan decompositions of systems of two $5$-adic qubits as in Theorem~\ref{thrortheorpqub} are
\begin{align}
&(\rho_5^{(1)})^{\otimes2}\simeq (\rho_5^{(2)})^{\otimes2}\simeq(\mathdutchcal{triv}\circ\pi_1)\oplus(\mathdutchcal{s}\circ\pi_1)\oplus \rho_5^{(2)},\label{eq:5qub1}\\ 
&\rho_5^{(1)}\otimes \rho_5^{(2)}\simeq \rho_5^{(2)}\otimes \rho_5^{(1)}\simeq \rho_5^{(1)}\oplus(\mathdutchcal{t}\circ \pi_1)\oplus (\mathdutchcal{st}\circ\pi_1),\label{eq:5qub2}
\end{align} 
where the one-dimensional irreps of $G_5$ are as in Eq.~\eqref{1dirrepsDp+1Gp} (cf.\ also~\cite[Prop.~VI.3]{our1st}). In terms of stable subspaces, Eq.~\eqref{eq:5qub1} corresponds to $2\otimes 2\simeq1\oplus1\oplus 2$ and similarly Eq.~\eqref{eq:5qub2} to $2\otimes 2\simeq2\oplus 1\oplus1$. We set up a similar calculation to those done for $p=2,3$ for the Clebsch-Gordan coefficients of both Eq.~\eqref{eq:5qub1} and Eq.~\eqref{eq:5qub2}. We find that the basis realising the decompositions of $(\rho_5^{(1)})^{\otimes2}$ and $\rho_5^{(1)}\otimes \rho_5^{(2)}$ is
\beq
\big(\ket{\phi^+},\,\ket{\psi^-},\ket{\phi^-},%-
\ket{\psi^+}\big),
\eeq 
as for $p=2,3$. Here these states are coming in two \emph{singlets} and a \emph{doublet}.
\begin{proposition}
The matrix elements of $\operatorname{T}_5\coloneqq\operatorname{T}_2$ in Eq.~\eqref{2CBcoeff} are the Clebsch-Gordan coefficients for two $5$-adic qubits $\big((\CC^2)^{\otimes2},(\rho_5^{(1)})^{\otimes2}%\simeq (\rho_5^{(2)})^{\otimes2}
\big)$ or $\big((\CC^2)^{\otimes2},\rho_5^{(1)}\otimes \rho_5^{(2)}%\simeq \rho_5^{(2)}\otimes \rho_5^{(1)}
\big)$ with respect to the canonical basis.
\end{proposition}

\medskip 

\boldsymbol{$p=7$}\\
According to Theorem~\ref{thrortheorpqub}, there are three inequivalent $7$-adic qubits, arising from the irreps $\rho_7^{(1)}$, $\rho_7^{(2)}$ and $\rho_7^{(3)}$ of $G_7$. As a consequence of Proposition~\ref{prop:clebschpdeco}, the Clebsch-Gordan decompositions of systems of two $7$-adic qubits as in Theorem~\ref{thrortheorpqub} are
%\begin{align} & {\sigma_{7}^{(2)}}^{\ox2}\simeq \sigma^{\textup{triv}}_{\textup{1D}}\oplus \sigma^{(1)}_{\textup{1D}} \oplus \sigma^{(2)}_{\textup{1D}}\oplus \sigma^{(3)}_{\textup{1D}},\\ & {\sigma_{7}^{(1)}}^{\ox2}\simeq {\sigma_{7}^{(3)}}^{\ox2}\simeq \sigma^{\textup{triv}}_{\textup{1D}}\oplus \sigma^{(1)}_{\textup{1D}}\oplus\sigma_{7}^{(2)},\\ &\sigma_{7}^{(1)}\ox \sigma_{7}^{(3)}\simeq \sigma_{7}^{(3)} \ox \sigma_{7}^{(1)}\simeq \sigma_{\textup{1D}}^{(2)}\oplus \sigma_{\textup{1D}}^{(3)}\oplus \sigma_{7}^{(2)},\\ & \sigma_{7}^{(1)}\ox \sigma_{7}^{(2)}\simeq \sigma_{7}^{(2)}\ox \sigma_{7}^{(1)} \simeq \sigma_{7}^{(2)}\ox \sigma_{7}^{(3)} \simeq \sigma_{7}^{(3)}\ox \sigma_{7}^{(2)}\simeq \sigma_{7}^{(1)}\oplus \sigma_{7}^{(3)}. \end{align}
\begin{align}
& \big(\rho_{7}^{(2)}\big)^{\ox2}\simeq 
(\mathdutchcal{triv}\circ\pi_1)\oplus (\mathdutchcal{s}\circ\pi_1) \oplus
(\mathdutchcal{t}\circ\pi_1)\oplus (\mathdutchcal{st}\circ\pi_1),\label{eq:7prima}\\
& \big(\rho_{7}^{(1)}\big)^{\ox2}\simeq \big(\rho_{7}^{(3)}\big)^{\ox2}\simeq (\mathdutchcal{triv}\circ\pi_1)\oplus (\mathdutchcal{s}\circ\pi_1)\oplus\rho_{7}^{(2)},\label{eq:7seconda}\\
&\rho_{7}^{(1)}\ox \rho_{7}^{(3)}\simeq \rho_{7}^{(3)} \ox \rho_{7}^{(1)}\simeq \rho_{7}^{(2)} \oplus (\mathdutchcal{t}\circ\pi_1)\oplus (\mathdutchcal{st}\circ\pi_1),\label{eq:7terza}\\
& \rho_{7}^{(1)}\ox \rho_{7}^{(2)}\simeq \rho_{7}^{(2)}\ox \rho_{7}^{(1)} 
\simeq \rho_{7}^{(2)}\ox \rho_{7}^{(3)} \simeq \rho_{7}^{(3)}\ox \rho_{7}^{(2)}\simeq \rho_{7}^{(1)}\oplus \rho_{7}^{(3)}.\label{eq:7quarta}
\end{align}
In terms of stable subspaces, Eq.~\eqref{eq:7prima} corresponds to $2\otimes 2\simeq 1\oplus1\oplus1\oplus1$, Eq.~\eqref{eq:7seconda} to $2\otimes 2\simeq 1\oplus1\oplus2$, similarly Eq.~\eqref{eq:7terza} to $2\otimes 2\simeq 2\oplus 1\oplus1$, and Eq.~\eqref{eq:7quarta} to $2\otimes 2\simeq 2\oplus2$. We set up a similar calculation to those done for $p=2,3,5$ for the Clebsch-Gordan coefficients of each of cases above. Again, we find that the basis with respect to which all these decompositions are realised is
\beq 
\big(\ket{\phi^+},\,\ket{\psi^-},\ket{\phi^-},%-
\ket{\psi^+}\big),
\eeq 
as for $p=2,3,5$.
\begin{proposition}
The matrix elements of $\operatorname{T}_7\coloneqq\operatorname{T}_2$ in Eq.~\eqref{2CBcoeff} are the Clebsch-Gordan coefficients for two $7$-adic qubits 
$\big((\CC^2)^{\otimes2},(\rho_7^{(2)})^{\otimes2}\big)$ or $\big((\CC^2)^{\otimes2},(\rho_7^{(1)})^{\otimes2}%\simeq(\rho_7^{(1)})^{\otimes2}
\big)$ or $\big((\CC^2)^{\otimes2},\rho_7^{(1)}\otimes \rho_7^{(3)}%\simeq \rho_7^{(3)}\otimes \rho_7^{(1)}
\big)$ or $\big((\CC^2)^{\otimes2},\rho_7^{(1)}\otimes \rho_7^{(2)}%\simeq \rho_7^{(2)}\otimes \rho_7^{(1)}\simeq \rho_7^{(2)}\otimes \rho_7^{(3)}\simeq \rho_7^{(3)}\otimes \rho_7^{(2)}
\big)$ with respect to the canonical basis.
\end{proposition}

%%%%%%%%%%%%%%%%%%%%%%%%%%%%%%%%%%%%%%%%%%%%%%%%%%%%%%%%%%%%%%%%%%%
\section{Entanglement for two $p$-adic qubits}\label{sec:entangldectp}
Here we aim at showing how entangled states appear from the Clebsch-Gordan problem of two $p$-adic qubits, compared to standard quantum mechanics. Indeed, entanglement is precious resource in quantum information and computation in several scenarios~\cite{oroent}. We investigate the structural nature of the basis vector states that realise a direct-sum decomposition of a tensor product of two $p$-adic qubits, in the various possible cases: $2\otimes 2\simeq 1\oplus 1\oplus 1\oplus 1$, $2\otimes 2\simeq 1\oplus 1\oplus 2$, $2\otimes 2\simeq 1\oplus 3$ or  $2\otimes 2\simeq 2\oplus 2$. We then study the quantum states corresponding to each of these subsystem.

Note that the decomposition $2\otimes 2\simeq 1\oplus 3$ (familiar from standard quantum mechanics) never appears in Section~\ref{CBdecocoedf}, as no three-dimensional irrep of $\SO(3)_p$ factorises through $G_p$ (cf.\ Subsection~\ref{subsec:strirrpsGp}). %magari ci sono irreps 3-dim di SO(3)_p o mai? e se sì, ottengo mai questa 3+1 deco standard? ne avrò mai di veramente complesse???
Contrarily to representation theory of $\SO(3)\simeq \SU(2)/\{\pm\mI\}$ for standard quantum mechanics, there are multiple $p$-adic qubit representations of $\SO(3)_p$, increasing in number as $p$ grows, so that entangled states can arise more often than in the unique standard Clebsch-Gordan decomposition~\eqref{eq:singltripl}. However, besides the singlet quantum states, doublets and triplets are given by separable quantum states.

\medskip

We recall the definition of entanglement in quantum mechanics, holding for both the standard and our $p$-adic frameworks, since the considered representations are over complex Hilbert spaces.
\begin{definition}
A pure (vector) state of a bipartite quantum system $\ket{\psi_{AB}}\in \mathcal{H}_A\otimes\mathcal{H}_B$ is said to be an {\it entangled state} if it is not a tensor product, i.e., if there do not exist $\ket{\psi_A}\in\mathcal{H}_A,\,\ket{\psi_B}\in\mathcal{H}_B$ such that $\ket{\psi_{AB}}=\ket{\psi_A}\otimes \ket{\psi_B}$.
\end{definition}
The \emph{Schmidt rank} of a state is defined as the number of non-vanishing coefficients appearing in the Schmidt decomposition of a bipartite pure state, writing the latter as a superposition of biorthogonal product states weighted by non-negative real coefficients via singular value decomposition. If a state has Schmidt rank equal to $1$, it is separable; if the Schmidt rank is larger, the state is entangled. Then, Bell states (as in Eq.~\eqref{eq:Bellpsiphi}) are classified as maximally entangled, because they achieve the maximum Schmidt rank $2$ for a two-qubit system with uniform singular values.

More generally, a mixed quantum state $\rho_{AB}$ in the space of density operators on $\cH_A\otimes \cH_B$ is called \emph{separable} if it can be written as a convex combination $\rho_{AB}=\sum_kp_k\rho_A^{(k)}\otimes \rho_B^{(k)}$, where $\rho_A^{(k)}$ and $\rho_B^{(k)}$ are density operators on $\cH_A$ and $\cH_B$ respectively (factorable quantum states are a subset of the separable ones); $\rho_{AB}$ is called \emph{entangled} if it is not separable. The \emph{PPT criterion} (positive partial transpose, cf.\ e.g.~\cite[Theorem $4.4.1$]{manwin}) tells that a necessary condition for the separability of a bipartite quantum state is the non-negativity of its partial transposition with respect to any of the two subsystems (it is also sufficient for two qubits). A quantum state on $\cH$ is \emph{maximally mixed} if its rank equals $\dim(\cH)$. The reduced density operator of a maximally entangled state is the maximally mixed state. For other basic notions in quantum information theory%(e.g. about mixed states)
, the reader may refer to~\cite{manwin,Wilde}. %QICnielchu

%For a unitary representation, the irreducible subrepresentations providing a direct-sum decomposition are pairwise orthogonal. This applies, in particular, for the stable subspaces in the list above.

\begin{proposition}
\label{prop:1irrmaxent}
%Every one-dimensional subrepresentation of the tensor product of two two-dimensional irreducible unitary representations of $SO(3)_p$ is spanned by a maximally entangled state. 
Let $\mathdutchcal{U}_A,\mathdutchcal{V}_B\colon \SO(3)_p\rightarrow \U(2)$ %può essere in realtà qualsiasi gruppo topologico
be $p$-adic qubit irreps of $\SO(3)_p$, and consider $\mathdutchcal{U}_A\otimes \mathdutchcal{V}_B\colon \SO(3)_p\rightarrow \mathrm{U}(\CC^2_A\otimes \CC^2_B)$. Every one-dimensional subrepresentation $\mathcal{M}\subseteq \CC_A^2\otimes \CC_B^2$ is spanned by a maximally entangled state. 
\end{proposition}
\begin{proof}
Let $\mathcal{M}\equiv \CC \ket{v}\subseteq \CC^2\otimes \CC^2$ be an $\SO(3)_p$-stable subspace, and let $P_\mathcal{M}\equiv \proj{v}\colon \CC^2\otimes \CC^2\rightarrow \mathcal{M} $ be the rank-one orthogonal projector onto $\mathcal{M}$. Then, %per rappresentazione unitaria (di qualsiasi dimensione), stable \big(\mathdutchcal{U}_A(L)\otimes\mathdutchcal{V}_B(L)\big)\mathcal{M}\subseteq\mathcal{M} means that the operators commute
$P_\mathcal{M}$ is an intertwining operator for $\mathdutchcal{U}_A\otimes \mathdutchcal{V}_B$, i.e.
\beq \label{eq:sottostabileunit}
P_\mathcal{M}=\big(\mathdutchcal{U}_A(L)\otimes\mathdutchcal{V}_B(L)\big)P_\mathcal{M}\big(\mathdutchcal{U}_A(L)^\dagger\otimes \mathdutchcal{V}_B(L)^\dagger\big),
\eeq 
for every $L\in \SO(3)_p$%, where $\mathdutchcal{U}_A(L)^\dagger$ denotes the Hermitian adjoint of $\mathdutchcal{U}_A(L)$
. We exploit the invariance of the partial trace $\Tr_A$ under conjugation by local unitaries on the subspace $\CC_A^2$ (a symmetric discussion holds for $\Tr_B$):
\beq \label{eq:proptracciaparz}
\Tr_A(P_\mathcal{M}) = \mathdutchcal{V}_B(L)\Tr_A(P_\mathcal{M})\mathdutchcal{V}_B(L)^\dagger,
\eeq 
for every $L\in \SO(3)_p$. In other words, $\Tr_A(P_\mathcal{M})$ is an intertwining operator for the local action $\mathdutchcal{V}_B$ of $\SO(3)_p$ on the subsystem $\CC_B^2$. Since $\mathdutchcal{V}_B$ is an irreducible representation (over the algebraically closed field $\CC$), Schur's lemma applies%l'unica matrice che commuta con $\mathdutchcal{V}_B$ è l'identità, che quindi è equivalente
: $\Tr_A(P_\mathcal{M})=\lambda \mathbbm{1}_B$ for some $\lambda\in\CC^\times$, where $\mathbbm{1}_B$ denotes the identity operator on $\CC_B^2$. We have $2\lambda=\Tr (\lambda \mathbbm{1}_B)=\Tr \Tr_A(P_\mathcal{M}) = \Tr(P_\mathcal{M})=1$, i.e. $\lambda=1/2$. It means that $\Tr_A(P_\mathcal{M})$ is a maximally mixed quantum state, hence $\ket{v}$ is a maximally entangled vector state. 
\end{proof}

%Now we look for a classification of orthogonal basis of $\CC^2\otimes \CC^2$ of four maximally entangled states.
\begin{fact}\label{fact:notoifo}
On $\CC^2\otimes \CC^2$, any maximally entangled state is obtainable by local unitary transformations on a Bell state.
\end{fact}
\begin{proof}
A state in a bipartite system of finite dimension is maximally entangled if it has maximum Schmidt rank, i.e. if all its Schmidt coefficients are equal to the inverse of the minimum dimension of the subsystems. Two pure states can be transformed into each other by means of local unitaries if and only if they have the same Schmidt coefficients. Therefore, every pair of maximally entangled states are related by local unitaries. The four Bell states are maximally entangled two-qubit states in $\CC^2\otimes \CC^2$.
\end{proof}
\begin{corollary}\label{cor:maxentproj1d}
Let $\mathdutchcal{U}_A,\mathdutchcal{V}_B\colon \SO(3)_p\rightarrow \U(2)$ be $p$-adic qubit irreps of $\SO(3)_p$, and consider $\mathdutchcal{U}_A\otimes \mathdutchcal{V}_B\colon \SO(3)_p\rightarrow \mathrm{U}(\CC^2_A\otimes \CC^2_B)$. Up to local change of basis, the vector spanning a one-dimensional subrepresentation is $\ket{\phi^+}$. The projector on any one-dimensional subrepresentation is a maximally entangled (Bell) state.
\end{corollary}
\begin{proof}
According to Proposition~\ref{prop:1irrmaxent}, a one-dimensional subrepresentation of any $2\otimes 2$ is spanned by a maximally entangled state $\ket{v}$. According to Fact~\ref{fact:notoifo}, $\ket{v}$ can be transformed to any desired maximally entangled state, by a change of basis of $\CC^2\otimes \CC^2$ through local unitaries. Hence, without loss of generality, we assume $\ket{v}$ to be the Bell state $\ket{\phi^+}$. Moreover, $\proj{\phi^+}$ is maximally entangled, as it has maximally mixed partial traces.
\end{proof}

\begin{lemma}\label{lemma:1+1}
Let $2\otimes2$ admit a direct-sum decomposition into irreps which contains $1\oplus1$. Up to a local change of basis, the vectors spanning the one-dimensional subrepresentations are $\ket{\phi^+}$ and $\ket{\phi^-}$.
\end{lemma}
\begin{proof}
By Proposition~\ref{prop:1irrmaxent} and Fact~\ref{fact:notoifo}, assume that the state spanning one of the two one-dimensional irreps is $\ket{\phi^+}$. Call $\ket{\phi}$ the maximally entangled state spanning the other one-dimensional irrep. The Bell state $\ket{\phi^+}$ is rotationally invariant, i.e., $(\mathdutchcal{U}\otimes \mathdutchcal{U}^\ast)\ket{\phi^+}=\ket{\phi^+}$ for every $\mathdutchcal{U}\in \U(2)$, where $\mathdutchcal{U}^\ast$ denotes the complex conjugate of $\mathdutchcal{U}$. It follows that
\beq \label{eq:proprBell+}
(\mathbbm{1}\otimes \mathdutchcal{U}^\ast)\ket{\phi^+} = (\mathdutchcal{U}^\dagger\otimes \mathbbm{1})\ket{\phi^+},
\eeq 
for every $\mathdutchcal{U}\in \U(2)$. This allows us to move all local unitary operators on $\ket{\phi^+}$ to one subsystem only, e.g. the first copy of $\CC^2$. Thus, without loss of generality, we can assume $\ket{\phi}=(\mathdutchcal{V}\otimes \mathbbm{1})\ket{\phi^+}$, for some $\mathdutchcal{V}\in\U(2)$. Now, we wonder if there exists $\mathdutchcal{U}\in \U(2)$ such that $(\mathdutchcal{U}\otimes \mathdutchcal{U}^\ast)\ket{\phi}\sim\ket{\phi^-}$, where $\sim$ denotes the ray equivalence relation in the projective space of $\CC^2\otimes \CC^2$. By Eq.~\eqref{eq:proprBell+} we get
\begin{align}
(\mathdutchcal{U}\otimes \mathdutchcal{U}^\ast)\ket{\phi} %&= (\mathdutchcal{U}\mathdutchcal{V}\otimes \mathdutchcal{U}^\ast)\ket{\phi^+} \nonumber\\ = (\mathdutchcal{U}\mathdutchcal{V}\otimes 1)(1\otimes \mathdutchcal{U}^\ast)\ket{\phi^+} \\ & 
= (\mathdutchcal{U}\mathdutchcal{V}\mathdutchcal{U}^\dagger \otimes \mathbbm{1})\ket{\phi^+}.
\end{align}
As $\mathdutchcal{V}$ is unitary, it is unitarily diagonalisable, and we assume $\mathdutchcal{U}$ to be such that $\mathdutchcal{U}\mathdutchcal{V}\mathdutchcal{U}^\dagger$ is diagonal. Moreover, as the $p$-adic qubit representations are unitary, it must be 
\beq 
0=\braket{\phi^+}{\phi} = \frac{1}{2}\Tr \mathdutchcal{V}.
\eeq 
This imposes the condition $\Tr (\mathdutchcal{U}\mathdutchcal{V}\mathdutchcal{U}^\dagger)=0$. The set of traceless diagonal matrices in $\U(2)$ is $\{e^{i\theta}Z,\ \theta\in\RR\}$. In particular, $\mathdutchcal{U}\mathdutchcal{V}\mathdutchcal{U}^\dagger$ is necessarily the $Z$ Pauli matrix up to a global phase. This provides 
\beq 
(\mathdutchcal{U}\mathdutchcal{V}\mathdutchcal{U}^\dagger \otimes \mathbbm{1})\ket{\phi^+}= (e^{i\theta}Z\otimes \mathbbm{1})\ket{\phi^+}=e^{i\theta}\ket{\phi^-}.
\eeq 
To summarise, there exists $\mathdutchcal{U}\in\U(2)$ such that $\mathdutchcal{U}\otimes \mathdutchcal{U}^\ast$ transforms $\ket{\phi}$ to $e^{i\theta}\ket{\phi^-}$, $\theta\in\RR$, and leaves $\ket{\phi^+}$ invariant. %Lastly, we can change $e^{i\theta}\ket{\phi^-}$ with whatever of its scalar multiples, which span the same one-dimensional stable subspace: we choose any $\theta=2k\pi$, $k\in\ZZ$.
\end{proof}

\begin{proposition}
Case $2\otimes 2=1\oplus 1\oplus 1\oplus 1$. Up to a local change of basis, this decomposition is realised on a basis of four Bell states, and a projector on each singlet is maximally entangled.
\end{proposition}
\begin{proof} By Lemma \ref{lemma:1+1}, we assume that two of the four one-dimensional subrepresentations are spanned by $\ket{\phi^+}$ and $\ket{\phi^-}$. Let $\ket{\psi}$ be spanning another one-dimensional subrepresentation. Again, this can be assumed of the form $\ket{\psi}=(\mathdutchcal{V}\otimes\mathbbm{1})\ket{\phi^+}$ for some $\mathdutchcal{V}\in\U(2)$ by Eq.~\eqref{eq:proprBell+}. By unitarity of the $p$-adic qubit representations, the four subrepresentations must be orthogonal one another: $\ket{\psi}\in\Span(\ket{\phi^+},\ket{\phi^-})^\perp = \Span (\ket{\psi^+},\ket{\psi^-}) = \Span(\ket{01},\ket{10})$, which means 
\beq 
\ket{\psi} = \lambda_{01}\ket{01}+\lambda_{10}\ket{10},
\eeq 
for some $\lambda_{01},\lambda_{10}\in\CC$%\footnote{Equivalently, $0=\langle \phi^+\ket{\psi} = \frac{1}{2}\Tr V$ as before, but also $0=\langle \phi^-\ket{\psi}=\langle \phi^+|(Z\otimes \mathbbm{1})(V\otimes \mathbbm{1})\ket{\phi^+} = \frac{1}{2}\Tr ZV$. This means that $V$ is orthogonal to $\mathbbm{1}$ and $Z$ with respect to the Hilbert-Schmidt product (here equal to the standard Hermitian product) in $\mathsf{M}_{2\times2}(\CC)$. An orthogonal basis for $\mathsf{M}_{2\times2}(\CC)$ is given by the identity together with the Pauli matrices, $\{\mathbbm{1},X,Y,Z\}$. Hence, $V\in\Span\{X,Y\}$ is a unitary off-diagonal matrix: $V=\lambda_XX+\lambda_YY$ for some $\lambda_X,\lambda_Y\in\CC$ provides $\ket{\psi} = \lambda_X\ket{\psi^+}-i\lambda_Y\ket{\psi^-} = \frac{\lambda_X-i\lambda_Y}{\sqrt{2}}|01\rangle + \frac{\lambda_X+i\lambda_Y}{\sqrt{2}}|10\rangle$.}
. By Proposition~\ref{prop:1irrmaxent}, $\ket{\psi}$ is a maximally entangled state, hence it must be $|\lambda_{01}|^2=|\lambda_{10}|^2=\frac{1}{2}$: we get $\ket{\psi}=\frac{1}{\sqrt{2}}\left(e^{i\arg\lambda_{01}}\ket{01} + e^{i\arg\lambda_{10}}\ket{10}\right)$. Except from a global multiplicative constant, the state $\ket{\psi}$ is uniquely determined up to a relative phase. We can remove this ambiguity by exploiting the common remaining symmetry of $\ket{\phi^+}$ and $\ket{\phi^-}$. Indeed, let $\Theta\in\U(2)$ be diagonal, namely $\Theta=\begin{pmatrix}e^{i\theta} & 0\\0&e^{i\varphi} \end{pmatrix}$ for some $\theta,\varphi\in\RR$. Clearly $(\Theta\otimes \Theta^\ast)\ket{\phi^\pm}=\ket{\phi^\pm}$, moreover 
\beq 
(\Theta\otimes \Theta^\ast)\ket{\psi} = \frac{e^{i(\arg\lambda_{01}+\theta-\varphi)}}{\sqrt{2}}\left(\ket{01}+e^{i(\arg\lambda_{10}-\arg\lambda_{01}+2\varphi-2\theta)}\ket{10}\right).
\eeq 
If $\varphi-\theta=\frac{1}{2}(\arg\lambda_{01}-\arg\lambda_{10}+2k\pi)$ then $(\Theta\otimes \Theta^\ast)\ket{\psi}\sim \ket{\psi^+}$; if $\varphi-\theta=\frac{1}{2}(\arg\lambda_{01}-\arg\lambda_{10}+(2k+1)\pi)$ then $(\Theta\otimes \Theta^\ast)\ket{\psi}\sim \ket{\psi^-}$. %La fase rimasta e' rispettivamente $\arg\lambda_{01}+\theta-\varphi = \frac{1}{2}(\arg\lambda_{01}+\arg\lambda_{10}-2k\pi)$ o $\arg\lambda_{01}+\theta-\varphi = \frac{1}{2}(\arg\lambda_{01}+\arg\lambda_{10}-(2k+1)\pi)$.
%We are just left to get rid of the global phase, by chosing the appropriate scalar multiple of $(\Theta\otimes \Theta^\ast)\ket{\psi}$ (spanning the same subspace), or considering rays in the projective Hilbert space.
We have shown that three of the four one-dimensional subrepresentations are spanned by three Bell states. The forth stable line must be orthogonal to the others and again spanned by a maximally entangled state. As by Corollary~\ref{cor:maxentproj1d}, the projector on each one-dimensional stable subspace is a Bell density operator, which is maximally entangled.
\end{proof}
This is, for instance, what we already noticed for $p=3,7$ in Subsection~\ref{subsec:CBcpeff2357}. 

\begin{proposition}
Case $2\otimes 2=1\oplus 1\oplus 2$. Up to a local change of basis, this decomposition is realised on a basis of four Bell states. Projectors on the singlets are maximally entangled, while a projector on the doublet is separable.
\end{proposition}
\begin{proof}
By Lemma~\ref{lemma:1+1}, the two one-dimensional subrepresentations are spanned by $\ket{\phi^+}$ and $\ket{\phi^-}$. These lines are  orthogonal to the bidimensional stable subspace. Hence, a possible basis for the latter is given by the remaining Bell states, $\ket{\psi^+}$ and $\ket{\psi^-}$. Again, the projectors onto the one-dimensional stable subspaces are maximally entangled by Corollary~\ref{cor:maxentproj1d}. On the other hand, by the PPT criterion, a convex combination of Bell density operators is separable if and only if all probability coefficients are smaller than or equal to $\frac{1}{2}$. It follows that the projector $
%\frac{1}{2}P_{\Span(\ket{\psi^\pm})}=
\frac{1}{2}\left(\proj{\psi^+}+\proj{\psi^-}\right)$ onto the bidimensional stable subspace is separable.
\end{proof}
We have specifically noticed this already for $p=2,5,7$ in Subsection~\ref{subsec:CBcpeff2357}.

\begin{proposition}
Case $2\otimes 2 = 1\oplus 3$. Up to a local change of basis, the one-dimensional stable subspace is spanned by the singlet $\ket{\psi^-}$, and the three-dimensional one by the triplet $(\ket{00}, \ket{\psi^+}, \ket{11})$ [or equivalently $(\ket{\psi^+}, \ket{\phi^+}, \ket{\phi^-})$]. The projector onto the singlet is maximally entangled, while the projector onto the triplet is separable.
\end{proposition}
\begin{proof}
By Proposition~\ref{prop:1irrmaxent} and Fact~\ref{fact:notoifo}, the one-dimensional subrepresentation is spanned by $\ket{\psi^-}$%se prendo le \phi, poi \ket{00} e \ket{11} non sono più ortogonali
; the three-dimensional stable subspace is its orthogonal complement. The projector onto the three-dimensional subspace is $\proj{00}+ \proj{11} +\proj{\psi^+}%=|00\rangle\!\langle00|+\frac{1}{2}|01\rangle\!\langle01|+ \frac{1}{2}|10\rangle\!\langle10|+ |11\rangle\!\langle11|  +\frac{1}{2}|01\rangle\!\langle10|+ \frac{1}{2}|10\rangle\!\langle01| = \begin{pmatrix}1&0&0&0\\0&\frac{1}{2}&\frac{1}{2}&0\\0&\frac{1}{2}&\frac{1}{2}&0\\0&0&0&1 \end{pmatrix}
$, whose partial transpositions are semidefinite-positive. %partial transposed with respect to the second subsystem equal to $|00\rangle\!\langle00|+\frac{1}{2}|01\rangle\!\langle01|+ \frac{1}{2}|10\rangle\!\langle10|+ |11\rangle\!\langle11|  +\frac{1}{2}|00\rangle\!\langle11|+ \frac{1}{2}|11\rangle\!\langle00| = \begin{pmatrix}1&0&0&\frac{1}{2}\\0&\frac{1}{2}&0&0\\0&0&\frac{1}{2}&0\\\frac{1}{2}&0&0&1 \end{pmatrix}$ of eigenvalues $3/2$ and $\frac{1}{2}$ with multiplicity three. %They are all $\geq0$, i.e. the PT operator is semidefinite positive, hence
By the PPT criterion, the projector on the triplet subspace is separable. %Instead, the Bell projector on the singlet subspace is clearly maximally entangled. 
\end{proof}
This is what happens in standard quantum mechanics (cf.\ Eq.~\eqref{eq:singltripl}), and has not shown in the $p$-adic framework at the level $\SO(3)_p\mod p$.

\begin{proposition}
%Every two-dimensional subrepresentation of the tensor product of two two-dimensional irreducible unitary representations is spanned by maximally entangled states. 
Let $\mathdutchcal{U}_A,\mathdutchcal{V}_B\colon \SO(3)_p\rightarrow \U(2)$ be $p$-adic qubit irreps of $\SO(3)_p$, and consider $\mathdutchcal{U}_A\otimes \mathdutchcal{V}_B\colon \SO(3)_p\rightarrow \mathrm{U}(\CC^2_A\otimes \CC^2_B)$. For every two-dimensional subrepresentation $\mathcal{M}\subseteq \CC_A^2\otimes \CC_B^2$, a projector $P_\mathcal{M}$ onto $\mathcal{M}$ is separable. Every two-dimensional subrepresentation %no need to be irreducible
$\mathcal{M}$ is spanned by two orthogonal maximally entangled states. Up to local change of basis, $2\otimes2\simeq 2\oplus2$ is realised on a basis of four Bell states. 
\end{proposition}
\begin{proof}
Let $(\ket{\phi_1},\ket{\phi_2})$ be an orthonormal basis for $\mathcal{M}$, which is an $\SO(3)_p$-stable subspace. Let $P_\mathcal{M}= \proj{\phi_1}+\proj{\phi_2}\colon \CC^2\otimes \CC^2\rightarrow \mathcal{M} $ be the rank-two orthogonal projector onto $\mathcal{M}$. In this context, %per stabilità e uintarietà
Eqs.~\eqref{eq:sottostabileunit},~\eqref{eq:proptracciaparz} still hold, hence $\Tr_A(P_\mathcal{M})=\lambda\mathbbm{1}_B$ by Schur's lemma. Here $2\lambda=\Tr (\lambda\mathbbm{1}_B)=\Tr(P_\mathcal{M})=2$ implies $\lambda=1$: we have $\mathbbm{1}_B=\Tr_A(P_\mathcal{M})=\Tr_A(\proj{\phi_1})+\Tr_A(\proj{\phi_2})$. Let $\rho_i\coloneqq \Tr_A(\proj{\phi_i})$ for $i=1,2$. Then $\rho_1$ and $\rho_2$ commute, and with respect to a common eigenbasis $\rho_1=\begin{pmatrix}\lambda_1&0\\0&1-\lambda_1\end{pmatrix}$, $\rho_2=\begin{pmatrix}1-\lambda_1&0\\0&\lambda_1\end{pmatrix}$. By~\cite{Parthasarathy},~\cite{CuMonAndreas}, the maximal dimension of a fully entangled subspace for two qubits is $1$, %r=2 ossia no prodotti tensori, arriva in dimensione massima (d_A-1)(d_B-1)... per 2 qubit d_A=d_B=2
hence $\mathcal{M}$ contains a product state. Let $\ket{\phi_1}=\ket{\alpha_1}_A\otimes\ket{\beta_1}_B$, for some $\ket{\alpha_1}_A\in\CC_A^2$, $\ket{\beta_1}_B\in\CC_B^2$. Equivalently, $\rho_1$ is a rank-one projector, and so is $\rho_2$%\rho_1,\rho_2 hanno stessi spettri di Smidth, vedi sopra la loro scrittura diagonale
, i.e., also $\ket{\phi_2}$ is a product state. Let $\ket{\phi_2}=\ket{\alpha_2}_A\otimes\ket{\beta_2}_B$, for some $\ket{\alpha_2}_A\in\CC_A^2$, $\ket{\beta_2}_B\in\CC_B^2$. Now $\mathbbm{1}_B=\Tr_A(\proj{\phi_1})+\Tr_A(\proj{\phi_2}) = \proj{\beta_1}+\proj{\beta_2}$, %somma di proiettori su linee è l'identità su tutto lo spazio sse le linee sono ortogonali
hence it must be $\braket{\beta_1}{\beta_2}=0$. Similarly $\Tr_B(P_\mathcal{M})=\mathbbm{1}_A$ provides $\braket{\alpha_1}{\alpha_2}=0$. It follows that $P_\mathcal{M}=\proj{\alpha_1}\otimes\proj{\beta_1}+\proj{\alpha_2}\otimes\proj{\beta_2}$ is separable. Up to a local change of basis, $\mathcal{M}$ is spanned by a basis of maximally mixed states, e.g. $(\ket{\psi^+},\ket{\psi^-})$. Then, a basis for the bidimensional orthogonal complement is $(\ket{\phi^+},\ket{\phi^-})$, with associated separable projector.
\end{proof}
This has already been noted in Subsection~\ref{subsec:CBcpeff2357} in particular for $p=7$.

\section{\texorpdfstring{$3$}{Lg}-Adically controlled quantum logic gates}\label{3pUnigates}
In the realm of quantum computing, we propose to construct $p$-adically controlled quantum logic gates on single-, two- and multi-qubit systems, using image elements of the unitary representations of $\SO(3)_p$. A gate on $n$ qubits is modelled by a $2^n$-dimensional unitary transformation from a representation of $\SO(3)_p$ (plus an ancillary qubit in the case of orthogonal transformations). As in the standard approach, the emphasis lies primarily on $p$-adic quantum logic gates operating on a small number of qubits. Our program aims to examine $2$-qubit gates from unitary representations of $\SO(3)_p$ acting on $(\CC^2)^{\otimes 2}$. We seek for an entangling $2$-qubit gate, with the ultimate goal of providing a universal set of $p$-adically controlled gates%\bibitem{Brylinski}: per avere un set universale di one- e two-qubit gates, è necessario avere una porta two-qubit entangling
. Once explicitly found some four-dimensional irrep of $\SO(3)_p$, we study its image elements in $\U(4)$, wondering whether they factorise as a tensor product of two unitaries in $\U(2)$. The non-tensor-product unitaries of a representation will possibly be entangling $2$-qubit gates (if not equal to a $\mathrm{SWAP}$ times a product unitary)~\cite{Brylinski}%\bibitem{Brylinski}: 
%U primitiva/non-entangling <->)def) mappa stati fattorizzabili in stati fattorizzabili <->(Th1.4) o U=S\oxT o U=(S\oxT)SWAP.. SWAP è l'unica primitiva non banale, nel senso che non entangla ma non è fattorizzabile... 
%quindi, equivalentemente, U imprimitive/entangling <-> crea entanglement da alcuni stati fattore <-> nè U nè USWAP sono S\oxT... quindi, verificare che U non fattorizza non è sufficiente a dimostrare che U è entangling
; from the product images, instead, we will extract the single tensor factors to be one-qubit gates. 

In the present work, we studied representations of $\SO(3)_p$ factorising through $G_p$, and only $G_3$ has four-dimensional irreps (see Remark~\ref{rem:solo4Dp3}). This sets our staring point, from which we will put forward a set of $3$-adically controlled gate (constructed as described above), and prove that it is approximately universal. This is %Subsections~\ref{subsec:irreps4G3},~\ref{eq:gateresearch} and~\ref{eq:gatesetuniversal} are 
done with the aid of a code intertwining GAP (Groups, Algorithms, and Programming)~\cite{GAP4} and Wolfram Mathematica~\cite{Math}; this code can be found in~\cite{ourcode}.

\subsection{Universal quantum computation}
We deal with the circuit model of quantum computation, where the building blocks of quantum circuits are quantum logic gates%interconnected by wires/wired together
. Each gate corresponds to a unitary operator acting on a certain number $n$ of qubits, i.e. to an operator in $\U(2^n)$. For physical implementations, it is desirable that quantum gates act on a small number of qubits, typically three or preferably two. It is useful to have a small repertoire of simple gates, combining which all the circuits (i.e. all the unitary transformations in any number of qubits) can be reconstructed approximatively. %This is why an extensive line of research has focussed on studying universality of sets of gates, as different sets are better suited to different purposes, both in the theoretical design of algorithms (and in understanding the power of quantum computing) and in their practical laboratory implementation.

\begin{definition}[\cite{Brylinski,Aharonov}]
\label{def:uinversalitu}
A set $\mathcal{S}$ of multi-qubit gates is called \emph{(approximatively) universal} if it satisfies one of the following equivalent conditions: there exists $n_0\in\NN$ such that, for every $n\geq n_0$, 
\begin{itemize}
\item every $n$-qubit gate can be approximated with arbitrary accuracy by a circuit made up of the $n$-qubit gates produced by the elements in $\mathcal{S}$;
\item i.e., for every $U\in\U(2^n)$ and every $\epsilon>0$ there exists $U'\in\U(2^n)$ made up of elements in $\mathcal{S}$ such that $\sup\limits_{\substack{\ket{\psi}\in(\CC^{2})^{\otimes n}\\\norm{\psi}=1}}\norm{(U-U')\ket{\psi}}<\epsilon;$
\item \vspace{-0.4cm} the $n$-qubit gates produced by the elements in $\mathcal{S}$ generate a subgroup dense in $\U(2^n)$.
\end{itemize}
In all these formulations, it is assumed that any given $k$-qubit gate from $\mathcal{S}$ can act on any subfamily of $k$ out of the $n$ qubits, i.e.~each such gate gives rise to really $k!{n\choose k}$ $n$--qubit unitaries.

A set $\mathcal{S}$ of qubit gates is called \emph{exactly universal} if there exists $n_0\in\NN$ such that, for each $n\geq n_0$, every $n$-qubit gate can be obtained exactly by a circuit made up of the $n$-qubit gates produced by the elements in $\mathcal{S}$, i.e., the $n$-qubit gates produced by the elements in $\mathcal{S}$ generate the full group $\U(2^n)$.
\end{definition}

A finite set of gates cannot be exactly universal, but just universal. There are also weaker notions of universality. These include reproducing a $U\in \U(2^n)$ as $U\otimes \mI$ using ancilla qubits~\cite{Shi}, as well as encoded universality via orthogonal gates~\cite{BernVazi} (as we will detail in the paragraph ``Real quantum computation'' below).

A first foundational result by Deutsch is that the three-qubit gate $Q\coloneqq C^2U_\alpha$ is universal~\cite{Deutsch} with $n_0=3$, where $U_\alpha\coloneqq ie^{-i\frac{\pi}{2}\alpha}X^\alpha$ with $\alpha\in\RR\setminus\QQ$ ($Q$ is a double-controlled unitary %with an angle $\alpha$ incommensurate with $\pi$
representing a continuous quantum generalization of the Toffoli gate $C^2X$). In particular, $Q$ approximates the Toffoli gate, which is universal for classical computation, and hence realises the permutations of the basis elements%non dei qubit subspaces bidimensionali come dice Llooyd, ma dei sottospazi unidimensioali di base
. With this, Deutsch proved that $Q$ approximates any unitary in $\U(2^3)$, and then any unitary in $\U(2^n)$ with $n\geq 3$. Upon this, DiVincenzo showed that two-qubit gates are universal~\cite{DiVincenzo} with $n_0=2$, by showing that two-qubit gates approximate three-qubit gates by Lie algebra techniques. If two elements of the Lie algebra $\mathfrak{u}(4)$ do not belong to a Lie subalgebra and do not commute, their commutator provides a third generator of $\mathfrak{u}(4)$, and iterating this process one generates the whole $\mathfrak{u}(4)$; the elements in the connected group $\U(4)$ corresponding to the generators of $\mathfrak{u}(4)$ %tramite esponenziazione.. ovvio che se esponenzio con un coefficiente continuo reale allora genero proprio tutto \U(4) 
generate a dense subgroup in $\U(4)$. %Then, it was proven that any gate on two (or more) qubits is universal~\cite{DeutschBarenkoE} (as also pointed out in~\cite{Lloyd}).
Then, it was proved that the set $\{\U(2),CX\}$ is exactly universal~\cite{BarencoetAl}, and the set $\{\textup{Toffoli}, \textup{Hadamard}\}$ is universal~\cite{Shi}, \cite{Aharonov}, just to cite a few.

\begin{remark}
\label{eq:remBVrelaQC}
Since a quantum state is a ray in a projective Hilbert space, where global phases are irrelevant, Definition~\ref{def:uinversalitu} can be given in terms of special unitary transformations only. %Unitary groups modular a phase, come dice~\cite{Shi}. %alla fine del circuito SU, avrò implementato una speciale unitaria, che è una unitaria up to global phase, e così per ogni circuito
Furthermore, Definition~\ref{def:uinversalitu} can be given in terms of special orthogonal transformations only, since complex quantum computation can be encoded in \emph{real quantum computation} which relies exclusively on such operations~\cite{BernVazi}. %fare computazione complessa con un'altra che usa solo coefficienti reali e trasformazioni ortogonali, porte reali ortogonali per l'universalità se lavoriamo in un sottospazio codificato (raddoppiando la dimensione possiamo rappresentare lo spazio di Hilbert complesso per uno reale), con SO(4), realizziamo SU(2) in codificazione... 
\end{remark}

\noindent\textbf{Real quantum computation}.  We recall how to concretely transform any quantum circuit to a circuit (computing the same function) composed solely of special orthogonal matrices~\cite{BernVazi}, \cite{Aharonov}. An extra qubit is added to the circuit, doubling the dimension of the initial Hilbert space% (as $\CC^m\simeq \RR^{2m}$)
. This is used as a flag to select between the real and imaginary parts of a state, so to get a purely real state. A conventional way to do it is as follows. The extra qubit's state is put to $\ket{0}$ (resp. $\ket{1}$) to indicate that the circuit's state is in the real (resp.~imaginary) part of the Hilbert space. Let $\{\ket{j}\}_j$ denote the canonical basis of $(\CC^2)^{\otimes n}$, and consider $(\CC^2)^{\otimes n} \ni \ket{\psi}=\sum_jc_j\ket{j}=\sum_j\left(\mathrm{Re}(c_j)+i\mathrm{Im}(c_j)\right)\ket{j}=\mathrm{Re}(\ket{\psi})+i\mathrm{Im}(\ket{\psi})$, with $c_j\in\CC$ for every $j$. This is encoded via the extra qubit by 
\begin{align} 
(\CC^2)^{\otimes n}\ni \ket{\psi} \mapsto \ket{\widetilde{\psi}}
 \coloneqq& \,\mathrm{Re}(\ket{\psi})\ket{0}+\mathrm{Im}(\ket{\psi})\ket{1}\in (\RR^2)^{\otimes (n+1)}\label{encodeingBVQC}\\
         =& \sum_j\mathrm{Re}(c_j)\ket{j}\ket{0}+\sum_j\mathrm{Im}(c_j)\ket{j}\ket{1},\nonumber
\end{align}
which is a real-linear map; clearly $\braket{\psi}{\psi}=\braket{\widetilde{\psi}}{\widetilde{\psi}}$. 
Now, in the initial circuit, each unitary gate $U$ acting on $m$ qubits is replaced by a real matrix $\widetilde{U}$ acting on $m+1$ qubits (the same $m$ qubits and the extra one), defined as follows:
\begin{align}
 \widetilde{U}\ket{j}\ket{0} &\coloneqq \phantom{-} \big(\mathrm{Re}(U)\ket{j}\big)\ket{0}+ \big(\mathrm{Im}(U)\ket{j}\big)\ket{1},\nonumber\\ 
 \widetilde{U}\ket{j}\ket{1} &\coloneqq -\big(\mathrm{Im}(U)\ket{j}\big)\ket{0}+ \big(\mathrm{Re}(U)\ket{j}\big)\ket{1}.\label{eq:realQC}
\end{align}
% puoi rimpiazzare ogni singola porta così, ed il qubit ausiliario rimarrà per la codifica del prossimo stato per la prossima porta... perchè in realtà puoi pensare che la porta locale U è I\otimes U \otimes I che agisce sullo stato totale, e lì fai la codifica..

This encoding places the flag qubit at the end of the tensor product (contrarily to the algebraic convention). %This convention is well-suited for quantum computation, mostly for the representation of local gates (applichi la stessa di struttura tensoriale con l'identità a contornare la gate locale tildata... vedi anche Aharonov CP(i) che diventa la localizzata aggraziata C^2(XZ), contro il mio approccio che dà una matrice di blocchi spread all over). Nevertheless, this is simply a basis ordering convention. By a basis change, placing the ancilla at the beginning of the tensor product, one recovers the convention used in linear algebra, where the canonical isomorphism $\CC^{m}\simeq \RR^{2m}$ of $\RR$-vector spaces is given by $\ket{\psi}\mapsto\ket{\widehat{\psi}}\coloneqq \ket{0}\mathrm{Re}(\ket{\psi})+\ket{1}\mathrm{Im}(\ket{\psi}) = \begin{pmatrix}\mathrm{Re}(\ket{\psi})\\ \mathrm{Im}(\ket{\psi}) \end{pmatrix}$. (This is related by a change of basis to Eq.~\eqref{encodeingBVQC} (in the same dimension $m=2^{n+1}$), moving the flag qubit from the beginning to the end in the tensor product. With respect to the algebraic approach, ) Here, it is easy to check that $U \mapsto \widehat{U} \coloneqq \begin{pmatrix} \mathrm{Re}(U) & -\mathrm{Im}(U)\\ \mathrm{Im}(U) & \mathrm{Re}(U)\end{pmatrix}$ provides the canonical embedding $\U(m)\hookrightarrow\SO(2m)$ (which is surjective on $\SO(2m)\cap \mathrm{Sp}(2m,\RR)$). We have $\widehat{U}\ket{\widehat{\psi}}=\begin{pmatrix} \mathrm{Re}(U) & -\mathrm{Im}(U)\\ \mathrm{Im}(U) & \mathrm{Re}(U)\end{pmatrix}\begin{pmatrix}\mathrm{Re}(\ket{\psi})\\\mathrm{Im}(\ket{\psi})\end{pmatrix} = \begin{pmatrix}\mathrm{Re}(U)\mathrm{Re}(\ket{\psi})\\\mathrm{Im}(U)\mathrm{Re}(\ket{\psi})\end{pmatrix} + \begin{pmatrix} -\mathrm{Im}(U)\mathrm{Im}(\ket{\psi})\\ \mathrm{Re}(U)\mathrm{Im}(\ket{\psi}) \end{pmatrix}$, in agreement with Eq.~\eqref{eq:realQC}. Since the matrices $\widetilde{U}$ and $\widehat{U}$ are conjugated, we get $\widetilde{U}\in\SO(2^{n+1})$. 
For $U\in\U(m)$, $m\in\NN$, this corresponds to transforming the matrix representation $U=\left(U_{ij}\right)_{ij}$, with respect to some basis, by $U_{ij}\mapsto\begin{pmatrix}\mathrm{Re}(U_{ij}) & -\mathrm{Im}(U_{ij})\\ \mathrm{Im}(U_{ij}) & \mathrm{Re}(U_{ij}) \end{pmatrix}$ for every $i,j$. This provides an embedding $\U(m)\hookrightarrow\SO(2m)$ (which is surjective on $\SO(2m)\cap \mathrm{Sp}(2m,\RR)$).

\subsection{Four-dimensional irreps of $G_3$}
\label{subsec:irreps4G3}

The group $G_3\coloneqq\SO(3)_3\mod 3$ has exactly four $4$-dimensional irreps (Theorem~\ref{theo:irrepsGp}), and we are going to explicitly write them by making use of the software GAP~\cite{ourcode}.

First of all, we define the group $G_3$, by listing its elements. Parametrisation~\eqref{eq:paramGp} for $p=3$ becomes
\beq \label{eq:paramG3}
G_3=\left\{\begin{pmatrix}
a & -s b & 0 \\
b & s a & 0 \\
c & d & s
\end{pmatrix}\midd  (a,b)\in\{(\pm1,0),\,(0,\pm1)\},c,d\in \ZZ/3\ZZ, s\in\{\pm1\}\right\}.
\eeq

By means of GAP, we establish that $G_3$ is a solvable (as all the groups of order $72$) but not a nilpotent group, of structure
\beq \label{eq:G3strctGens} 
G_3 = \left\langle %M10 =
\begin{pmatrix} 0& 1& 0 \\ -1& 0& 0 \\ 0& 0& 1 \end{pmatrix},\ %M4=
\begin{pmatrix} 0& 1& 0 \\ 1& 0& 0 \\ 1 & 0& -1 \end{pmatrix}\right\rangle \simeq (\mathrm{S}_3\times \mathrm{S}_3)\rtimes C_2\simeq (C_3\times C_3)\rtimes \mathrm{D}_4.
\eeq

We compute the character table of $G_3$ by GAP through the Dixon-Schneider algorithm~\cite{GAPch1,DixonCh,SchneiderCh,Hulpke}, and obtain the following:
\begin{center}
\begin{tabular}{ |c||c|c|c|c|c|c|c|c|c| } 
 \hline
 $G_3$ & $C_1$ &  $C_2$ & $C_3$ & $C_4$ &  $C_5$ & $C_6$ & $C_7$ &  $C_8$ & $C_9$ \\ \hline\hline
 $\mathdutchcal{triv}$ & $1$ & $1$ & $1$ & $1$ & $1$ & $1$ & $1$ & $1$ & $1$ \\ \hline
 $\mathdutchcal{s}$ & $1$ & $1$ & $1$ & $1$ & $1$ & $-1$ & $-1$ & $-1$ & $-1$ \\ \hline
 $\mathdutchcal{t}$ & $1$ & $1$ & $1$ & $1$ & $-1$ & $1$ & $1$ & $-1$ & $-1$ \\ \hline
 $\mathdutchcal{st}$ & $1$ & $1$ & $1$ & $1$ & $-1$ & $-1$ & $-1$ & $1$ & $1$\\ \hline
 $\rho_{3,1}$ & $2$ & $2$ & $2$ & $-2$ & $0$ & $0$ & $0$ & $0$ & $0$\\ \hline
 $\mathdutchcal{U}_{3,1}^{(1)}$ & $4$ & $1$ & $-2$ & $0$ & $0$ & $-2$ & $1$ & $0$ & $0$  \\ \hline
 $\mathdutchcal{U}_{3,1}^{(2)}$ & $4$ & $1$ & $-2$ & $0$ & $0$ & $2$ & $-1$ & $0$ & $0$ \\ \hline
 $\mathdutchcal{U}_{3,1}^{(3)}$ & $4$ & $-2$ & $1$ & $0$ & $0$ & $0$ & $0$ & $-2$ & $1$ \\ \hline
 $\mathdutchcal{U}_{3,1}^{(4)}$ & $4$ & $-2$ & $1$ & $0$ & $0$ & $0$ & $0$ & $2$ & $-1$ \\ \hline
\end{tabular}
\end{center}

\medskip
\noindent
Here, the conjugacy classes $C_i$, $i=1,...,9$, of $G_3$ are as in~\cite[App.~C]{our2nd}, the $1$-dimensional irreps of $G_3$ are as in Eq.~\eqref{1dirrepsDp+1Gp}, $\rho_{3,1}$ denotes the $2$-dimensional irrep of $G_3$ as in Eq.~\eqref{eq:GpdaDp+1}, and $\mathdutchcal{U}_{3,1}^{(j)}$, $j=1,2,3,4$, denote the four $4$-dimensional irreps of $G_3$. All irreducible representations of $G_3$ are self-dual, as their characters take real values. Moreover, by the product of the characters, we see the following equivalences of representations:
\begin{align}
&\mathdutchcal{s}\ox \mathdutchcal{U}_{3,1}^{(1)}\simeq \mathdutchcal{U}_{3,1}^{(2)},\qquad \mathdutchcal{s}\ox \mathdutchcal{U}_{3,1}^{(3)}\simeq \mathdutchcal{U}_{3,1}^{(4)},\label{eq:equivtens}\\
\mathdutchcal{t}\ox \mathdutchcal{U}_{3,1}^{(1)}\simeq& \mathdutchcal{U}_{3,1}^{(1)},\qquad \mathdutchcal{t}\ox \mathdutchcal{U}_{3,1}^{(2)}\simeq \mathdutchcal{U}_{3,1}^{(2)},\qquad  \mathdutchcal{t}\ox \mathdutchcal{U}_{3,1}^{(3)}\simeq \mathdutchcal{U}_{3,1}^{(4)}.
\end{align}

Now, through GAP, we construct a unitary irreducible representation of a finite group affording an irreducible character by exploiting Dixon's method~\cite{DixonIrr}.
We obtain the following expression for the representation $\mathdutchcal{U}_{3,1}^{(2)}$, up to equivalences (with respect to some basis of $\CC^4$ GAP chooses), on the minimal set of generators of $G_3$ as in Eq.~\eqref{eq:G3strctGens}:
\begin{align}
&\mathdutchcal{U}_{3,1}^{(2)}\colon G_3\rightarrow \U(4),\label{repnU2chi9daGAP}\\
&\begin{pmatrix} 0& 1& 0 \\ -1& 0& 0 \\ 0& 0& 1 \end{pmatrix}   %M10
\mapsto 
\begin{pmatrix}
\frac{1}{2} e^{\frac{i \pi}{3}} & \frac{1}{2} e^{-\frac{5i\pi}{6}} & 0 & \frac{e^{-\frac{2i\pi}{3}}}{\sqrt{2}} \\
\frac{i}{2} & \frac{1}{2} e^{-\frac{2i\pi}{3}} & 0 & \frac{i}{\sqrt{2}} \\
\frac{e^{-\frac{2i\pi}{3}}}{\sqrt{2}} & \frac{e^{-\frac{5i\pi}{6}}}{\sqrt{2}} & 0 & 0 \\
0 & 0 & -1 & 0
\end{pmatrix},
\quad 
%in realtà GAP mi dava U2 = \begin{pmatrix} 0& 1& 0 \\ 1& 0& 0 \\ 0& 1& -1 \end{pmatrix} \mapsto \begin{pmatrix} 0 & 0 & \frac{e^{-i\frac{2}{3}\pi}}{\sqrt{2}} & \frac{e^{-i\frac{2}{3} \pi}}{\sqrt{2}} \\ 0 & 0 & -\frac{i}{\sqrt{2}} & \frac{i}{\sqrt{2}} \\ \frac{e^{-i\frac{2}{3} \pi}}{\sqrt{2}} & \frac{e^{-i\frac{5}{6} \pi}}{\sqrt{2}} & 0 & 0 \\ \frac{1}{\sqrt{2}} & \frac{e^{i \frac{5}{6}\pi}}{\sqrt{2}} & 0 & 0 \end{pmatrix}
\begin{pmatrix} 0& 1& 0 \\ 1& 0& 0 \\ 1 & 0& -1 \end{pmatrix}   %M4
\mapsto \begin{pmatrix} 
0 & 0 & \frac{e^{\frac{2 i \pi}{3}}}{\sqrt{2}} & \frac{1}{\sqrt{2}} \\
0 & 0 & \frac{e^{\frac{5 i \pi}{6}}}{\sqrt{2}} & \frac{e^{-\frac{5 i \pi}{6}}}{\sqrt{2}} \\
\frac{e^{\frac{2 i \pi}{3}}}{\sqrt{2}} & \frac{i}{\sqrt{2}} & 0 & 0 \\
\frac{e^{\frac{2 i \pi}{3}}}{\sqrt{2}} & -\frac{i}{\sqrt{2}} & 0 & 0
\end{pmatrix}.
\nonumber
\end{align}
The representation $\mathdutchcal{U}_{3,1}^{(2)}$ is faithful, for $\big\lvert\mathdutchcal{U}_{3,1}^{(2)}(G_3)\big\rvert = 72=\lvert G_3\rvert$. Then, $\mathdutchcal{U}_{3,1}^{(1)}$, up to equivalence, is found by $\mathdutchcal{U}_{3,1}^{(1)}\simeq \mathdutchcal{s}\ox \mathdutchcal{U}_{3,1}^{(2)}$, observing that $\begin{pmatrix} 0& 1& 0 \\ -1& 0& 0 \\ 0& 0& 1 \end{pmatrix}\in C_5$, $\begin{pmatrix} 0& 1& 0 \\ 1& 0& 0 \\ 1 & 0& -1 \end{pmatrix} %=U4, ma anche U2=\begin{pmatrix} 0& 1& 0 \\ 1& 0& 0 \\ 0& 1& -1 \end{pmatrix}
\in C_9$, and recalling that the representation $\mathdutchcal{s}$ maps every matrix in $C_5$ to $1$, and every matrix in $C_9$ to $-1$. Then, the images of the generator in $C_5$ with respect to $\mathdutchcal{U}_{3,1}^{(1)}$ and to $\mathdutchcal{U}_{3,1}^{(2)}$ coincide, while those of the generator in $C_9$ are opposite to each other:
\begin{align}
&\mathdutchcal{U}_{3,1}^{(1)}\colon G_3\rightarrow \U(4),\\
&\begin{pmatrix} 0& 1& 0 \\ -1& 0& 0 \\ 0& 0& 1 \end{pmatrix}  \mapsto \begin{pmatrix}
\frac{1}{2} e^{\frac{i \pi}{3}} & \frac{1}{2} e^{-\frac{5i\pi}{6}} & 0 & \frac{e^{-\frac{2i\pi}{3}}}{\sqrt{2}} \\
\frac{i}{2} & \frac{1}{2} e^{-\frac{2i\pi}{3}} & 0 & \frac{i}{\sqrt{2}} \\
\frac{e^{-\frac{2i\pi}{3}}}{\sqrt{2}} & \frac{e^{-\frac{5i\pi}{6}}}{\sqrt{2}} & 0 & 0 \\
0 & 0 & -1 & 0
\end{pmatrix},
\quad 
\begin{pmatrix} 0& 1& 0 \\ 1& 0& 0 \\ 1 & 0& -1 \end{pmatrix} \mapsto \begin{pmatrix}
0 & 0 & \frac{e^{-\frac{i\pi}{3}}}{\sqrt{2}} & -\frac{1}{\sqrt{2}} \\
0 & 0 & \frac{e^{-\frac{i\pi}{6}}}{\sqrt{2}} & \frac{e^{\frac{i\pi}{6}}}{\sqrt{2}} \\
\frac{e^{-\frac{i\pi}{3}}}{\sqrt{2}} & -\frac{i}{\sqrt{2}} & 0 & 0 \\
\frac{e^{-\frac{i\pi}{3}}}{\sqrt{2}} & \frac{i}{\sqrt{2}} & 0 & 0
\end{pmatrix}
.\nonumber
\end{align}

We use GAP again to find $\mathdutchcal{U}_{3,1}^{(4)}$ as an irrep affording the last character of the table, and get $\mathdutchcal{U}_{3,1}^{(3)}$ as $\mathdutchcal{s}\ox \mathdutchcal{U}_{3,1}^{(4)}$, up to equivalences:
\begin{align}\label{repnU4chi8daGAP}
\mathdutchcal{U}_{3,1}^{(3)},\,\mathdutchcal{U}_{3,1}^{(4)}\colon G_3&\rightarrow\U(4),\qquad \begin{pmatrix} 0& 1& 0 \\ -1& 0& 0 \\ 0& 0& 1 \end{pmatrix} \mapsto \begin{pmatrix} 0& -\frac{1}{2}& -\frac{\sqrt{3}}{2}%\frac{1}{2}e^{-i\frac{5}{6}\pi}-\frac{1}{2}e^{-i\frac{\pi}{6}}
& 0 \\ 1& 0& 0& 0 \\ 0& 0& 0& -1 \\ 0& \frac{\sqrt{3}}{2}%-\frac{1}{2}e^{-i\frac{5}{6}\pi}+\frac{1}{2}e^{-i\frac{\pi}{6}}
& -\frac{1}{2}& 0 \end{pmatrix},\nonumber\\
\mathdutchcal{U}_{3,1}^{(4)}\begin{pmatrix} 0& 1& 0 \\ 1& 0& 0 \\ 1 & 0& -1 \end{pmatrix} & =   \begin{pmatrix}
    1 & 0 & 0 & 0 \\
0 & -\frac{1}{2} & \frac{\sqrt{3}}{2} & 0 \\
0 & -\frac{\sqrt{3}}{2} & -\frac{1}{2} & 0 \\
0 & 0 & 0 & -1
\end{pmatrix}   = - \mathdutchcal{U}_{3,1}^{(3)}\begin{pmatrix} 0& 1& 0 \\ 1& 0& 0 \\ 1 & 0& -1 \end{pmatrix}.
\end{align}
The representation $\mathdutchcal{U}_{3,1}^{(4)}$ is faithful.

The Frobenius-Schur indicator of $\mathdutchcal{U}_{3,1}^{(j)}$, $j=1,2,3,4$, is equal to $1$, meaning that %complex representation has a real realization.. ma siccome la complessa è unitaria, la realizzazione reale sarà ortogonale
$\mathdutchcal{U}_{3,1}^{(j)}$ is equivalent to a representation of $G_3$ in $\mathrm{O}(4)$. The determinant on $\mathdutchcal{U}_{3,1}^{(j)}(G_3)$ takes values $\pm1$%strettamente in U(1) delle unitarie
. The irreps $\mathdutchcal{U}_{3,1}^{(3)}$ and $\mathdutchcal{U}_{3,1}^{(4)}$ are already in orthogonal form.
%*%*%*%*%*%*I think we need to write $\mathdutchcal{U}_{3,1}^{(1)}$ and $\mathdutchcal{U}_{3,1}^{(2)}$ in orthogonal form as well, by a change of basis using GAP/Mathematica. Even better, as said in Subsection~\ref{subsec:strirrpsGp}, one should induce the four $4$-dimensional orthogonal irreps of $G_3$ from the two stabiliser subgroups isomorphic to $C_2$.

\medskip

From the real representations of $G_3$, we aim at finding a universal set of orthogonal gates for quantum computing in an encoded real space (cf. Remark~\ref{eq:remBVrelaQC} and below).

We shall investigate whether or not the image elements of $\mathdutchcal{U}_{3,1}^{(j)}$ in $\U(4)$ (or $\mathrm{O}(4)$) factorise into the tensor product of two elements in $\U(2)$ (or $\mathrm{O}(2)$).  We can just focus on $\mathdutchcal{U}_{3,1}^{(2)}$ and $\mathdutchcal{U}_{3,1}^{(4)}$ by virtue of Eq.~\eqref{eq:equivtens}. %For each of them, we shall look for entangling two-qubit gates among the non-factorable unitaries, and extract one-qubit gates as the tensor-factors of the factorable unitaries.

%%%%%%%%%%%%%%%%%%%%%%%%%%%%%%%%%%%%%%%%%%%%%%%%%%%%%%%%%%%%%%%%%%%%%%%%%%%%%%%%%%%%%

\subsection{Gates from the representations \texorpdfstring{$\mathdutchcal{U}_{3,1}^{(2)}$}{Lg} and \texorpdfstring{$\mathdutchcal{U}_{3,1}^{(4)}$}{Lg}}\label{eq:gateresearch}
We start studying the gates emerging from the unitary representation $\mathdutchcal{U}_{3,1}^{(2)}$, thereby clarifying the reasoning and methods employed. Then, we will apply the same techniques to the orthogonal representation $\mathdutchcal{U}_{3,1}^{(4)}$. From our studies, we will not observe universal sets of gates from $\mathdutchcal{U}_{3,1}^{(2)}$ (at least not for $n_0\leq 4$ as in Definition~\ref{def:uinversalitu}), while $\mathdutchcal{U}_{3,1}^{(4)}$ will give a universal set of three orthogonal gates (two in $\mathrm{O}(2)$ and one in $\mathrm{O}(4)$), through which any circuit in $\U(2^n)$ can be approximated in  real encoding for every $n\geq2$.

First of all, we study the tensor-product factorisation of the transformations in $\mathdutchcal{U}_{3,1}^{(2)}(G_3)$, independently on the choice of the basis with respect to which the corresponding matrices are written, by working with their eigenvalues. In fact, a unitary transformation $\U\in\U(4)$ with (unimodular) eigenvalues $\{a_i\}_{i=1}^4$ admits a factorisation into the tensor product of two unitary %or antiunitary???? come sono in quel caso gli autovalori?
transformations $\mathrm{V}_1,\mathrm{V}_2\in\U(2)$ of (unimodular) eigenvalues $\{\lambda_i\}_{i=1,2}$, $\{\mu_i\}_{i=1,2}$, respectively, if and only if %corretto? La certezza è che NO SOLUZ A (20) -> NON FATTORIZZAZIONE, mentre SOLUZ A (20) -> fattorizzazione in principio / rispetto qualche isomorphismo di C4 con C2xC2 / rispetto qualche base... ma tutte le uintarie che fattorizzano in principio non possono condividere tutte la stessa base, come vedremo poi
there exist bases with respect to which $\diag(a_1,a_2,a_3,a_4) = \diag(\lambda_1\mu_1,\lambda_1\mu_2,\lambda_2\mu_1,,\lambda_2\mu_2)$, as it turns out from the Kronecker product of $\diag(\lambda_1,\lambda_2)$ with $\diag(\mu_1,\mu_2)$. Then, the issue is to calculate the eigenvalues $\{a_i\}_{i=1}^4$ of the unitaries image of $G_3$ with respect to the four-dimensional irreps, and check whether the system
\beq \label{eq:sistemfactteor}
\begin{cases}
\lambda_1\mu_1=a_1,\\
\lambda_1\mu_2=a_2,\\
\lambda_2\mu_1=a_3,\\
\lambda_2\mu_2=a_4,
\end{cases}
\eeq 
admits solution $(\lambda_1,\lambda_2,\mu_1,\mu_2)\in\U(1)^4$.

To do so, we translate the unitary matrices in $\mathdutchcal{U}_{3,1}^{(2)}(G_3)$ from GAP's to Wolfram Mathematica's language, where we write a code answering the above question~\cite{ourcode}. As outcome, $18$ elements %indici 46-63
of $\mathdutchcal{U}_{3,1}^{(2)}(G_3)<\U(4)$ do not tensor-factorise with respect to any basis.
In turn, out of the $72$ unitaries of $\mathdutchcal{U}_{3,1}^{(2)}(G_3)$, there are $54$ which admit a tensor-factorisation with respect to some isomorphism of $\CC^4$ as $\CC^2\ox \CC^2$ (i.e. with respect to some basis, or equivalence of representations, or similarity of matrices), though not all $54$ simultaneously with respect to the same isomorphism. In fact, with respect to a particular isomorphism $\CC^4\simeq \CC^2\ox \CC^2$, the subset of unitaries which tensor-factorise closes a subgroup of $\mathdutchcal{U}_{3,1}^{(2)}(G_3)$, whose order must divide $72$ by Lagrange's theorem, while $54\nmid72$. Therefore, the maximum number of unitaries in $\mathdutchcal{U}_{3,1}^{(2)}(G_3)$ which are simultaneously tensor-factorising is $36$.

We now investigate whether such a maximal subgroup of $36$ tensor-factorising unitaries exists within $\mathdutchcal{U}_{3,1}^{(2)}(G_3)\simeq G_3$, trying to maximise the number of factors in $\U(2)$, which will play the role of $1$-qubit gates.
% insieme, c'è poi da vedere se c'e' una entangled per fare il set di porte auspicabilmemte universale

The group $G_3$ has three normal subgroups of order $36$,
\begin{align}
&N_1\coloneqq \big\{L(\pm1,0,c,d,1),\, L(0,\pm1,c,d,-1),\textup{ for }c,d\in\ZZ/3\ZZ\big\},\nonumber\\
&N_2\coloneqq \big\{L(a,b,c,d,1),\textup{ for }(a,b)\in\{(\pm1,0),(0,\pm1)\},\ c,d\in\ZZ/3\ZZ\big\},\label{eq:G3normal36N}\\
&N_3\coloneqq\big\{L(\pm1,0,c,d,s),\textup{ for }s\in\{\pm1\},\,c,d\in\ZZ/3\ZZ\big\},\nonumber
\end{align}
no subgroups of order $24$, five subgroups of order $18$,
\begin{align}
&H_1\coloneqq\big\{L(1,0,c,d,s),\textup{ for }s\in\{\pm1\},\,c,d\in\ZZ/3\ZZ\big\},\nonumber\\
&H_2\coloneqq\big\{L(1,0,c,d,1),\,L(-1,0,c,d,-1),\textup{ for }c,d\in\ZZ/3\ZZ\big\},\nonumber\\
&H_3\coloneqq\big\{L(\pm1,0,c,d,1),\textup{ for }c,d\in\ZZ/3\ZZ\big\}=N_1\cap N_2\cap N_3,\qquad (\textup{normal})\\
&H_4\coloneqq\big\{L(1,0,c,d,1),\,L(0,1,c,d,-1),\textup{ for }c,d\in\ZZ/3\ZZ\big\},\nonumber\\
&H_5\coloneqq\big\{L(1,0,c,d,1),\,L(0,-1,c,d,-1),\textup{ for }c,d\in\ZZ/3\ZZ\big\},\nonumber
\end{align}
and twelve subgroups of order $12$,
\begin{align}
&K_1\coloneqq\big\{L(a,0,0,d,s),\textup{ for }a,s\in\{\pm1\},\,d\in\ZZ/3\ZZ\big\},\nonumber
%{37,38,39,46,47,48,55,56,57,64,65,66}
\\
&K_2\coloneqq\big\{L(\pm1,0,0,d,\pm1),\,L(\pm1,0,\pm1,d,\mp1),\textup{ for } d\in\ZZ/3\ZZ\big\} 
%{37,38,39,49,50,51,55,56,57,70,71,72}
,\nonumber\\
&K_3\coloneqq\big\{L(\pm1,0,0,d,\pm1),\,L(\pm1,0,\mp1,d,\mp1),\textup{ for }d\in\ZZ/3\ZZ\big\}
%{37,38,39,52,53,54,55,56,57,67,68,69}
,\nonumber\\
&K_4\coloneqq\big\{L(a,0,c,0,s),\textup{ for }a,s\in\{\pm1\},\,c\in\ZZ/3\ZZ\big\}
%{37,40,43,46,49,52,55,58,61,64,67,70}
,\nonumber\\
&K_5\coloneqq\big\{L(a,0,c,s(1-a),s),\textup{ for }a,s\in\{\pm1\},\,c\in\ZZ/3\ZZ\big\}
%{37,40,43,46,49,52,56,59,62,66,69,72}
,\nonumber\\
&K_6\coloneqq\big\{L(1,0,c,0,\pm1),\,L(-1,0,c,\pm1,\pm1),\textup{ for }c\in\ZZ/3\ZZ\big\}, \label{eq:subG3ord12}
%{37,40,43,46,49,52,57,60,63,65,68,71}
\\
&K_7\coloneqq\big\{L(\pm1,0,c,c,1),\,L(0,\pm1,c,c,-1),\,c\in\ZZ/3\ZZ\big\},\nonumber
%{1,5,9,28,32,36,37,41,45,64,68,72}
\\
&K_8\coloneqq\left\{\begin{aligned}
&L(0,1,c,c,-1),\,L(0,-1,c',d',-1),\,L(1,0,c,c,1),\,L(-1,0,-c',-d',1),\nonumber\\
& \textup{for } c\in\ZZ/3\ZZ,\,(c',d')\in\{(0,1),(1,-1),(-1,0)\}
\end{aligned}\right\},%{1,5,9,29,33,34,37,41,45,66,67,71}
\\
&K_9\coloneqq\left\{\begin{aligned}
&L(0,1,c,c,-1),\,L(0,-1,-c',-d',-1),\,L(1,0,c,c,1),\,L(-1,0,c',d',1),\nonumber\\
& \textup{for } c\in\ZZ/3\ZZ,\,(c',d')\in\{(0,1),(1,-1),(-1,0)\}
\end{aligned}\right\},%{1,5,9,30,31,35,37,41,45,65,69,70}
\\ &K_{10}\coloneqq\left\{L(a,b,c,d,s)\midd (a,b,s)\in\{(\pm1,0,1),(0,\pm1,-1)\},\, c+d\equiv0\right\},\nonumber%{1,6,8,28,33,35,37,42,44,64,69,71}
\\ & K_{11}\coloneqq\left\{\begin{aligned}&L(0,1,c,d,-1)\midd c+d\equiv1,\\
& L(0,-1,c,d,-1),L(1,0,c,d,1)\midd c+d\equiv0,\,L(-1,0,c,d,1)\midd c+d\equiv-1\end{aligned}\right\},\nonumber %{2,4,9,28,33,35,37,42,44,66,68,70}
\\ & K_{12}\coloneqq \left\{\begin{aligned}&L(0,1,c,d,-1)\midd c+d\equiv-1,\\
& L(0,-1,c,d,-1),L(1,0,c,d,1)\midd c+d\equiv0,\,L(-1,0,c,d,1)\midd c+d\equiv1\end{aligned}\right\}
.\nonumber %{3,5,7,28,33,35,37,42,44,65,67,72}
\end{align}

In particular, $\mathdutchcal{U}_{3,1}^{(2)}(N_3)$ contains the $18$ non-tensor-product unitaries of $\mathdutchcal{U}_{3,1}^{(2)}(G_3)$, so this cannot be a maximal subgroup of factorising unitaries. We are left to check whether either one of the subgroups $\mathdutchcal{U}_{3,1}^{(2)}(N_1)$, $\mathdutchcal{U}_{3,1}^{(2)}(N_2)$ is composed entirely of tensor product unitaries with respect to the same basis.

We approach this problem developing our Mathematica code~\cite{ourcode}, as follows: for each unitary in $\mathdutchcal{U}_{3,1}^{(2)}(G_3)$ (most of them has degenerate spectrum), we find a basis of eigenvectors, write all the elements in $\mathdutchcal{U}_{3,1}^{(2)}(G_3)$ with respect to that eigenbasis, and check whether each of them is a Kronecker product. We find that there are (at least) the following subgroups composed each of factorised unitaries with respect to the same basis: 
\begin{itemize}
    \item the normal subgroup $\mathdutchcal{U}_{3,1}^{(2)}(N_1)$ of order $36$, %\\up to isomorphisms, given by a suitable choice of an eigenbasis of the first nine matrices in matricesU2; %ste matrici hanno spettro DEGENERE
    \item the normal subgroup $\mathdutchcal{U}_{3,1}^{(2)}(H_3%=N18=H_{tot}
    )$  of order $18$,%\\ up to isomorphisms, given by a suitable choice of an eigenbasis of the elements of matricesU2 at positions $38,39,40,42,43,44$; %ste matrici hanno spettro DEGENERE
    \item the subgroups $\mathdutchcal{U}_{3,1}^{(2)}(K_7),\,\mathdutchcal{U}_{3,1}^{(2)}(K_8),\,\mathdutchcal{U}_{3,1}^{(2)}(K_9)$ of $\mathdutchcal{U}_{3,1}^{(2)}(N_1)$ of order $12$%, \\ up to isomorphisms, given by a suitable choice of an eigenbasis of the elements of matricesU2 at positions $29,30,31,32,34,36,64,65,66,67,68,69,70,71,72$. %e ste matrici hanno spettro NON degenere!! le prime e DEGENERE le ultime
    .
\end{itemize}

At this point, for each of these subgroups, with respect to the related eigenbasis, we collect the Kronecker factors of the factorising unitaries%delating duplicates, also up to global phase
, and we analyse whether the non-tensor-product unitaries in the complement are entangling or product up to $\mathrm{SWAP}$.

Starting from $\mathdutchcal{U}_{3,1}^{(2)}(N_1)$, we get that there is no entangling $2$-qubit gate among the non-tensor-product unitaries in $\mathdutchcal{U}_{3,1}^{(2)}(G_3\setminus N_1)$ with respect to the related eigenbasis. But an entangling gate is necessary for a universal set of $3$-adically controlled gates. %This is somehow unfortunate, because with a maximum subgroup like $\mathdutchcal{U}_{3,1}^{(2)}(N_1)$ we were trying to maximise the number of Kronecker factors, i.e., the number of one-qubit gates. But it could also be positive, since this scenario minimises the non-factorable two-qubit gates just to one representative: the subgroup $\mathdutchcal{U}_{3,1}^{(2)}(N_1)$ is of index $2$ in $\mathdutchcal{U}_{3,1}^{(2)}(G_3)$, then it is a normal subgroup and $\mathdutchcal{U}_{3,1}^{(2)}(G_3)=\mathdutchcal{U}_{3,1}^{(2)}(N_1)\sqcup \mathrm{U}_1\mathdutchcal{U}_{3,1}^{(2)}(N_1)$ for some $\mathrm{U}_1\in \mathdutchcal{U}_{3,1}^{(2)}(G_3)\setminus \mathdutchcal{U}_{3,1}^{(2)}(N_1)$, meaning that the set $\mathrm{U}_1\mathdutchcal{U}_{3,1}^{(2)}(N_1)$ of non-factorable gates is given by just one non-factorable gate $\mathrm{U}_1$ multiplied by the factorable ones.

Then, we move on and consider the subgroup $\mathdutchcal{U}_{3,1}^{(2)}(H_3)$. 
It is $H_3\simeq \mathdutchcal{U}_{3,1}^{(2)}(H_3) \simeq (C_3\times C_3)\rtimes C_2$, with quotient group
\beq \label{Klein4iso}
G_3/H_3\simeq \mathdutchcal{U}_{3,1}^{(2)}\left(G_3/H_3\right)\simeq C_2\times C_2,
\eeq
the Klein four-group of two generators, say $f_1$ and $f_2$. With respect to the basis
\beq\label{eq:autobase38Last18} %autobase della matrice 38 con relativi autovalori
\left\{\mathbf{v}_1\coloneqq\begin{pmatrix}0\\0\\1\\0\end{pmatrix},\ %e^{-\frac{2}{3}\pi i}
\mathbf{v}_2\coloneqq\begin{pmatrix}0\\0\\0\\-i\end{pmatrix},\ %1
\mathbf{v}_3\coloneqq\begin{pmatrix}\frac{e^{-\frac{i \pi}{6}}}{\sqrt{2}}\\\frac{1}{\sqrt{2}}\\0\\0\end{pmatrix},\ %1
\mathbf{v}_4\coloneqq\begin{pmatrix}\frac{e^{\frac{5i\pi}{6}}}{\sqrt{2}}\\\frac{1}{\sqrt{2}}\\0\\0\end{pmatrix} %e^{\frac{2}{3}\pi i}
\right\},
\eeq 
(written in terms of the basis GAP chose to write Eq.~\eqref{repnU2chi9daGAP}), the irrep $\mathdutchcal{U}_{3,1}^{(2)}$ is redefined as follows:
\begin{align}
&\mathdutchcal{U}_{3,1}^{(2)}\colon G_3\rightarrow \U(4),\label{repnU2chi9wrt38}\\
& \begin{pmatrix} 0& 1& 0 \\ -1& 0& 0 \\ 0& 0& 1 \end{pmatrix}   %U10
\mapsto  U_{10}\coloneqq
\begin{pmatrix} 0 & 0 & e^{- \frac{5i\pi}{6}} & 0 \\ -i & 0 & 0 & 0 \\ 0 & 0 & 0 & e^{- \frac{2i\pi}{3}} \\ 0 & 1 & 0 & 0 \end{pmatrix},
\quad 
\begin{pmatrix} 0& 1& 0 \\ 1& 0& 0 \\ 1 & 0& -1 \end{pmatrix}   %U4
\mapsto U_4\coloneqq \begin{pmatrix}
0 & 0 & i & 0 \\
0 & 0 & 0 & 1 \\
e^{\frac{5i\pi}{6}} & 0 & 0 & 0 \\
0 & e^{\frac{2i\pi}{3}} & 0 & 0
\end{pmatrix}.
\nonumber
\end{align}
As by Eq.~\eqref{Klein4iso}, we find the four (left and right) cosets of $\mathdutchcal{U}_{3,1}^{(2)}(H_3)$ in $\mathdutchcal{U}_{3,1}^{(2)}(G_3)$:
\begin{itemize}
\item $\mathdutchcal{U}_{3,1}^{(2)}(H_3)$ itself of factorising unitaries, with representative $\mI_4$, corresponding to $e$ in $C_2\times C_2$; %matrici [37-45]\cap[64-72]
\item a coset of $18$ entangling unitaries, denoted by $\mathdutchcal{U}_{3,1}^{(2)}(Ent_1)$ where $Ent_1\coloneqq\big\{M(0,\pm1,c,d,1),$\\$\textup{for }c,d\in\ZZ/3\ZZ\big\}$, 
%matrici [37-45]\cap[10-27]
a representative of which is $U_{10}$, which corresponds to, say, $f_1$ in $C_2\times C_2$;
\item a coset of $18$ entangling unitaries, denoted by $\mathdutchcal{U}_{3,1}^{(2)}(Ent_2)$ where $Ent_2\coloneqq\big\{M(0,\pm1,c,d,-1),$\\$\textup{for }c,d\in\ZZ/3\ZZ\big\}$,  %matrici [1,9]\cap[28,36]
such that $\mathdutchcal{U}_{3,1}^{(2)}(Ent_1)\sqcup \mathdutchcal{U}_{3,1}^{(2)}(Ent_2)$ is the set of all the $36$ entangling unitaries in $\mathdutchcal{U}_{3,1}^{(2)}(G_3)$ with respect to the basis~\eqref{eq:autobase38Last18}. A representative of $\mathdutchcal{U}_{3,1}^{(2)}(Ent_2)$ is $U_4$, and corresponds to $f_2$ in $C_2\times C_2$;
\item a coset of unitaries, denoted by $\mathdutchcal{U}_{3,1}^{(2)}(S)$ where $S\coloneqq\big\{M(\pm1,0,c,d,-1),\textup{ for }c,d\in\ZZ/3\ZZ\big\}$, %[46,63]
which are both non-entangling and non-tensor-product with respect to the basis~\eqref{eq:autobase38Last18}. A representative of this set is
\beq \label{eq:theSWAP}
    %U_{49}\coloneqq
    \mathrm{SWAP}=\begin{pmatrix} 1 & 0 & 0 & 0 \\ 0 & 0 & 1 & 0 \\ 0 & 1 & 0 & 0 \\ 0 & 0 & 0 & 1 \end{pmatrix}=U_{10}U_4,
    \eeq 
and corresponds to $f_1f_2$ in $C_2\times C_2$. 
\end{itemize}
By the group structure, one has $\mathdutchcal{U}_{3,1}^{(2)}(Ent_1)\mathdutchcal{U}_{3,1}^{(2)}(Ent_2)=\mathdutchcal{U}_{3,1}^{(2)}(S)$, and 
for every $E_1 \in \mathdutchcal{U}_{3,1}^{(2)}(Ent_1)$ and $E_2 \in \mathdutchcal{U}_{3,1}^{(2)}(Ent_2)$ one has $E_1E_2  = %H_3^{(10)}U_{10}H_3^{(4)}U_4 = H_3^{(10)}U_{10}H_3^{(4)}U_{10}^{-1}U_{10}U_4=H_3^{(10)}H_3'U_{10}U_4 per normalità
(A\otimes B)\mathrm{SWAP}$ for some $A,B\in\U(2)$ such that $A\otimes B\in \mathdutchcal{U}_{3,1}^{(2)}(H_3)$. The SWAP operator allows us to treat the single factors of the matrices in $\mathdutchcal{U}_{3,1}^{(2)}(H_3)$ as $1$-qubit gates needing to track which of the two $\CC^2$ subspaces they originate from, and to consider all the $1$- and $2$-qubit gates without worrying about which qubits in a register they are applied to. 

Furthermore, we collect all the tensor-factors in $\U(2)$ coming from the unitaries in $\mathdutchcal{U}_{3,1}^{(2)}(H_3)$, 
we consider the group they generate%of order 72 ma non è G3, generare da tutte o da up to phase sole è uguale qui
, and find a minimal set of its generators:
\beq \label{eq:gen1qubitgates}
\left\{ A\coloneqq \begin{pmatrix}
1 & 0 \\
0 & e^{\frac{2 i \pi}{3}}
\end{pmatrix}, \quad
B\coloneqq \begin{pmatrix}
0 & 1 \\
e^{\frac{5 i \pi}{6}} & 0
\end{pmatrix} \right\}.
\eeq 
%One can check that the $X$ and $Y$ Pauli matrices are not in this generated group, however the $Z$ Pauli matrix is... nella precedente autobase della 38* matrice data in automatico da Mathematica, invece, avevamo Pauli $X$

We collect the generators of the $1$-qubit gates as in Eq.~\eqref{eq:gen1qubitgates}, together with the entangling $2$-qubit gates as in Eq.~\eqref{repnU2chi9wrt38}%, as well as with the SWAP gate from Eq.~\eqref{eq:theSWAP}.. non ce la metto perche' viene dal prodotto delle altre 2 entangling... e mi ricordo sempre quindi che posso swappare in un circuito di porte
; hence, we find the following set of $3$-adically controlled gates: $\left\{
A %=\begin{pmatrix} 1 & 0 \\ 0 & e^{\frac{2 i \pi}{3}} \end{pmatrix}
, \
B%=\begin{pmatrix} 0 & 1 \\ e^{\frac{5 i \pi}{6}} & 0 \end{pmatrix}
, \
U_4%=\begin{pmatrix} 0 & 0 & i & 0 \\ 0 & 0 & 0 & 1 \\ e^{\frac{5i\pi}{6}} & 0 & 0 & 0 \\ 0 & e^{\frac{2i\pi}{3}} & 0 & 0 \end{pmatrix}
, \
U_{10}%=\begin{pmatrix} 0 & 0 & e^{- \frac{5i\pi}{6}} & 0 \\ -i & 0 & 0 & 0 \\ 0 & 0 & 0 & e^{- \frac{2i\pi}{3}} \\ 0 & 1 & 0 & 0 \end{pmatrix}
\right\}$.
%E' minimale o serve togliere altro? $G_3$ generato da 2 elmts, uno puo essere preso fattorizzabile.. ma vedendo la lista delle possibili coppie di generatori, non ce n'è una che vede U10 o U4 insieme ad un elemento fattorizzato qui (38-45 o 64-72)

Clearly $\mathdutchcal{U}_{3,1}^{(2)}(H_3)$ is a subgroup of $\langle A\otimes \mI_2,\, \mI_2\otimes A, \,B\otimes \mI_2,\, \mI_2\otimes B\rangle$. The latter group is of finite order%$432$
; also the group $\langle A\otimes \mI_2,\, \mI_2\otimes A, \,B\otimes \mI_2,\, \mI_2\otimes B,\, U_4, \,U_{10}\rangle<\U(4)$ is of finite order%$10368$, %of which $\mathdutchcal{U}_{3,1}^{(2)}(G_3)$ is a proper subgroup
, while a group must be countable to be dense in $\U(4)$. Seeking for universality, we add a third ancillary qubit, however the group $\left\langle A\otimes \mI_4, \mI_2\otimes A\otimes \mI_2, \mI_4\otimes A, B\otimes \mI_4, \mI_2\otimes B\otimes \mI_2, \mI_4\otimes B,  \mI_2\otimes U_4, U_4\otimes \mI_2, \mI_2\otimes U_{10},\right.$\\$\left.U_{10}\otimes \mI_2\right\rangle <\U(8)$ has finite order%$13436928$
. The group generated by $\{A,B,U_4,U_{10}\}$ in dimension $2^4$ in a similar fashion is again finite%139314069504
. This is not a universal set, at least not for $n_0\leq 4$ in Definition~\ref{def:uinversalitu}.

\medskip
We now proceed with the analysis of the orthogonal transformations in $\mathdutchcal{U}_{3,1}^{(4)}(G_3)$, using methods and codes~\cite{ourcode} analogous to those employed for the unitaries in $\mathdutchcal{U}_{3,1}^{(2)}(G_3)$. 

By studying the factorisation from the spectrum, we get that $18$ elements of $\mathdutchcal{U}_{3,1}^{(4)}(G_3)$ %indici 1-9, 28-36
do not tensor-factorise with respect to any basis.
\begin{comment}
\begin{itemize}
    \item the six unitaries sharing eigenvalue $1$ of multiplicity three and eigenvalue $-1$;
    \item the twelve unitaries sharing the set of eigenvalues $\{1,-1,e^{\pm i\frac{2}{3}\pi}\}$.
\end{itemize}
These unitaries are different and correspond to different elements in $G_3$, with respect to the scenario for $\mathdutchcal{U}_{3,1}^{(2)}(G_3)$. 
\end{comment}
In turn, there may be a subgroup of $\mathdutchcal{U}_{3,1}^{(4)}(G_3)$ of factorising unitaries with respect to the same isomorphism $\CC^4\simeq \CC^2\ox\CC^2$ at most of order $36$. The subgroup $\mathdutchcal{U}_{3,1}^{(4)}(N_1)$ has to be excluded, as it contains those $18$ non-tensor-product unitaries found above. 

However, from our code in Mathematica, we detect subgroups of order at maximum $12$ of factorising unitaries in $\mathdutchcal{U}_{3,1}^{(4)}(G_3)$ with respect to a same (eigen)basis; in particular 
$\mathdutchcal{U}_{3,1}^{(4)}(K_1)$, %w.r.t. 33                %{37,38,39,46,47,48,55,56,57,64,65,66} 
$\mathdutchcal{U}_{3,1}^{(4)}(K_2)$, %w.r.t. 35 (e38,39,71,72) %{37,38,39,49,50,51,55,56,57,70,71,72} 
$\mathdutchcal{U}_{3,1}^{(4)}(K_3)$, %w.r.t. 1 (e 69)          %{37,38,39,52,53,54,55,56,57,67,68,69} 
$\mathdutchcal{U}_{3,1}^{(4)}(K_4)$, %w.r.t. 40 (e43,64,67)    %{37,40,43,46,49,52,55,58,61,64,67,70} 
$\mathdutchcal{U}_{3,1}^{(4)}(K_6)$ %w.r.t. 5 (e 65)          %{37,40,43,46,49,52,57,60,63,65,68,71}
[cf. Eq.~\eqref{eq:subG3ord12}]. These five subgroups are not normal in $\mathdutchcal{U}_{3,1}^{(4)}(G_3)$, and they are all subgroups of $\mathdutchcal{U}_{3,1}^{(4)}(N_3)$.

\medskip

With respect to the basis
\beq\label{eq:autobase1U4} %autobase della matrice 1 con relativi autovalori
\left\{\mathbf{v}_1\coloneqq\begin{pmatrix}-\frac{\sqrt{3}}{2}\\0\\0\\\frac{1}{2}\end{pmatrix},\ \mathbf{v}_2\coloneqq\begin{pmatrix}\frac{1}{2}\\0\\0\\\frac{\sqrt{3}}{2}\end{pmatrix},\ 
\mathbf{v}_3\coloneqq\begin{pmatrix}0\\0\\1\\0\end{pmatrix},\ \mathbf{v}_4\coloneqq\begin{pmatrix}0\\1\\0\\0\end{pmatrix}
\right\},
\eeq 
(written in terms of the basis GAP chose to write Eq.~\eqref{repnU4chi8daGAP}), the  irrep $\mathdutchcal{U}_{3,1}^{(4)}$ is redefined as follows:
\begin{align}
&\mathdutchcal{U}_{3,1}^{(4)}\colon G_3\rightarrow \U(4),\label{repnU4chi8wrt1}\\
&\begin{pmatrix} 0& 1& 0 \\ -1& 0& 0 \\ 0& 0& 1 \end{pmatrix}   %U10
\mapsto 
\begin{pmatrix}
0 & 0 & \frac{1}{2} & \frac{\sqrt{3}}{2} \\
0 & 0 & -\frac{\sqrt{3}}{2} & \frac{1}{2} \\
-\frac{1}{2} & -\frac{\sqrt{3}}{2} & 0 & 0 \\
-\frac{\sqrt{3}}{2} & \frac{1}{2} & 0 & 0
\end{pmatrix},
\quad 
\begin{pmatrix} 0& 1& 0 \\ 1& 0& 0 \\ 1 & 0& -1 \end{pmatrix}   %U4
\mapsto \begin{pmatrix}
\frac{1}{2} & -\frac{\sqrt{3}}{2} & 0 & 0 \\
-\frac{\sqrt{3}}{2} & -\frac{1}{2} & 0 & 0 \\
0 & 0 & -\frac{1}{2} & -\frac{\sqrt{3}}{2} \\
0 & 0 & \frac{\sqrt{3}}{2} & -\frac{1}{2}
\end{pmatrix}
.
\nonumber
\end{align}
%All the unitaries in $\mathdutchcal{U}_{3,1}^{(4)}(G_3)$ are distinct, also up to global phase factors. However, 
In this case $\mathrm{SWAP}\not\in \mathdutchcal{U}_{3,1}^{(4)}(G_3)$, and the group of unitaries $\mathdutchcal{U}_{3,1}^{(4)}(G_3)$ is not invariant under $\mathrm{SWAP}$, i.e., it is not true that for every $U\in \mathdutchcal{U}_{3,1}^{(4)}(G_3)$ there exists $U'\in \mathdutchcal{U}_{3,1}^{(4)}(G_3)$ such that $\mathrm{SWAP}\,U\,\mathrm{SWAP} \propto U'$. 

%Eq.~\eqref{eq:autobase1U4} shows an eigenbasis of the first matrix in the list $\texttt{matricesU4}$, with respect to which 
With respect to the basis~\eqref{eq:autobase1U4}, $\mathdutchcal{U}_{3,1}^{(4)}(K_3)\simeq \mathrm{D}_6$ is a non-normal subgroup of order $12$ of $\mathdutchcal{U}_{3,1}^{(4)}(G_3)$ of all factorising unitaries. We collect the generating $1$-qubit gates 
%\beq \left\{\ X:=\begin{pmatrix} 0 & 1 \\ 1 & 0 \end{pmatrix}, \quad \begin{pmatrix} -\frac{1}{2} & \frac{\sqrt{3}}{2} \\ \frac{\sqrt{3}}{2} & \frac{1}{2} \end{pmatrix} \right\}, \eeq 
among the Kronecker tensor factors from $\mathdutchcal{U}_{3,1}^{(4)}(K_3)$.
\begin{comment}
\begin{align} \label{eq:factorsU4wrt1}
&\mathrm{I}_2, \ 
X:=\begin{pmatrix} 0 & 1 \\ 1 & 0 \end{pmatrix}, \ 
Z:=\begin{pmatrix} 1 & 0 \\ 0 & -1 \end{pmatrix}, \ 
\begin{pmatrix} -\frac{1}{2} & \frac{\sqrt{3}}{2} \\ \frac{\sqrt{3}}{2} & \frac{1}{2} \end{pmatrix}, \ 
\begin{pmatrix} \frac{1}{2} & -\frac{\sqrt{3}}{2} \\ \frac{\sqrt{3}}{2} & \frac{1}{2} \end{pmatrix}, \
\begin{pmatrix} \frac{1}{2} & \frac{\sqrt{3}}{2} \\ -\frac{\sqrt{3}}{2} & \frac{1}{2} \end{pmatrix}, \ 
\begin{pmatrix} \frac{1}{2} & \frac{\sqrt{3}}{2} \\ \frac{\sqrt{3}}{2} & -\frac{1}{2} \end{pmatrix}, \\
&-\mathrm{I}_2,\ 
-Z, \
-\begin{pmatrix} -\frac{1}{2} & \frac{\sqrt{3}}{2} \\ \frac{\sqrt{3}}{2} & \frac{1}{2} \end{pmatrix}, \
-\begin{pmatrix} \frac{1}{2} & -\frac{\sqrt{3}}{2} \\ \frac{\sqrt{3}}{2} & \frac{1}{2} \end{pmatrix}, \
-\begin{pmatrix} \frac{1}{2} & \frac{\sqrt{3}}{2} \\ -\frac{\sqrt{3}}{2} & \frac{1}{2} \end{pmatrix}, \ 
-\begin{pmatrix} \frac{1}{2} & \frac{\sqrt{3}}{2} \\ \frac{\sqrt{3}}{2} & -\frac{1}{2} \end{pmatrix}.\nonumber
\end{align}
\end{comment}
The elements in $\mathdutchcal{U}_{3,1}^{(4)}(G_3\setminus K_3)$ are all entangling unitaries. We can pick a representative for each of the non-trivial five left cosets of $\mathdutchcal{U}_{3,1}^{(4)}(K_3)$ in $\mathdutchcal{U}_{3,1}^{(4)}(G_3)$, however it is enough to take just one entangling unitary such that it generates $\mathdutchcal{U}_{3,1}^{(4)}(G_3)$ together with a product unitary.
%le generating one-qubit gates X, A1 generano tutte le porte 1-qubit (gruppo d'ord 12) di G3 e più.. e quindi ovviamente A1xI, IxA1, XxI, IxX generano anche di più di quel sottogruppo d'ordine 12... basta allora aggiungere a questo una 4x4 entangling U tale che U con una fattorizzabile formano un set minimale di 2 generatori di U4(G3), perchè con le suddette 1-qubit gates ottengo sicuro tale fattorizzabile, che insieme alla entangling generano sicuro tutto U4(G3), e quindi le altre entangling sono già ottenuto e non devono essere messe nel set di porte logiche perchè saranno già ottenibili, quindi sono generatori extra non necessari in un set minimo...
%Through GAP, we see that a minimal set of generators of $\mathdutchcal{U}_{3,1}^{(4)}(G_3)\simeq G_3$ is given by the elements at position $53$ (factorable) and $1$ (entangled from Coset 4), and this is what we choose. 
We implicitly include $\mathrm{SWAP}$, i.e., we allow ourselves to apply gates to any qubits in a circuit. In this way, we propose the following set of $3$-adically controlled gates:
\beq \label{eq:1gatesU4wrt1}
\mathcal{G}_1\mathdutchcal{p}_3\coloneqq\left\{ X\coloneqq \begin{pmatrix} 0 & 1 \\ 1 & 0 \end{pmatrix},\ %A_1=
S\coloneqq\begin{pmatrix} -\frac{1}{2} & \frac{\sqrt{3}}{2} \\ \frac{\sqrt{3}}{2} & \frac{1}{2} \end{pmatrix},\ %U1 =
\widetilde{CZ}\coloneqq\begin{pmatrix}
-1 & 0 & 0 & 0 \\
0 & 1 & 0 & 0 \\
0 & 0 & 1 & 0 \\
0 & 0 & 0 & 1
\end{pmatrix}%,\,\mathrm{SWAP}
\begin{comment} %sennò, con gli altri rappresentanti di coset avrei preso anche queste
,\ \begin{pmatrix}
\frac{1}{2} & -\frac{\sqrt{3}}{2} & 0 & 0 \\
-\frac{\sqrt{3}}{2} & -\frac{1}{2} & 0 & 0 \\
0 & 0 & -\frac{1}{2} & -\frac{\sqrt{3}}{2} \\
0 & 0 & \frac{\sqrt{3}}{2} & -\frac{1}{2}
\end{pmatrix} %=U4
,\ 
\begin{pmatrix}
\frac{1}{2} & \frac{\sqrt{3}}{2} & 0 & 0 \\
\frac{\sqrt{3}}{2} & -\frac{1}{2} & 0 & 0 \\
0 & 0 & -\frac{1}{2} & \frac{\sqrt{3}}{2} \\
0 & 0 & -\frac{\sqrt{3}}{2} & -\frac{1}{2}
\end{pmatrix} %=U7
, \ \begin{pmatrix}
-\frac{1}{2} & -\frac{\sqrt{3}}{2} & 0 & 0 \\
\frac{\sqrt{3}}{2} & -\frac{1}{2} & 0 & 0 \\
0 & 0 & -\frac{1}{2} & \frac{\sqrt{3}}{2} \\
0 & 0 & -\frac{\sqrt{3}}{2} & -\frac{1}{2}
\end{pmatrix}%=U40
,\quad 
\begin{pmatrix}
-\frac{1}{2} & \frac{\sqrt{3}}{2} & 0 & 0 \\
-\frac{\sqrt{3}}{2} & -\frac{1}{2} & 0 & 0 \\
0 & 0 & -\frac{1}{2} & -\frac{\sqrt{3}}{2} \\
0 & 0 & \frac{\sqrt{3}}{2} & -\frac{1}{2}
\end{pmatrix},%=U43
\end{comment}
\right\}.
\eeq 
In this set, we find the Pauli matrix $X$ (aka $\mathrm{NOT}$%che mappa \ket{0}<->\ket{1}
), the improper rotation %(i.e. orthogonal matrix with determinant $-1$) 
$S=R(2\pi/3)Z$ where the Pauli $Z\coloneqq\begin{pmatrix} 1 & 0\\ 0 & -1 \end{pmatrix}$ is a reflection and $R(\theta)\coloneqq\begin{pmatrix}\cos\theta & -\sin\theta \\ \sin\theta & \cos\theta
\end{pmatrix}$ is a (proper) rotation, and a sort of controlled $Z$%se \widetilde{Z}\coloneqq\begin{pmatrix}-1&0\\0&1\end{pmatrix} allora \widetilde{Z}\ket{0}=-\ket{0} e \widetilde{Z}\ket{1}=\ket{1}.. e qui il controllo può essere visto sul 2° qubit: se il 2° è in \ket{0} allora si agisce con \widetilde{Z} sul primo, sennò no
. By GAP, it is easy to see that $\langle X\otimes \mI_2,\, \mI_2\otimes X,\, S\otimes \mI_2,\mI_2\otimes S, \widetilde{CZ}%,\mathrm{SWAP}
\rangle \leq \U(4)$ is an infinite group. 

The situation is very similar for the groups $\mathdutchcal{U}_{3,1}^{(4)}(K_1)$, $\mathdutchcal{U}_{3,1}^{(4)}(K_2)$ and $\mathdutchcal{U}_{3,1}^{(4)}(K_6)$, leading to the same set of $3$-adically controlled gates $\mathcal{G}_1\mathdutchcal{p}_3$ in Eq.~\eqref{eq:1gatesU4wrt1}.

\medskip

Another option is the following. With respect to the basis
\beq\label{eq:autobase40U4} %autobase della matrice 40 con relativi autovalori
\left\{\mathbf{v}_1\coloneqq\begin{pmatrix}-\frac{i}{\sqrt{2}}\\0\\0\\\frac{1}{\sqrt{2}}\end{pmatrix},\ \mathbf{v}_2\coloneqq\begin{pmatrix}0\\ \frac{i}{\sqrt{2}}\\ \frac{1}{\sqrt{2}} \\ 0 \end{pmatrix},\ 
\mathbf{v}_3\coloneqq\begin{pmatrix} \frac{i}{\sqrt{2}} \\0\\0\\ \frac{1}{\sqrt{2}} \end{pmatrix},\ \mathbf{v}_4\coloneqq\begin{pmatrix}0\\- \frac{i}{\sqrt{2}}\\\frac{1}{\sqrt{2}}\\0\end{pmatrix}
\right\},
\eeq 
% the unitary irrep $\mathdutchcal{U}_{3,1}^{(4)}$ is redefined as follows: \begin{align} &\mathdutchcal{U}_{3,1}^{(4)}\colon G_3\rightarrow \U(4),\label{repnU4chi8wrt40}\\ &\begin{pmatrix} 0& 1& 0 \\ -1& 0& 0 \\ 0& 0& 1 \end{pmatrix}  U10 \mapsto \begin{pmatrix} 0 & 0 & 0 & e^{-\frac{2 \pi i}{3}} \\ -1 & 0 & 0 & 0 \\ 0 & e^{\frac{2 \pi i}{3}} & 0 & 0 \\ 0 & 0 & -1 & 0 \end{pmatrix}, \quad \begin{pmatrix} 0& 1& 0 \\ 1& 0& 0 \\ 1 & 0& -1 \end{pmatrix}  \mapsto \begin{pmatrix} 0 & 0 & -1 & 0 \\ 0 & e^{-\frac{2 \pi i}{3}} & 0 & 0 \\ -1 & 0 & 0 & 0 \\ 0 & 0 & 0 & e^{\frac{2 \pi i}{3}} \end{pmatrix} . \nonumber \end{align} Here, 
the group $\mathdutchcal{U}_{3,1}^{(4)}(K_4)$ is formed by all factorising unitaries, and the elements of $\mathdutchcal{U}_{3,1}^{(4)}(G_3\setminus K_4)$ are all entangling unitaries.
%%%%%%Vedi meglio invarianza sotto SWAP qui
% all $1$-qubit gates: \begin{align}\label{eq:factorsU4wrt40} &\mI_2,\ X,\ \begin{pmatrix} 1 & 0 \\ 0 & e^{\pm \frac{2 \pi i}{3}} \end{pmatrix},\ \begin{pmatrix} 0 & 1 \\ e^{\pm\frac{2 \pi i}{3}} & 0 \end{pmatrix}\\ & \nonumber -\mI_2, \ e^{\pm \frac{2 \pi i}{3}}\mI_2,\ e^{\pm \frac{\pi i}{3}}\mI_2, \ -X,\ e^{\pm \frac{2 \pi i}{3}}X,\ e^{\pm \frac{\pi i}{3}}X \end{align}
%generatori $1$-qubit gates non possono presi tra quelle a meno di fase...
%genero tutto il gruppo G3 con matrici 9 (ent, coset 5) e 58 (fattorizzata)
Arguing as in the above cases, we can give the following set of $3$-adically controlled gates:
\beq 
\left\{ -X,\ \begin{pmatrix}
0 & 1 \\
e^{-\frac{2 \pi i}{3}} & 0
\end{pmatrix},\ 
\begin{pmatrix}
0&0&-1&0\\
0&1&0&0\\
-1&0&0&0\\
0&0&0&1
\end{pmatrix}%,\,\mathrm{SWAP}
\right\}.
\eeq 
However, in dimension $2^n$ ($n=2,3,4$) this set of gates generates finite groups.
%1296, 839808, 88159684608  .. oppure 3888, 17635968$ con SWAP

\subsection{A universal set of $3$-adically controlled quantum gates}
\label{eq:gatesetuniversal}
From the previous section, we have found a promising set of $3$-adically controlled gates, i.e., $\mathcal{G}_1\mathdutchcal{p}_3$ in Eq.~\eqref{eq:1gatesU4wrt1}. Here, we prove that $\mathcal{G}_1\mathdutchcal{p}_3$ is a universal set of gates. 

Let %$\mathcal{G}\mathdutchcal{s}$ or 
\beq 
\mathdutchcal{G}\coloneqq\langle X\otimes\mI_2,\mI_2\otimes X, S\otimes\mI_2,\mI_2\otimes S,\widetilde{CZ}\rangle\leq \mathrm{O}(4)
\eeq 
be the group generated by the gates in $\mathcal{G}_1\mathdutchcal{p}_3$ in dimension $4$. First of all, we show that $\mathdutchcal{G}$ is dense in $\mathrm{O}(4)$, or equivalently that its connected component at the identity, $\mathdutchcal{G}_0$, is dense in $\SO(4)$.

\medskip

For every $m\in\NN$, $\mathrm{O}(m)\simeq\SO(m)\rtimes\{\pm1\}$. Indeed $\det(\mathrm{O}(m))=\{\pm1\}$ and $\mathrm{O}(m)$ has two connected components: $\SO(m)$ and the set of orthogonal transformations with determinant equal to $-1$.  The Lie algebra associated to $\mathrm{O}(m)$ and $\SO(m)$ is $\mathfrak{o}(m)=\mathfrak{so}(m)=\{A\in \mathsf{M}(m,\RR)\midd A^\top=-A\}$ %authomatically traceless
of dimension $\frac{m(m-1)}{2}$. 

We have that $\det(\overline{\mathdutchcal{G}})=\{\pm1\}$, and $\overline{\mathdutchcal{G}}$ also %non connesso, perchè if f on X continuous and X connected then f(X) connected. For me det continuous but \det(\mathcal{G}_3) disconnected, so \mathcal{G}_3 disconnected... since continuous surjective f is such that #CC[f(X)]\leq#CC[X], \leq#CC[X]>1 X not connected, and #CC[f(X)]=2, then we can only deduce \leq#CC[X]\geq 2 but do not know if exactly two connected components.. ma posso dire che det è anche homomorphismo gruppale e allora posso dedurre che ho esattamente 2 CC:
has two connected components: $\overline{\mathdutchcal{G}}_0=\ker(\det)\leq \SO(4)$, and $\det^{-1}(-1)=\widetilde{CZ}\overline{\mathdutchcal{G}}_0$.
%of all matrices in $\overline{\mathcal{G}_3}$ of determinant equal to $-1$
%ker(det) connesso perchè det continua e {1} connesso, simile per det^{-1})(-1) o perchè la moltiplicazione per un elemento è un homeomorphism
%la preimmagine w.r.t homo f di un qualunque altro elemento del gruppo diverso da 1, f^{-1}(g), è un coset di f^{-1}(1)=\ker(f).... Since $\det$ is a surjective homomorphism, $\det^{-1}(-1)$ is the non-trivial coset of $\ker(\det)$.
%\beq \label{eq:G3densoU4?} \overline{\mathcal{G}_3}=(\overline{\mathcal{G}_3})_0\sqcup \widetilde{CZ}(\overline{\mathcal{G}_3})_0. \eeq 
%rifacendo lo stesso discorso ma per la chiusura di \mathcal{G}_3, troviamo \mathcal{G}_3=(\mathcal{G}_3)_0\sqcup \widetilde{CZ}(\mathcal{G}_3)_0
%The group $\U(4)\simeq \mathrm{SU(4)}\rtimes\U(1)$ is connected, but its subgroups, like $\mathcal{G}_3$ and $\overline{\mathcal{G}_3}$, are not necessarily so. 
By the closed-subgroup theorem, $\overline{\mathdutchcal{G}}$ is a Lie subgroup of $\mathrm{O}(4)$%compatto ma non connesso
, $\overline{\mathdutchcal{G}_0}\simeq \overline{\mathdutchcal{G}}_0$ is a Lie subgroup of $\SO(4)$, 
and the algebra $\mathrm{Lie}(\overline{\mathdutchcal{G}})=\mathrm{Lie}(\overline{\mathdutchcal{G}}_0)\subseteq \mathfrak{o}(4)$. Since $\SO(4)$ and $\overline{\mathdutchcal{G}}_0$ are both connected, one has 
\beq 
\mathrm{Lie}(\overline{\mathdutchcal{G}}_0)=\mathfrak{o}(4) \textup{ if and only if } \overline{\mathdutchcal{G}}_0=\SO(4).
\eeq 
Then, our goal is to prove $\mathrm{Lie}(\overline{\mathdutchcal{G}}_0)=\mathfrak{o}(4)$ by reconstructing $\mathrm{Lie}(\overline{\mathdutchcal{G}}_0)$ generator by generator.

But first, we recall some technical results. The eigenvalues of an element $O$ in $\mathrm{O}(m)$ are in $\U(1)$: if $O$ has some non-real eigenvalues, these must be pairwise conjugated phases%la parte complessa degli autovalori non può dare -1 al determinante, perchè la matrice è reale, in caso deve venire da qualche autovalore reale -1
. An element $O\in\mathrm{O}(m)$ is of infinite order if and only if it has an eigenvalue $e^{i\theta}\in \U(1)$ with $\theta\in\RR\setminus\pi\QQ$, i.e., with $\theta$ incommensurate with $\pi$. In fact, $e^{i\theta}$ is of finite order %q, cioè genera il ciclico di ordine q
if and only if $\theta\in\pi\QQ$%(i.e. $e^{i\theta}$ is a root of unity $e^{2\pi i \frac{p}{q}}$)
, otherwise $\langle e^{i\theta}\rangle %=\left\{e^{in\theta},\,n\in\ZZ\right\}
$ is dense in $\U(1)$. Moreover, $(e^{i\theta_k})_{k=1}^m\subseteq%\mathrm{T}^4\coloneqq
\U(1)^m$ generates $\{(e^{in\theta_k})_k,\,n\in\ZZ\}$ which is dense in the $d$-dimensional torus $\U(1)^d$, $1\leq d\leq m$, if and only if there are $d$ $\QQ$-linearly independent $\theta_k\not\in\pi\QQ$, if and only if $d=\mathrm{dim}_\QQ\mathrm{Span}_\QQ\{\pi,\theta_k\}_k-1\geq1$ %while is discrete if and only if $\theta_i/2\pi\in\QQ$ for all $i$ cioè \mathrm{dim}_\QQ\mathrm{Span}_\QQ\{1,\theta_i/2\pi\}_i-1=0
(by the Kronecker-Weyl theorem)%devo considerare la lineare indipendenza solo delle \theta_i/\pi non razionali, perchè solo quelle singolarmente danno densità e non discretezza in T. tutti i numeri razionali sono Q-lin dip tra loro (Q spazio vett di dim_Q 1).. Invece razionale e irrazionale sempre Q-lin indip (se non almeno uno delle due categorie nulla), quindi considerare anche le fasi razionali mi aggiunge una dimensione, che diminuisco correttamente con il -1. l'1 presente nel set mi seleziona la dimensione della lin indip delle irrazionali (che sono lin indip con 1) e assorbe la dimesione delle razionali (che sono lin dip con 1
%l'alternativa è $m=\mathrm{dim}_\QQ\mathrm{Span}_\QQ\{\theta_i/2\pi\}_i$ where at least one $\theta_i/2\pi$ is irrational.. ma questo può dare dimensione aggiuntiva rispetto all'altra formula, considerando anche le fasi razionali che danno una dimensione del toro aggiuntiva: lin indip tra razionali e irrazionali, quindi incommensurabili quando ne prendo le potenze, e la densità della irrazionale si va a mixare incommensurabilmente con la discretezza della razionale, creando una ulteriore dimensione/direzione di densità. ma non credo.. è irrealistico dire che una fase in una dimensione mi dia qualcosa di denso nel toro se in realtà le sue potenze sono discrete
.

If $R\in \mathdutchcal{G}_0$, the group $\overline{\langle R\rangle}$ is a Lie subgroup of $\overline{\mathdutchcal{G}}_0\leq\SO(4)$%perchè \overline{\langle U\rangle} nel sottogruppo chiuso \overline{G_0}=G_0, e a maggior ragione nella componente identità della chiusura.. e uso the closed-subgroup theorem
%, for which it makes sense to find the generators of the relative Lie algebra
, and $\mathrm{Lie}(\overline{\langle R\rangle})\subseteq \mathrm{Lie}(\overline{\mathdutchcal{G}}_0)\subseteq \mathfrak{o}(4)$. %esiste un cambio di base ortogonale tale che O diventa una matrice diagonale a blocchi, in cui ogni blocco 2x2 è la rotazione planare cos\theta -\sin\theta sin\theta \cos\tehta associata alla coppia di autovalori $e^{\pmi\theta}$ ed il restante sono gli autovalori diagonali +1 e -1... ovviamente, in questa "diagonalizzazione", l'lagebra di Lie associata è reale, è in so... ma a me non me frega e vado ai copmlessi e diagonalizzo proprio in una matrice complessa unitaria di fasi
There exists $\Lambda\in\U(4)$ such that $R=\Lambda D\Lambda^\dagger$ where $D\coloneqq\diag(e^{i\theta_k})_{k=1}^4$ is the diagonal matrix of eigenvalues of $R$ of unitary determinant. If $d=\mathrm{dim}_\QQ\mathrm{Span}_\QQ\{\pi,\theta_k\}_k-1$, then $\overline{\langle R\rangle} = \Lambda \overline{\langle D\rangle} \Lambda^\dagger \simeq \overline{\langle D\rangle} \simeq\U(1)^d$%is an $d$-parameter Lie group, hence it 
, with associated $d$-dimensional Lie algebra $\mathrm{Lie}(\overline{\langle R\rangle})=\Lambda\mathrm{Lie}(\overline{\langle D\rangle})\Lambda^\dagger\simeq \mathrm{Lie}(\overline{\langle D\rangle})\subset\mathfrak{su}(4)$.
One finds a set of generators $\{A_j\in\mathfrak{su}(4)\}_{j=1}^d$ of $\mathrm{Lie}(\overline{\langle D\rangle})$,
\begin{comment}
as follows: find a basis $\{2\pi%per i razionali
,\alpha_j%\in\RR\setminus2\pi\QQ per gli irrazionali
\}_{j=1}^d$ of $\mathrm{Span}_\QQ\{2\pi,\theta_k\}_{k=1}^4$; write the linear combination $(\theta_k)_{k=1}^4=\mathbf{v}_0+\sum_j\mathbf{v}_j\alpha_j$ for suitable $\mathbf{v}_0,\mathbf{v}_j=({v_j}_l)_l\in\QQ^4$ ($\mathbf{v}_0=\mathbf{0}$ if and only if $\theta_k\not\in\pi\QQ$ for all $k$); define the generator $A_j\coloneqq i \diag({v_j}_l)_{l=1}^4$ up to real multiples, for every $j=1,\dots,d$%cioè ho d generatori A_j, e lascio fuori l'eventuale (d+1)-esimo v_0 dovuto alle eventiali fasi razionali, che corrisponde a elementi di ordine finito e non genera un flusso continuo, quindi non contribuisce ai generatori dell'algebra
. Thus 
\end{comment}
so that $\mathrm{Lie}(\overline{\langle D\rangle})=\mathrm{Span}_\RR\{A_j\}_{j=1}^d$ and $\overline{\langle D\rangle}=\{e^{\sum_{j=1}^dt_jA_j},\,t_j\in\RR\}$. %mi si genera così perchè è connected as a torus e \det(D)=\det(U)=1 con tutti gli autovalori vicini a 1.. se fossi partita da una matrice con det=-1, o anche con det=1 ma perchè ce ne sono due -1, non è nella componente dell'identità, e la Lie algebra associata mi genera solo gli elementi vicini e non quelli lontani ad 1
By conjugation by $\Lambda$, one finds $\mathrm{Lie}(\overline{\langle R\rangle})$ and $\overline{\langle R\rangle}$.

We carry out this process for non-commuting matrices in $\mathdutchcal{G}_0$ of infinite order, and we will calculate the commutator of the generators of the associated Lie algebra, in order to find new linearly independent vectors, up to $\dim(\mathfrak{o}(4))=6$.

By GAP and Mathematica~\cite{ourcode}, we find the matrix
$R_1\coloneqq \left((S\otimes S)\widetilde{CZ}\right)^2 %=(346)^2
=\Lambda_1 D_1\Lambda_1^\dagger\in\mathdutchcal{G}_0$ where $\Lambda_1\coloneqq \frac{1}{\sqrt{10}}\begin{pmatrix}
0 & 0 & -i\sqrt{5} & i\sqrt{5} \\
-\sqrt{5} & \sqrt{3} & -1 & -1 \\
\sqrt{5} & \sqrt{3} & -1 & -1 \\
0 & 2 & \sqrt{3} & \sqrt{3}
\end{pmatrix}\in\U(4)$ and $D_1\coloneqq\diag(e^{i\theta_k})_k$ with $\theta_1=\theta_2=0$%the 1* one is the one coming from eigenvalue (-1) of the non-squared matrix #346
, $\theta_3=\arctan\left(\frac{\sqrt{15}}{7}\right)-\pi,\theta_4=-\theta_3\not\in\pi\QQ$. Then $\overline{\langle R_1\rangle}\simeq\U(1)$, the one-dimensional Lie algebra $\mathrm{Lie}(\overline{\langle D_1\rangle})\subseteq\mathfrak{su}(4)$ is generated by %$A=i\diag(0,0,\pi+\arctan\sqrt{15},-\pi-\arctan\sqrt{15})$ o piu minimalmente col metodo sopra
$A\coloneqq i\diag(0,0,1,-1)$%secondo l'altro metodo invece, avrei 2 generatori: a basis for $\mathrm{Span}_\QQ\{\theta_i/2\pi\}_i$ is $\{\alpha_1=1,\alpha_2=\frac{\arctan\sqrt{15}}{\pi}\}$, from which one finds $A_1=i\diag(0,1,1,1)$ and $A_2=i\diag(0,0,1,-1)$. 
, and $\mathrm{Lie}(\overline{\langle R_1\rangle})\subseteq\mathfrak{o}(4)$ is generated by
\beq \label{eq:genLieT1}
A_1\coloneqq\Lambda_1 A\Lambda_1^\dagger=
\frac{1}{\sqrt{5}}
\begin{pmatrix}
0 & -1 & -1 & \sqrt{3}\\ 1 & 0 & 0 & 0 \\ 1 & 0 & 0 & 0 \\ -\sqrt{3} & 0 & 0 & 0
\end{pmatrix}.
\eeq

Moreover, a matrix which does not commute with $R_1$ is $R_2\coloneqq \left((S\otimes \mI_2)\widetilde{CZ}(\mI_2\otimes S)\right)^2 %=(364)^2
=\Lambda_2 D_2\Lambda_2^\dagger\in\mathdutchcal{G}_0$ where $\Lambda_2\coloneqq
\frac{1}{\sqrt{5}}\begin{pmatrix}
0 & \sqrt{3} & e^{i(\arctan(\sqrt{5/3})-\pi)} & e^{-i(\arctan(\sqrt{5/3})-\pi)}\\
0 & 1 & \sqrt{2}e^{-i(\arctan\sqrt{15}-\pi)} & \sqrt{2}e^{i(\arctan\sqrt{15}-\pi)}\\
0 & 1 & \sqrt{2} & \sqrt{2}\\
\sqrt{5} & 0 & 0 & 0
\end{pmatrix}\in \U(4)$ and $D_2\coloneqq\diag(e^{i\phi_k})_k$ with $\phi_1=\phi_2=0$, 
$\phi_4=-\phi_3 =\theta_3\not\in\pi\QQ$. Again $\overline{\langle U_2\rangle}\simeq\U(1)$, the one-dimensional Lie algebra $\mathrm{Lie}(\overline{\langle D_2\rangle})$ is generated by $A=i\diag(0,0,1,-1)$, and $\mathrm{Lie}(\overline{\langle R_2\rangle})$ is generated by
\beq 
A_2\coloneqq \Lambda_2 A\Lambda_2^\dagger= 
\frac{1}{\sqrt{5}}\begin{pmatrix}
0 & -1 & 1 & 0 \\
1 & 0 & -\sqrt{3} & 0 \\
-1 & \sqrt{3} & 0 & 0 \\
0 & 0 & 0 & 0
\end{pmatrix}.
\eeq

Now, $A_1,\,A_2\in \mathrm{Lie}(\overline{\mathdutchcal{G}}_0)$ are elements of an $\RR$-basis of $\mathrm{Lie}(\overline{\mathdutchcal{G}}_0)$, and we seek for the other generators by taking their Lie bracket (i.e. commutator). We find that $A_1,\,A_2,\,[A_1,A_2]$ are $\RR$-linearly independent, hence $[A_1,A_2]$ is a third basis vector of $\mathrm{Lie}(\overline{\mathdutchcal{G}}_0)$. Similarly, we find that $\{A_1,\,A_2,\,[A_1,A_2],\,[A_1,[A_1,A_2]],\,[A_2,[A_1,A_2]],\,[A_1,[A_1,[A_1,A_2]]]$ is a set of $6$ $\RR$-linearly independent matrices of $\mathrm{Lie}(\overline{\mathdutchcal{G}}_0)$. This shows that $\mathrm{Lie}(\overline{\mathdutchcal{G}}_0)=\mathfrak{so}(4)$, and $\mathdutchcal{G}_0$ is dense in $\SO(4)$ (as well as $\overline{\mathdutchcal{G}}=\mathrm{O}(4)$).

\medskip

Now, one can show (with usual Lie algebra techniques) that $\langle\SO(4)\otimes\mI_2,\mI_2\otimes\SO(4)\rangle=\SO(8)$. %Up to a phase, this allows us to approximate the Toffoli gate , and thus all the classical gates. In particular, we are allowed to permute the basis elements of the working space. ...Ma ho gia' swappato $\SO(4)$ nei due suddetti embedding...
Then, by the real encoding of quantum computation (cf. Remark~\ref{eq:remBVrelaQC} and below), the set of gates $\mathcal{G}_1\mathdutchcal{p}_3$ approximates any circuit or gate in $\U(4)$. Finally, $\U(4)$ is universal by~\cite{DiVincenzo},~\cite{Deutsch}. Nevertheless, we can argue the universality of $\mathcal{G}_1\mathdutchcal{p}_3$ with a stronger argument as follows. As established in classical literature (e.g.~\cite[Sec.~5.1.8]{GolubGivens}), we show how $\SO(n)$ in dimension $m$ generates the whole group $\SO(m)$ for every $2\leq n<m$.

\begin{definition}
For $m\geq2$, $1\leq j<k\leq m$, and $\theta\in\RR$, a \emph{Givens rotation} $G_{jk}(\theta)$ in $\SO(m)$ is the following rank-$2$ correction of the identity,
\begin{align}
G_{jk}(\theta)\coloneqq&\begin{pmatrix} 1 & \cdots & 0 & \cdots & 0 & \cdots & 0 \\
\vdots & \ddots & \vdots &  & \vdots &  & \vdots\\
0 & \cdots & \cos\theta & \cdots & -\sin\theta & \cdots & 0 \\
\vdots &  & \vdots &\ddots  & \vdots &  & \vdots\\
0 & \cdots & \sin\theta & \cdots & \cos\theta & \cdots & 0 \\
\vdots &  & \vdots &  & \vdots & \ddots & \vdots\\
0 & \cdots & 0 & \cdots & 0 & \cdots & 1
\end{pmatrix}
\begin{matrix} \phantom{1} \\
\phantom{\vdots} \\
j \\
\phantom{\vdots} \\
k  \\
\phantom{\vdots} \\
\phantom{0}
\end{matrix}\\
&\hspace{1.85cm} j \hspace{1.65cm} k\nonumber
\end{align}
i.e., a rotation of angle $\theta$ in $\SO(2)$ embedded in $\SO(m)$ to act non-trivially only in the $(j,k)$ coordinate plane in $\RR^m$.
\end{definition}
In what follows, an embedding of $\SO(n)$ in $\SO(m)$, $n\leq m$, is always intended to be a canonical embedding as in the above definition.

Givens rotations are useful for computations, as they allow to annihilate all but one component of a vector (or of a column of a matrix). If $\mathbf{x}=(x_i)_i\in\RR^m$ and $\mathbf{y}=(y_i)_i\coloneqq G_{jk}(\theta)\mathbf{x}$, then $y_j=\cos(\theta) x_j-\sin(\theta) x_k$, $y_k=\sin(\theta)x_j + \cos(\theta)x_k$, and $y_l=x_l$ for all $l\neq j,k$. If not already $x_k=0$, then we can force $y_k=0$ by setting 
\beq \label{eq:annprimapenultimacomp}
\cos\theta=\frac{x_j}{\sqrt{x_j^2+x_k^2}},\qquad \sin\theta=-\frac{x_k}{\sqrt{x_j^2+x_k^2}}.
\eeq 
Similarly, if $x_j\neq 0$, we can make $y_j=0$ by setting %\label{eq:annsecondaultimacomp}
$\cos\theta=\frac{x_k}{\sqrt{x_j^2+x_k^2}}$, $\sin\theta=\frac{x_j}{\sqrt{x_j^2+x_k^2}}$. 

\begin{remark}\label{remark:GivensJacobi}
Givens rotations provide a \emph{generalisation of Euler- or nautical-angle decomposition} of $\SO(3)$ to any higher-dimensional $\SO(m)$. Indeed, we show how, for every $m\geq 2$, the group $\SO(m)$ is generated by products of rotations in $\SO(2)$ (or in $\SO(n)$ with $2\leq n< m$) embedded in $\SO(m)$.

Given $O\in\mathrm{O}(m)$, we want to annihilate its first column except for one element, by using Givens rotations. Since the first column of the orthogonal matrix $O$ is non-zero (e.g. $O_{1m}\neq 0$%primo, oppure l'ultimo
), we can use Eq.~\eqref{eq:annprimapenultimacomp} repeatedly%oppure \eqref{eq:annsecondaultimacomp}
. There exist $\theta_2,\dots,\theta_m\in\RR$ such that
\beq \label{eq:Givensdeco}
O'\coloneqq G_{1m}(\theta_m)\dots G_{12}(\theta_2)O = \begin{pmatrix} O'_{11} &\mathbf{0}^\top\\\mathbf{0}& O'_{m-1} \end{pmatrix}.
\eeq 
One has $O'\in\mathrm{O}(m)$%Since $G_{1m}(\theta_m),\dots ,G_{12}(\theta_2)\in\SO(m)$, $O\in\mathrm{O}(m)$
, thus $O'_{11}\in\{\pm1\}$ and $O'_{m-1}\in\mathrm{O}(m-1)$. We can write $O'$ in terms of special orthogonal matrices only: either already $O'_{m-1}\in\SO(m-1)$, or we modify last Givens rotation as $G_{1m}(\theta_m)\mapsto G_{1m}(\theta_m+\pi)$. This changes the sign of $O'_{11}$ and of last column of $O'_{m-1}$, whose determinant then changes sign. This shows that every $O\in\mathrm{O}(m)$ can be decomposed into the product of (at most) $m-1$ Givens rotations and an orthogonal transformation $O'$; the Givens rotations are elements of $\SO(2)$ suitably embedded in $\SO(m)$, and $O'$ is an element of $\SO(m-1)$ suitably embedded in $\mathrm{O}(m)$. Moreover $\det(O')=O'_{11}\det(O'_{m-1})$ in Eq.~\eqref{eq:Givensdeco}. Then, every $O\in\SO(m)$ can be written in terms of lower-dimensional rotations as just described with $O'\in\SO(m)$ (that is $O'_{m-1}\in\SO(m-1)$ and $O'_{11}=+1$ %rotation axis fixed by O'_{m-1}
in Eq.~\eqref{eq:Givensdeco}).

Every element of $\SO(3)$ can be written as the product of three rotations around the reference axes of $\RR^3$ according to the nautical- (aka Cardano-) or Euler-angle decomposition~\cite{our1st}, which is what the above method describes for $m=3$. %Every element in $\SO(4)$ decomposes into the product of $3$ Givens rotations and a rotation from $\SO(3)$, i.e. into the product of $6$ Givens rotations. 
It is easily checked by induction that, for every $m\geq2$, every $R\in\SO(m)$ is decomposable into %m-1 Givens e una (m-1)dimensionale, quest'ultima con m-1 Givens e una (m-3)dimensionale, quest'ultima......... 4dimensionale, quest'ultima con 3 Givens e una di SO(3), quest'ultima con 2 Givens e una di SO(2) (ossia con Cardano/Euler, con 3 Givens)
the product of (at most) $\sum_{n=1}^{m-1}n=\frac{m(m-1)}{2}=\dim(\mathfrak{o}(m))$ Givens rotations, i.e. ${m\choose 2}$ elements of $\SO(2)$ in $\SO(m)$. In other words, the group $\SO(2)$ (acting on the ${m \choose 2}$ coordinate planes) generates the group $\SO(m)$.
From this it also follows that for every $2\leq n\leq m$, the group $\SO(n)$, acting on the ${m\choose n}$ coordinate subspaces, generates the group $\SO(m)$.
\end{remark}

We have seen that our set of gates $\mathcal{G}_1\mathdutchcal{p}_3$ generates a dense group in $\mathrm{O}(4)$, through which to generate $\mathrm{O}(8)$. %se dico di farlo come O(4)\otimes I, I\otimes O(4) sto usando SWAP (che ho detto che posso farlo dall'inizio... se dico di usare Givens da 4 a 5 a 6 a 7 a 8 sto invece permutando i vettori di base per fare i diversi embedding (che quindi non è ancora concesso)
Hence, we are able to approximate the Toffoli gate%det-1
, and thus all the classical gates. In particular, we are allowed to permute the basis elements of the working space (which is different from, in fact more general than the \textup{SWAP} permuting two-dimensional qubit spaces in a tensor product). Therefore, we can realise Givens %eleemnti di SO(2) già li avevo, ora li posso embeddare su ogni permutazione di assi a caso, cioè farli agire su ogni asse
rotations and generate $\SO(2^{n+1})$ for every $n\in\NN$ by Remark~\ref{remark:GivensJacobi}. In other words, compositions of elements from  $\mathcal{G}_1\mathdutchcal{p}_3$ approximate all of  $\SO(2^{n+1})$ for every $n\in\NN$. Then, by the real encoding of quantum computation (cf. the discussion immediately following Remark~\ref{eq:remBVrelaQC}), we can realise every circuit or gate in $\U(2^n)$ for every $n\in\NN$.

All of this proves the following result.
\begin{theorem}
The set of $3$-adically controlled gates $\mathcal{G}_1\mathdutchcal{p}_3$ in Eq.~\eqref{eq:1gatesU4wrt1} is universal %, as it approximate any (special) orthogonal transformation, and hence can be served for a real encoding of universal 
for quantum computation.
\end{theorem}

%è fico che loro partono sempre dal Toffoli, perchè Toffoli universale per computazione classica (3 bit necessari per computazione reversibile), e vedono cosa devono aggiungere per universalità della computazione propriamente quantistica. Ma noi arriviamo ad un set universale senza Toffoli e soprattutto sempre porte a 1 o 2 qubit... anche se ce ne sono diversi così, e con meno porte anche

\section{Conclusions}
\label{sec:conclusions}
Advancing the study of $p$-adic quantum mechanical spin and angular momentum, we have here shown that low-dimensional irreducible representations of the three-dimensional $p$-adic rotation group $\SO(3)_p$ indeed give rise to meaningful elements of quantum information processing: $p$-adic qubits and universal sets of $p$-adically controlled quantum logic gates. %While the main objective is a $p$-adic formulation of quantum theories, this work has strong mathematical roots on the unitary irreducible representations of $\SO(3)_p$. OR While this work has strong mathematical roots, its main objective is to advance a step forward in a $p$-adic formulation of quantum mechanics and quantum computation \'a la Volovich.

We established that the finite-dimensional representations of $\SO(3)_p$ factorise through finite quotients $\SO(3)_p\mod p^k$. By focusing on the level $k=1$, we classified the unitary irreducible representations arising from $\SO(3)_p \bmod p$, identifying a $p$-adic qubit with such a representation of dimension two. A striking feature emerging from this classification, unlike the standard real case, is that $\SO(3)_p$ admits more than one $p$-adic qubit representation for primes $p>3$. This abundance is reflected in several possibilities of composing a multipartite system of $p$-adic qubits, through the tensor product of their representations. An open question is how many non-equivalent $p$-adic qubit representations there are for every prime $p$, and what they signify. The existence of non-equivalent one-dimensional representations of $\SO(3)_p$ might suggest to reduce $\SO(3)_p$ to a suitable subgroup %(e.g. the one corresponding to $p$-adic quaternions of reduced norm $1$)
with unique trivial character (as it happens in the real case from $\mathrm{O}(3)$ to $\SO(3)$), which could thereby provide a selection of physically meaningful $p$-adic qubit representations.

We addressed the Clebsch-Gordan problem for two $p$-adic qubits from $\SO(3)_p \bmod p$, decomposing in singlet and doublet states, driven by the representation theory of the dihedral group $\mathrm{D}_{p+1}$. Although the tensor product of two two-dimensional representations can theoretically decompose into a singlet and a triplet, this structure has not yet been observed in a system of two $p$-adic qubits factorizing modulo $p$. %Then, we identified maximally entangled Bell bases realising those decompositions: the singlet quantum states are the only maximally entangled ones, whereas the projectors onto two- (and three-) dimensional subspaces are separable. 
Then, we identified entangled states realising those decompositions: every singlet, doublet (and triplet) can be given by maximally entangled Bell states. Note however that the latter two cases present separable states, which are equivalently spanned by product vectors. 
%\textcolor{red}{This means that entangled vector states arises more often in the $p$-adic setting,} hence it would be interesting to quantify this aspect through the volume of entangled states in the whole set of states and contrast it with the standard case. 
%On the other hand, in terms of quantum states, the singlet quantum states are the only maximally entangled ones, whereas the projectors onto two- (and three-) dimensional subspaces are separable. 
We leave it open to analyse the irreducible representations of $\SO(3)_p\mod p^k$ for $k>1$, and whether this eventually yields three-dimensional representations, potentially allowing the Clebsch-Gordan decomposition of two $p$-adic qubits to recover the singlet and triplet structure familiar from standard quantum mechanics. A natural step beyond is to extend this analysis to systems composed of more than two $p$-adic qubits.

Lastly, quantum logic gates form the foundation of the circuit model of quantum computation, by manipulating qubit states via unitary operators. Here, we have proposed to take elements from the $2^n$-dimensional unitary representations of $\SO(3)_p$ as logic gates on $n$ qubits, drawing inspiration from Kitaev's topological anyon model with the actions of a braid group~\cite{KitaevBraid}. Our focal point is to identify some entangling logic gate on two qubits, that can be used to construct a universal set of gates. In this context, the prime $p=3$ appeared privileged: at the level of $\SO(3)_p\mod p$, $p=3$ is the only case exhibiting a unique $p$-adic qubit representation and, crucially, $4$-dimensional irreducible representations. We proved that the elements of such representation of $\SO(3)_3$ lead to a universal set of $3$-adically controlled quantum logic gates, ensuring the capacity of our model for arbitrary quantum computations, and paving the way of the research over $p$-adically driven quantum algorithms. Finally, we would like to understand if the singled out universal set of logic gates is of interest for practical application.

So far, the study of representations of $\SO(3)_p$ and their application to quantum physics has been restricted to the level of $\SO(3)_p\mod p$. It remains to be understood if the prominence of $p=3$ is just an artifact of this restriction. Moreover, the $4$-dimensional irreducible representations of $\SO(3)_3\mod 3$ are actually orthogonal, and an ancillary qubit needs to be added to encode unitary circuits through our universal set of orthogonal gates. We still need to see if genuinely complex unitary representations emerge at higher levels $\SO(3)_p \bmod p^k$, and if $4$-dimensional ones exist for other primes $p$, potentially offering universal sets of gates for all primes. To fully answer this, as well as to have a complete $p$-adic theory of angular momentum, spin and qubit, future work has to pursue a complete classification of unitary irreducible representations of $\SO(3)_p$. This will be addressed via Kirillov's orbit method and Howe's generalised theory, before extending to projective representations with the help of the Haar measure on $\SO(3)_p$.

Looking farther afield, we plan to provide a definition of qubit, entanglement, logic gates, etc. in the adelic framework, so to encompasses both standard and $p$-adic quantum mechanics for all primes $p$ simultaneously.

%%%%%%%%%%%%%%%%%%%%%%%%%%%%%%%%%%%%%%%%
\acknowledgments
The authors would like to thank Nicola Ciccoli, Jessica Fintzen, Massimo Giulietti, Decimus Phostle, Thomas S. Weigel and Evgeny I. Zelenov, for fruitful discussions and their valuable insights during the preparation of the present manuscript.

IS was supported by the Istituto Nazionale di Fisica Nucleare (INFN), Sezione di Perugia.
SLI acknowledges support from PRIN22 Models, Sets and Classifications, and from GNSAGA.
SM acknowledges financial support from COST Action CA23115: Relativistic Quantum 
Information, funded by COST (European Cooperation in Science and Technology) and from PNRR Italian Ministry of University and Research project PE0000023-NQSTI. 
AW is or was supported by the European Commission QuantERA project ExTRaQT (Spanish MICIN grant no.~PCI2022-132965); by the Spanish MICIN (project PID2022-141283NB-I00) with the support of FEDER funds; by the Spanish MICIN with funding from European Union NextGenerationEU (PRTR-C17.I1) and the Generalitat de Catalunya; by the Spanish MTDFP through the QUANTUM ENIA project: Quantum Spain, funded by the European 
Union NextGenerationEU within the framework of the ``Digital Spain 2026 Agenda''; and by the Alexander von Humboldt Foundation.

%%%%%%%%%%%%%%%%%%%%%%%%%%%%%%%%%%%%%%%%%%%%%%%%%%%%%%%%%%%%%%%%%%%%%%%%%%%%

\bibliographystyle{unsrt}

%%%%%%%%%%%%%%%%%%%%%%%%%%%%%%%%%%%%%%%%%%%%%%%%%%

\end{document}